\documentclass[oneside,english,reqno]{amsart}
\usepackage[T1]{fontenc}
\usepackage[utf8]{luainputenc}
\usepackage{geometry}
\usepackage{amsthm}
\usepackage{amssymb}
\usepackage{graphicx}
\usepackage{multirow}
\usepackage{xcolor}

\makeatletter

\usepackage[foot]{amsaddr}

\makeatletter
\renewcommand{\email}[2][]{%
  \ifx\emails\@empty\relax\else{\g@addto@macro\emails{,\space}}\fi%
  \@ifnotempty{#1}{\g@addto@macro\emails{\textrm{(#1)}\space}}%
  \g@addto@macro\emails{#2}%
}
\makeatother

\numberwithin{equation}{section}
\theoremstyle{plain}
\newtheorem{thm}{\protect\theoremname}[section]
  \theoremstyle{plain}
  \newtheorem{prop}[thm]{\protect\propositionname}
  \theoremstyle{remark}
  \newtheorem{rem}[thm]{\protect\remarkname}
  \theoremstyle{plain}
  \newtheorem{lem}[thm]{\protect\lemmaname}
  \theoremstyle{plain}
  \newtheorem{definition}[thm]{\protect\definitionname}
  \newtheorem{assumption}[thm]{\protect\assumptionname}
  \theoremstyle{plain}
  \newtheorem{corollary}[thm]{Corollary}

\makeatother

\usepackage{babel}
  \providecommand{\lemmaname}{Lemma}
  \providecommand{\propositionname}{Proposition}
  \providecommand{\remarkname}{Remark}
\providecommand{\theoremname}{Theorem}
\providecommand{\definitionname}{Definition}
\providecommand{\assumptionname}{Assumption}

\newcommand{\be}{\begin{equation}}
\newcommand{\ee}{\end{equation}}

\newcommand{\bea}{\begin{eqnarray}}
\newcommand{\eea}{\end{eqnarray}}
\newcommand{\beaa}{\begin{eqnarray*}}
\newcommand{\eeaa}{\end{eqnarray*}}
\newcommand{\bes}{\begin{subequations}}
\newcommand{\ees}{\end{subequations}}

\newcommand{\EE}{{\mathbb E}}
\newcommand{\R}{{\mathbb R}}
\newcommand{\RR}{{\mathbb R}}

\newcommand{\NN}{{\mathbb N}}

\newcommand{\cf}{{\mathcal F}}

\newcommand{\cm}{{\mathcal M}}

\newcommand{\cB}{\mathcal{B}}

\newcommand{\cF}{\mathcal{F}}

\newcommand{\cM}{\mathcal{M}}

\newcommand{\cP}{\mathcal{P}}

\newcommand{\cU}{\mathcal{U}}

\newcommand{\fP}{\mathfrak{P}}

\newcommand{\eps}{\varepsilon}

\DeclareMathOperator{\Law}{Law}

\title{
Bounds for VIX futures given S\&P~500 smiles}%
\thanks{
\indent  Date: \today. }

\author{Julien Guyon}
\author{Romain Menegaux}
\address{Bloomberg L.P., Quantitative Research}
\email{jguyon2@bloomberg.net, rmenegaux@bloomberg.net}

\author{Marcel Nutz}
\address{Columbia University, Departments of Statistics and Mathematics}
\email{mnutz@columbia.edu}

\begin{document}

\begin{abstract}
We derive sharp bounds for the prices of VIX futures using the full information of S\&P~500 smiles. To that end, we formulate the model-free sub/superreplication of the VIX by trading in the S\&P~500 and its vanilla options as well as the forward-starting log-contracts. A dual problem of minimizing/maximizing certain risk-neutral expectations is introduced and shown to yield the same value.

The classical bounds for  VIX futures given the smiles only use a calendar spread of log-contracts on the S\&P~500. We analyze for which smiles the classical bounds are sharp and how they can be improved when they are not. In particular, we introduce a family of functionally generated portfolios which often improves the classical bounds while still being tractable; more precisely, determined by a single concave/convex function on the line. Numerical experiments on market data and SABR smiles show that the classical lower bound can be improved dramatically, whereas the upper bound is often close to optimal.
\end{abstract}

\maketitle

%%%%%%%%%%%%%%%%%%%%%%%%%
\section{Introduction}

In this article, we derive sharp bounds for the prices of VIX futures by using the full information of S\&P~500 smiles at two maturities. The VIX (short for volatility index) is published by the Chicago Board Options Exchange (CBOE) and used as an indicator of short-term options-implied volatility. By definition, the VIX is the implied volatility of the 30-day variance swap on the S\&P~500; see~\cite{CBOE}. Equivalently, using the well-known link between realized variance and log-contracts~\cite{neuberger}, the VIX at date $T_1$ is the implied volatility of a log-contract that delivers $\ln(S_{T_2}/S_{T_1})$ at $T_{2}=T_{1}+ \tau$, where $\tau=30$~days and $S_{T_{i}}$ is the S\&P~500 at date~$T_{i}$:
\beaa
(\mathrm{VIX}_{T_1})^{2} = -\frac{2}{\tau}\mathrm{Price}_{T_1}\left[ \ln\left(\frac{S_{T_2}}{S_{T_1}} \right)\right];% = -\frac{2}{\tau}\left(\mathrm{Price}_{T_1}\left[ \ln S_{T_2}\right] - \ln S_{T_1}\right);
\eeaa
we are assuming zero interest rates, repos, and dividends for simplicity. The log-contract can itself be replicated at $T_1$ using call and put options on the S\&P~500 with maturity $T_2$.
The VIX index cannot be traded, but VIX futures can: the VIX future expiring at $T_1$ is an instrument that pays $\mathrm{VIX}_{T_1}$ at $T_1$. While $\mathrm{VIX}_{T_1}^{2}$ can be replicated, its square root $\mathrm{VIX}_{T_1}$ cannot; instead, sub/superreplication in the S\&P~500 and its options leads to model-free lower/upper bounds on the price of the VIX future.

The classical sub/superreplication argument is based on the fact that one can replicate any affine function of $\mathrm{VIX}_{T_1}^{2}$ at $T_1$ using cash and log-contracts with maturities $T_1$ and $T_2$. Thus, one searches for the sub/superreplication of the square root function by an affine function that gives the maximum/minimum portfolio price. Since the square root is a concave function, it is below all its tangent lines, and the classical superreplication boils down to selecting the line that gives the minimum portfolio price. This argument shows that, in the absence of
arbitrage, the price of the VIX future at time $T_0=0$ cannot exceed the implied volatility
$\sigma_{12}$ of the forward-starting log-contract on the index,
starting at the VIX future's expiry $T_{1}$ and maturing at~$T_{2}$,
\beaa
\sigma_{12}^2 \equiv -\frac{2}{\tau}\mathrm{Price}_{T_0}\left[ \ln\left(\frac{S_{T_2}}{S_{T_1}} \right)\right].% = -\frac{2}{\tau}\left(\mathrm{Price}_{T_0}\left[ \ln S_{T_2}\right] - \mathrm{Price}_{T_0}\left[ \ln S_{T_1}\right]\right).
\eeaa
Subreplicating the VIX future using the same instruments corresponds to subreplicating the square root by an affine function. This yields zero as a lower bound for the future's price, which is clearly a poor estimate.

These classical bounds are suboptimal in the sense that they only use the prices of log-contracts. Our aim is, instead, to extract the full information contained in the S\&P~500 smiles at $T_{1}$ and $T_{2}$, by also including all vanilla options at these maturities as (static) hedging instruments, as well as trading (dynamically, i.e., at~$T_{1}$) in the S\&P~500 itself and the forward-starting log contract. Moreover, we allow the deltas at $T_{1}$ to depend on the information available, that is, the S\&P~500 and the VIX index at $T_{1}$.

\medskip

The first part of the paper analyzes this problem for general smiles. We formulate the sub/superreplication as a linear programming (LP) problem and define absence of arbitrage in this setting. The latter leads to the existence of risk-neutral joint distributions $\mu$ for $(S_{T_{1}},S_{T_{2}},\mathrm{VIX}_{T_1})$  which constitute the domain of an optimization problem dual to sub/superreplication (Theorem~\ref{thm:arbitrage-free}). The first two marginals $\mu_{1}$ and $\mu_{2}$ are given by the market smiles at $T_{1}$ and $T_{2}$, whereas the distribution of $\mathrm{VIX}_{T_1}$ merely satisfies a certain constraint. The dual problem is thus reminiscent of a (constrained) martingale optimal transport problem, but falls outside the transport framework because the third marginal is not prescribed. This necessitates a novel argument for our duality theorem which establishes the absence of a duality gap (Theorem~\ref{th:strongDuality}), i.e., primal and dual problem have the same value. This theorem holds, more generally, for an option payoff $f(S_{T_1},S_{T_2},\mathrm{VIX}_{T_1})$ rather than just the VIX. As a last abstract contribution, we characterize those smiles $\mu_{1},\mu_{2}$ for which the classical bounds for the VIX future are optimal (Theorems~\ref{thm:D=0} and~\ref{thm:D=Dbar}). The lower bound is optimal if and only if $\mu_{1}=\mu_{2}$, which never happens in practice. The characterization for the upper bound is more subtle, it states that a convex-order condition in two dimensions holds, or equivalently that a model with constant forward volatility is contained in the dual domain.

While our theoretical bounds are sharper than the classical ones, the corresponding hedging portfolios can only be found numerically, and the numerical problem is far from trivial. Aiming for a balance between flexibility and tractability, we introduce a family of functionally generated portfolios that are determined by a one-dimensional convex/concave function and a constant (Definition~\ref{de:functionallyGen}). The space of one-dimensional convex functions is easy to search numerically, and the generated portfolios are guaranteed to satisfy the sub/superhedging conditions at all values of the underlying, by our construction. We show that the lower price bound obtained by functionally generated portfolios improves the classical one as soon as $\mu_{1}\neq\mu_{2}$ (Proposition~\ref{prop:P>0}) and here the generating function can be chosen explicitly of an inverse ``hockey stick'' form.

In the second part of the paper, we study specific families of smiles $\mu_{1},\mu_{2}$ and corresponding portfolios. The case where $\mu_{2}$ is a Bernoulli distribution gives rise to a ``complete market'' where the VIX future can be replicated. While the classical upper bound is not sharp unless $\mu_{1}$ has a very particular form, we show how functionally generated portfolios lead to the sharp bound as given by the unique risk-neutral expectation (Section \ref{sec:binomial}). When $\mu_{2}$ is a general distribution with compact support, we present various sufficient conditions for the classical upper bound to be suboptimal (Section~\ref{sec:mu2_compact_support}). Finally, we discuss a family of examples for which the classical upper bound is already sharp (Section~\ref{sec:examples_where_classical_is_optimal}).

The third part of the paper presents numerical experiments using smiles from market data as well as smiles generated by a SABR model. We compare the classical bounds, the bounds obtained from functionally generated portfolios, and the bounds computed by an LP solver that correspond to the theoretical, optimal bounds modulo discretization error. For the generating functions, we use piecewise linear maps and a cut square root; the latter yields the best approximation in our experiments. The results suggest that the classical lower bound can be improved dramatically by functionally generated portfolios; the bound from the LP solver is only slightly better. On the other hand, the classical upper bound is already surprisingly sharp for typical smiles.

\medskip

Turning to the existing literature on volatility derivatives, the most closely related work is due to De~Marco and Henry-Labord\`ere~\cite{phl-demarco} who investigate bounds for VIX options, i.e., calls and puts on the VIX, given the smile of the S\&P~500 and the VIX future as liquidly tradable instruments. Thus, compared to~\cite{phl-demarco}, we take a step back by investigating bounds for the VIX future itself, given the smile of the S\&P~500. The sub/superreplication problem in~\cite{phl-demarco} leads to a linear program with a dual akin to (constrained) martingale optimal transport. The numerical results show that, for typical market smiles, the optimal upper bound on VIX options is equal to an analytical (a priori suboptimal) bound that the authors derive.  For a further discussion of numerical solutions to sub/superreplication problems, we also refer to~\cite{HenryLabordere.11}, and to~\cite{beiglbock, galichon} for background on martingale optimal transport. While~\cite{phl-demarco} and the present paper consider derivatives on options-implied volatility,  previous literature has studied derivatives on realized volatility. Using power payoffs, Carr and Lee~\cite{carr-lee} show that, if the returns and the volatility of an asset are driven by independent Brownian motions, the asset smile at a given maturity~$T$ determines the distribution of the realized variance at $T$, hence allowing perfect replication of derivatives on realized variance. Using business-time hedging, Dupire~\cite{dupire} derives a lower bound for a call on realized variance at a given maturity $T$, given the asset smile at $T$. Carr and Lee \cite{carr-lee2} extend Dupire's idea to tackle the cases of puts on realized variance as well as forward-starting calls and puts on realized variance. More recently, Cox and Wang~\cite{cox-wang} have derived the optimal portfolio subreplicating convex functions of realized variance. While these works assume that the underlying has continuous paths, Kahal\'e~\cite{Kahale.16} considers the model-free hedging of discretely monitored variance swaps based on squared logarithmic returns in a setting where the underlying may have jumps, and thus perfect hedging is not possible despite dynamic trading in the underlying. The optimal subhedging price and a corresponding hedge are obtained. This analysis has been extended by Hobson and Klimmek~\cite{HobsonKlimmek.12} who consider both sub- and superhedging; moreover, they study a variety of variance swap contracts and dynamic trading can be constrained to discrete rebalancing dates in their setting. The authors link this problem to Perkins' solution of the Skorokhod embedding problem and derive optimal price bounds and hedging strategies.

\medskip

The remainder of the article is structured as follows. Section~\ref{sec:primal} describes the primal problems of sub/super\-replication and recalls the classical bounds on VIX futures, while absence of arbitrage and existence of risk-neutral measures are characterized in Section~\ref{sec:absence_of_arbitrage}. In Section~\ref{sec:duality} we formulate the dual problem of maximizing over risk-neutral expectations and prove the absence of a duality gap. Next, we characterize in Section~\ref{sec:characterization_classical_is_optimal} the market smiles for which the classical bounds on VIX futures are already sharp.
In Section~\ref{sec:analytical_portfolios} we introduce the functionally generated portfolios and show that they essentially always improve the classical lower bound. The subsequent sections study examples where the classical upper bound is not optimal, when the smile $\mu_{2}$ is a two-point distribution (Section \ref{sec:binomial}) or more generally a distribution with compact support (Section~\ref{sec:mu2_compact_support}). On the other hand, Section~\ref{sec:examples_where_classical_is_optimal} provides smiles for which the classical upper bound is optimal. Finally, numerical experiments using SABR and market smiles are presented in Section \ref{sec:num}.

%%%%%%%%%%%%%%%%%%%%%%%%%%%%%%%%%%%%%%%%%%%%%%%%%%%%%%%%%%%%%%
\section{Primal problem and classical bounds}\label{sec:primal}

\subsection{Setting and notation}
For simplicity, we assume zero interest rates, repos, and dividends. Moreover, we take as given the full market smiles of the S\&P~500 index $S$ at two maturities $T_{1}$ and $T_{2}\equiv T_{1}+30$
days, that is, the continuum of all call prices $C(K)$ for strikes $K\ge 0$. For each maturity $T_i$, $i\in\{1,2\}$, absence of static arbitrage (or butterfly arbitrage) is equivalent to the existence of a risk-neutral measure $\mu_i\equiv \partial^2 C_i/\partial K^2$ such that the price of any vanilla option $u_i$ written on $S_i\equiv S_{T_i}$ is the expectation $\EE^{i}[u_i(S_i)]\equiv \EE^{\mu_i}[u_i(S_i)]$ of the payoff under $\mu_i$, and we shall refer to $\mu_{i}$ as a smile as well. (The notation $\partial^2 C_i/\partial K^2$ refers to the second derivative measure of the convex function $C_{i}$.)
 In particular, the price of $S_i$ at time~0 must be the initial value $S_0\in \RR_+^*\equiv(0,\infty)$ of the S\&P~500. Therefore, throughout the article, $(S_{1},S_{2})$ denotes the identity on~$(\RR_+^*)^{2}$ and 
$\mu_1$, $\mu_2$ are probability measures on $\RR_+^*$ with mean $S_0$. 
Absence of dynamic (calendar) arbitrages will be discussed in Section \ref{sec:absence_of_arbitrage}.

We call \emph{forward-starting log-contract} (\emph{FSLC} for short) the financial derivative that pays $-\frac{2}{\tau}\ln\frac{S_{2}}{S_{1}}$
at $T_{2}$, where $\tau\equiv T_{2}-T_{1}=30$ days. We recall that, by definition of the VIX (substituting the strip of out-the-money
options by the log-contract for simplicity), the price at $T_{1}$
of the FSLC is $\mathrm{VIX}^2_{T_1}$, the square of the VIX at~$T_{1}$. For the log-contracts to have finite prices, the following is in force throughout the paper.

\begin{assumption}\label{as:integrability}
  The given marginals $\mu_{i}$ satisfy
  $$
    \EE^{i}[S_{i}]=S_0,\qquad \EE^{i}[|\ln(S_{i})|]<\infty,\qquad i\in\{1,2\}.
  $$
\end{assumption}

For convenience, we set
\beaa
L(x)\equiv-\frac{2}{\tau}\ln(x), \qquad
\ell_1 \equiv \mathbb{E}^{1}[L(S_{1})], \qquad \ell_2 \equiv \mathbb{E}^{2}[L(S_{2})]. \label{eq:L}
\eeaa
Moreover, we denote by
\bea
\sigma_{12}^{2}\equiv -\frac{2}{\tau}\left(\mathbb{E}^{2}\left[\ln S_{2}\right]-\mathbb{E}^{1}\left[\ln S_{1}\right]\right)=\mathbb{E}^{2}\left[L(S_{2})\right]-\mathbb{E}^{1}\left[L(S_{1})\right] = \ell_2 - \ell_1 \label{eq:sigma12}
\eea
the price at time 0 of the FSLC. As we will recall in Section \ref{sec:absence_of_arbitrage}, absence
of dynamic arbitrage implies that $\mu_{1}$ and $\mu_{2}$ are in convex
order. Then, $\sigma_{12}^{2}\ge0$, and $\sigma_{12}^{2}=0$ if
and only if $\mu_{1}=\mu_{2}$, as $L$ is strictly convex.

\subsection{The primal problem}

We consider a market with two trading dates ($T_0=0$ and $T_1$) where the financial instruments are the S\&P~500 (tradable at $T_0$ and $T_1$), the vanilla options on it with maturities $T_1$ and $T_2$ (tradable at $T_0$), and the FSLC (tradable at $T_1$). Note that we consider only static positions in vanilla options, but we allow dynamic trading, that is, trading at $T_1$, in the S\&P~500 and the FSLC. We are interested in deriving the optimal lower and upper bounds on the price of the VIX future expiring at~$T_1$, given these instruments.
%i.e., on the price of a financial derivative that pays a nonlinear function (the square root) of the price at $T_1$ of the FSLC. Therefore, a
Similarly as in De Marco and Henry-Labord\`ere \cite{phl-demarco}, the model-independent
no-arbitrage upper bound for the VIX future is the smallest price at time 0 of a superreplicating portfolio,
\begin{equation}
P_\mathrm{super}\equiv\inf_{\cU_\mathrm{super}}\left\{\mathbb{E}^{1}[u_{1}(S_{1})]+\mathbb{E}^{2}[u_{2}(S_{2})]\right\}\label{eq:primal}
\end{equation}
where $\cU_\mathrm{super}$ is the set of integrable superreplicating portfolios,
i.e., the set of all measurable functions $(u_{1},u_{2},\Delta^{S},\Delta^{L})$ with $u_1\in L^{1}(\mu_{1})$, $u_2\in L^{1}(\mu_{2})$ that satisfy the superreplication constraint
\begin{equation}
\forall(s_{1},s_{2},v)\in(\RR_+^*)^2\times\mathbb{R}_{+},\qquad u_{1}(s_{1})+u_{2}(s_{2})+\Delta^{S}(s_{1},v)(s_{2}-s_{1})+\Delta^{L}(s_{1},v)\left(-\frac{2}{\tau}\ln\frac{s_{2}}{s_{1}}-v^{2}\right)\ge v.\label{eq:constraint}
\end{equation}
This linear program is known as the primal problem; $P$ stands for ``primal'' and $v$ stands for the value of
the VIX at the future date $T_{1}$. At time $T_{1}$, delta-hedging in
the S\&P~500 and in the FSLC is allowed.
The respective deltas, $\Delta^{S}(s_{1},v)$ and $\Delta^{L}(s_{1},v)$,
may depend on the values $s_{1}$ and $v$ of the S\&P~500 and the
VIX at $T_{1}$. Since the price at $T_{1}$
of the FSLC is $v^{2}$, the delta strategies are costless, and the price of the portfolio is $\mathbb{E}^{1}[u_{1}(S_{1})]+\mathbb{E}^{2}[u_{2}(S_{2})]$. Similarly, the lower bound on the VIX future is the largest price of a subreplicating portfolio,
\begin{equation}
P_\mathrm{sub}\equiv\sup_{\cU_\mathrm{sub}}\left\{\mathbb{E}^{1}[u_{1}(S_{1})]+\mathbb{E}^{2}[u_{2}(S_{2})]\right\}, \label{eq:primal_sub}
\end{equation}
where $\cU_\mathrm{sub}$ is the set of integrable subreplicating portfolios, defined like $\cU_\mathrm{super}$ but with the opposite inequality
\begin{equation}
\forall(s_{1},s_{2},v)\in(\RR_+^*)^2\times\mathbb{R}_{+},\qquad u_{1}(s_{1})+u_{2}(s_{2})+\Delta^{S}(s_{1},v)(s_{2}-s_{1})+\Delta^{L}(s_{1},v)\left(-\frac{2}{\tau}\ln\frac{s_{2}}{s_{1}}-v^{2}\right)\le v.\label{eq:constraint_sub}
\end{equation}
We shall see in Remark~\ref{rk:contStrategies} that the portfolios could also be required to be continuous, without changing the values of $P_\mathrm{super}$ and $P_\mathrm{sub}$.

\begin{rem}\label{rk:realVIX}
CBOE \cite{CBOE} actually uses a finite number of out-the-money options to approximate the log profile. If this trapezoidal approximation were known at time $T_0$, our analysis would apply to the real VIX formula by simply substituting the approximation for the log profile in the constraints (\ref{eq:constraint}) and (\ref{eq:constraint_sub}). However, the list of out-the-money options used in the official VIX formula depends on the locations of zero bid quotes and hence varies over time, so the approximation used for computing the official VIX at time $T_1$ is only known at $T_1$. Maturity interpolation is another practical concern if no options with maturity $T_2$ are quoted. Then CBOE considers the two closest quoted maturities $T_2^-$ and $T_2^+$ straddling $T_2$ and interpolates linearly the corresponding variance swap variances. Sub/superreplication then requires trading options with maturities $T_2^-$ and $T_2^+$. Instead, our idealized setting corresponds to interpolating the $T_2^-$ and $T_2^+$ smiles and trading only in the synthetic $T_2$ options.
\end{rem}

\begin{rem}\label{rk:generalization}
  An analogous linear program can be studied if the payoffs $-\frac{2}{\tau} \ln \frac{s_2}{s_1}$ and $v$ are replaced by general payoffs $g(s_1,s_2)$ and $f(s_1,s_2,p)$, where $p$ denotes the price at $T_1$ of $g(s_1,s_2)$. In particular, bounds for options written on the price at $T_1$ of forward-starting calls or puts can be found using the same approach.
\end{rem}

\subsection{The classical bounds for VIX futures}\label{sec:classicalBounds}

Suppose for the moment that $\sigma_{12}^{2}$ as defined in (\ref{eq:sigma12}) is nonnegative, as will be the case in the absence of arbitrage. Then, 
it is well known that
\begin{equation}
P_\mathrm{super}\le\sigma_{12}\label{eq:P_le_sigma12}.
\end{equation}
Indeed, if $\sigma_{12}>0$, the portfolio given by
\begin{equation}
{u}_{1}(s_{1})=\frac{\sigma_{12}}{2}-\frac{L(s_{1})}{2\sigma_{12}},\qquad{u}_{2}(s_{2})=\frac{L(s_{2})}{2\sigma_{12}},\qquad{\Delta}^{S}(s_{1},v)=0,\qquad{\Delta}^{L}(s_{1},v)=-\frac{1}{2\sigma_{12}}\label{eq:portfolio_superrep_classical}
\end{equation}
has price $\mathbb{E}^{1}[{u}_{1}(S_{1})]+\mathbb{E}^{2}[{u}_{2}(S_{2})]=\sigma_{12}$ and belongs to $\cU_\mathrm{super}$ because
\[
{u}_{1}(s_{1})+{u}_{2}(s_{2})+{\Delta}^{S}(s_{1},v)(s_{2}-s_{1})+{\Delta}^{L}(s_{1},v)\left(L\left(\frac{s_{2}}{s_{1}}\right)-v^{2}\right)=\frac{\sigma_{12}}{2}+\frac{1}{2\sigma_{12}}v^{2}=v+\frac{1}{2\sigma_{12}}(v-\sigma_{12})^{2}\geq v.
\]
This corresponds to the superreplication
of a straight line ($v$) by a tangent parabola ($\frac{\sigma_{12}}{2}+\frac{1}{2\sigma_{12}}v^{2}$), or, equivalently, to the superreplication of the square root ($\sqrt{v^2}$) by its tangent line at $v^2=\sigma_{12}^2$. If $\sigma_{12}=0$, one can simply replace
$\sigma_{12}$ by an arbitrarily small $\varepsilon>0$ in (\ref{eq:portfolio_superrep_classical}), showing that $P_\mathrm{super}\le0=\sigma_{12}$. Similarly, canceling the dependency of the left-hand side of the subreplication constraint (\ref{eq:constraint_sub})
on $(s_{1},s_{2})$ yields 
\begin{equation*}
P_\mathrm{sub}\ge 0\label{eq:Psub_ge_0}.
\end{equation*}
Thus, the classical lower bound is trivial. In the following sections, we shall investigate how to obtain bounds sharper than the interval $[0,\sigma_{12}]$.

%\medskip
%
%In this article two central questions are the following:
%\begin{itemize}
%\item 
%\item How do we practically build such portfolios? Numerically solving the primal problem can prove difficult.
%\end{itemize}
%
%Does there exist ``market smiles'' $\mu_{1}$ and $\mu_{2}$ for
%which we can find a better sub/superreplication portfolio, i.e., a superreplicating
%portfolio whose market price at time 0 is strictly smaller than $\sigma_{12}$, and a subreplicating portfolio whose market price at time 0 is positive? Can we characterize the ``market smiles'' $\mu_{1}$ and $\mu_{2}$
%for which this is possible? To answer these questions, let us introduce the dual formulation of the LP problems (\ref{eq:primal}) and (\ref{eq:primal_sub}).

%%%%%%%%%%%%%%%%%%%%%%%%%%%%%%%%%%%%%%%%%
\section{Arbitrages and martingale measures}\label{sec:absence_of_arbitrage}

In this section, we define (dynamic) arbitrage and relate its absence to the existence of certain risk-neutral measures.
A model-free arbitrage is usually defined as a strategy that has a negative price at time 0 and generates a nonnegative payoff at the final horizon $T_2$. We shall distinguish two types of arbitrages.

An \emph{$S$-arbitrage} is an arbitrage that only trades in the S\&P~500 and its vanilla options. More precisely, let~$\cU_S^0$ be the set of measurable functions $(u_{1},u_{2},\Delta)$ with $u_1\in L^{1}(\mu_{1})$, $u_2\in L^{1}(\mu_{2})$ such that
\begin{equation*}
\forall(s_{1},s_{2})\in(\RR_+^*)^2,\qquad u_{1}(s_{1})+u_{2}(s_{2})+\Delta(s_{1})(s_{2}-s_{1})\ge 0 \label{eq:constraint_arb_S};
\end{equation*}
then an $S$-arbitrage is an element of $\cU_S^0$ such that $\mathbb{E}^{1}[u_{1}(S_{1})]+\mathbb{E}^{2}[u_{2}(S_{2})]<0$.

An \emph{$(S,V)$-arbitrage}, on the other hand, also trades in the FSLC at $T_1$. Let $\cU_{S,V}^0$ be the set of measurable functions $(u_{1},u_{2},\Delta^S,\Delta^V)$ with $u_1\in L^{1}(\mu_{1})$, $u_2\in L^{1}(\mu_{2})$, such that
\begin{equation*}
\forall(s_{1},s_{2},v)\in(\RR_+^*)^2\times\mathbb{R}_{+},\quad u_{1}(s_{1})+u_{2}(s_{2})+\Delta^{S}(s_{1},v)(s_{2}-s_{1})+\Delta^{L}(s_{1},v)\left(L\left(\frac{s_{2}}{s_{1}}\right)-v^{2}\right)\ge 0;\label{eq:constraint_arb_S_V}
\end{equation*}
then an $(S,V)$-arbitrage is an element of $\cU_{S,V}^0$ such that $\mathbb{E}^{1}[u_{1}(S_{1})]+\mathbb{E}^{2}[u_{2}(S_{2})]<0$.
Since any such element can be scaled, we observe that there is an $(S,V)$-arbitrage in the market if and only if
\begin{equation*}
\inf_{\cU_{S,V}^0}\left\{\mathbb{E}^{1}[u_{1}(S_{1})]+\mathbb{E}^{2}[u_{2}(S_{2})]\right\} = -\infty, \label{eq:primal_arb_S}
\end{equation*}
and the analogue with $\cU_{S}^0$ holds for $S$-arbitrages.

Clearly, absence of $(S,V)$-arbitrage implies absence of $S$-arbitrage, and we shall see in Theorem~\ref{thm:arbitrage-free} that they are in fact equivalent. Before that, let us introduce the risk-neutral measures that are dual to the portfolios.

\begin{definition}\label{def:Mbar} 
	We denote by $\mathcal{M}(\mu_{1},\mu_{2})$ the set of all martingale laws~$\mu$ on~$(\mathbb{R}_{+}^*)^{2}$ with marginals~$\mu_{1}$
	and~$\mu_{2}$, i.e., probability measures $\mu$ such that 
	\[
	S_{1}\sim\mu_{1},\qquad S_{2}\sim\mu_{2},\qquad \mathbb{E}^{\mu}\left[S_{2}|S_{1}\right]=S_{1}.
	\]
	For $\mu\in\mathcal{M}(\mu_{1},\mu_{2})$,
	\begin{equation}
	\Lambda_\mu(S_1) \equiv \mathbb{E}^{\mu}\left[L\left(\frac{S_{2}}{S_{1}}\right)\middle|S_{1}\right] \label{eq:def_Lambda}
	\end{equation}
	is the price at $T_{1}$ of the FSLC under $\mu$, and we denote by
	$\bar{\mathcal{M}}(\mu_{1},\mu_{2})$ the subset of all $\mu\in\mathcal{M}(\mu_{1},\mu_{2})$ such that $\Lambda_\mu(S_1)$ is $\mu$-a.s.\ constant. In this case, necessarily $\Lambda_\mu(S_1) = \sigma_{12}^{2}$ $\mu$-a.s.
\end{definition}
%
%Note that ,
%\begin{eqnarray}
%\Lambda_\mu(S_1) & \ge & 0 \label{eq:Lambda_Positive}
%%\EE^{1}[\Lambda_\mu(S_1)] & = & \sigma_{12}^{2} \label{eq:Expectation_Lambda}
%\end{eqnarray}
%where (\ref{eq:Lambda_Positive}) follows from Jensen's inequality applied to the convex function $L$, and the facts that $\mathbb{E}^{\mu}\left[S_{2}|S_{1}\right]=S_{1}$ and $L(1)=0$.

Next, we define a set of measures on the extended space $(\RR_+^*)^2\times\mathbb{R}_{+}$, where the last coordinate will accommodate the VIX at $T_1$. With a mild abuse of notation, we write $(S_1,S_2,V)$ for the identity on this space.

\begin{definition}
  Let $\mathcal{M}_{V}(\mu_{1},\mu_{2})$ be the set of all the probability
measures $\mu$ on $(\RR_+^*)^2\times\mathbb{R}_{+}$ such
that
\begin{equation}
S_{1}\sim\mu_{1},\qquad S_{2}\sim\mu_{2},\qquad\mathbb{E}^{\mu}\left[S_{2}|S_{1},V\right]=S_{1},\qquad\mathbb{E}^{\mu}\left[L\left(\frac{S_{2}}{S_{1}}\right)\middle|S_{1},V\right]=V^{2}.
\label{eq:def_MV}
\end{equation}
\end{definition}

Note that for all $\mu\in\mathcal{M}_{V}(\mu_{1},\mu_{2})$, we have
\bea
\EE^\mu[V^2] = \sigma_{12}^2. \label{eq:Unconditional_expectation_V2}
\eea
More precisely, extending the definition (\ref{eq:def_Lambda}) of $\Lambda_\mu(S_1)$ to $\mu\in\mathcal{M}_{V}(\mu_{1},\mu_{2})$, we have
\begin{equation}
\EE^\mu[V^2|S_{1}]=\Lambda_\mu(S_1),
\label{eq:Expectation_V2}
\end{equation}
that is, projecting the VIX squared onto functions of $S_1$ always yields $\Lambda_\mu(S_1)$. Moreover, Jensen's inequality implies 
\begin{equation}\label{eq:Lambda_Positive}
 \Lambda_\mu(S_1)  \ge  0.
\end{equation}

%\medskip

To relate the spaces  $\mathcal{M}(\mu_{1},\mu_{2})$ and $\mathcal{M}_{V}(\mu_{1},\mu_{2})$, we introduce the following notation. For $\mu\in\mathcal{M}_{V}(\mu_{1},\mu_{2})$, let $\mu_{(1,2)}$ be the projection of $\mu$ onto the first two coordinates, i.e.,
\begin{equation}
\mu_{(1,2)}(A)=\mu(A\times\mathbb{R}_{+}), \qquad A\in \cB((\RR_+^*)^2)\label{eq:mu12}
\end{equation}
and let $\mu_{3}$ be the projection onto the third coordinate. Thus, $(S_{1},S_{2})\sim\mu_{(1,2)}$ and $V\sim\mu_{3}$ under $\mu$. Conversely, let $\mu\in\mathcal{M}(\mu_{1},\mu_{2})$, then we denote by
$\mu_\Lambda$ the law of $(S_{1},S_{2},\sqrt{\Lambda_\mu(S_1)})$ under $\mu$. Recalling (\ref{eq:Lambda_Positive}), 
$\mu_\Lambda$ is the unique probability
distribution on $(\RR_+^*)^2\times\mathbb{R}_{+}$ such that $(\mu_\Lambda)_{(1,2)}=\mu$ and $V^2=\Lambda_\mu(S_1)$ $\mu_\Lambda$-a.s. The following is an immediate consequence of the tower property.
%\end{itemize}
%The next lemma states that projecting a measure in $\mathcal{M}_{V}(\mu_{1},\mu_{2})$ onto the first two coordinates $s_1$ and $s_2$ gives an element of $\mathcal{M}(\mu_{1},\mu_{2})$. Recall the definition (\ref{eq:mu12}) of $\mu_{(1,2)}$.

\begin{lem}
\label{lem:projection_MV_M}
\begin{itemize}
\item[(i)] If $\mu\in\mathcal{M}_{V}(\mu_{1},\mu_{2})$,  then $\mu_{(1,2)}\in\mathcal{M}(\mu_{1},\mu_{2})$.
\item[(ii)] If  $\mu\in\mathcal{M}(\mu_{1},\mu_{2})$, then $\mu_\Lambda\in\mathcal{M}_V(\mu_{1},\mu_{2})$.
\end{itemize}
\end{lem}
%\begin{proof}
%(i) The two marginals of $\mu_{(1,2)}$ are $\mu_{1}$ and $\mu_{2}$. Moreover,
%\[
%\mathbb{E}^{\mu_{(1,2)}}\left[S_{2}|S_{1}\right]=\mathbb{E}^{\mu}\left[S_{2}|S_{1}\right]=\mathbb{E}^{\mu}\left[\mathbb{E}^{\mu}\left[S_{2}|S_{1},V\right]|S_{1}\right]=\mathbb{E}^{\mu}\left[S_{1}|S_{1}\right]=S_{1}
%\]
%Therefore $\mu_{(1,2)}\in\mathcal{M}(\mu_{1},\mu_{2})$.
%
%(ii) Let $\mu\in\mathcal{M}(\mu_{1},\mu_{2})$. Let us check that $\mu_\Lambda\in\mathcal{M}_{V}(\mu_{1},\mu_{2})$. The first
%two marginals of $\mu_\Lambda$ are $\mu_{1}$ and $\mu_{2}$. Moreover, since
%$V$ is $\mu_\Lambda$-a.s. a function of $S_1$, we have
%\[
%\mathbb{E}^{\mu_\Lambda}\left[S_{2}|S_{1},V\right]=\mathbb{E}^{\mu_\Lambda}\left[S_{2}|S_{1}\right]=\mathbb{E}^{\mu}\left[S_{2}|S_{1}\right]=S_{1}
%\]
%and
%\[
%\mathbb{E}^{\mu_\Lambda}\left[L\left(\frac{S_{2}}{S_{1}}\right)\middle|S_{1},V\right]=\mathbb{E}^{\mu_\Lambda}\left[L\left(\frac{S_{2}}{S_{1}}\right)\middle|S_{1}\right]=\mathbb{E}^{\mu}\left[L\left(\frac{S_{2}}{S_{1}}\right)\middle|S_{1}\right]=\Lambda_\mu(S_1)=V^{2}
%\]
%Therefore $\mu_\Lambda\in\mathcal{M}_{V}(\mu_{1},\mu_{2})$.
%\end{proof}

While elements of $\mathcal{M}_{V}(\mu_{1},\mu_{2})$ are interpreted as general joint models for the S\&P and the VIX, the second part of the lemma shows that elements of $\mathcal{M}(\mu_{1},\mu_{2})$ can be seen as special models of a ``local volatility'' type where the VIX is a function of $S_{1}$. Finally, elements of $\bar{\mathcal{M}}(\mu_{1},\mu_{2})$ are even more particular models where this function is constant; they are of ``Black--Scholes'' type as far as the forward volatility is concerned.

We can now formulate the main result of this section, stating that absence of arbitrage is equivalent to the existence of risk-neutral measures. More precisely, absence of $S$-arbitrages and $(S,V)$-arbitrages turns out to be equivalent, meaning that the possibility of trading the FSLC at $T_1$ does not add any restriction in our model-free setting. Hence, we will simply speak of \emph{absence of arbitrage} in later sections.

\begin{thm}\label{thm:arbitrage-free}
The following assertions are equivalent:
\begin{enumerate}
\item[(i)] The market is free of $S$-arbitrage,
\item[(ii)] the market is free of $(S,V)$-arbitrage,
\item[(iii)] $\mathcal{M}(\mu_{1},\mu_{2}) \neq \emptyset$,
\item[(iv)] $\mathcal{M}_{V}(\mu_{1},\mu_{2}) \neq \emptyset$,
\item[(v)] $\mu_1$ and $\mu_2$ are in convex order, i.e., $\EE^{1}[f(S_1)] \le \EE^{2}[f(S_2)]$ for any convex function $f:\RR_+^*\rightarrow\RR$.
\end{enumerate}
\end{thm}

\begin{proof} We show (ii) $\Rightarrow$ (i) $\Rightarrow$ (v) $\Rightarrow$ (iii) $\Rightarrow$ (iv) $\Rightarrow$ (ii); the first implication is obvious. To see that (i)~$\Rightarrow$~(v), assume that there exists a convex function $f:\RR_+^*\rightarrow\RR$ such that $\EE^{1}[f(S_1)] > \EE^{2}[f(S_2)]$. By convexity,  $f(s_{2})\ge f(s_{1})+f_{r}'(s_{1})(s_{2}-s_{1})$
for all $s_{1},s_{2}>0$, where $f'_r$ denotes the right derivative of $f$. As a consequence, the portfolio 
\[
u_{1}(s_{1})=- f(s_{1}),\qquad u_{2}(s_{2})= f(s_{2}),\qquad\Delta(s_{1})=- f_{r}'(s_{1})
\]
is an $S$-arbitrage. The implication (v)~$\Rightarrow$~(iii) is Strassen's theorem~\cite{strassen} and (iii)~$\Rightarrow$~(iv) is a direct consequence of Lemma~\ref{lem:projection_MV_M}(ii). Finally, let us prove that (iv)~$\Rightarrow$~(ii). Let $\mu\in \mathcal{M}_{V}(\mu_{1},\mu_{2})$ and let $(u_{1},u_{2},\Delta^S,\Delta^V)\in\cU_{S,V}^0$. Using~(\ref{eq:def_MV}), the price of this portfolio at time 0 is
\begin{multline*}
\mathbb{E}^{1}[u_{1}(S_{1})]+\mathbb{E}^{2}[u_{2}(S_{2})] 
= \mathbb{E}^{\mu}\left[u_{1}(S_{1})+u_{2}(S_{2})+\Delta^{S}(S_{1},V)(S_{2}-S_{1})+\Delta^{L}(S_{1},V)\left(L\left(\frac{S_{2}}{S_{1}}\right)-V^{2}\right)\right] \ge 0
\end{multline*}
and as consequence, the market is free of $(S,V)$-arbitrage.
\end{proof}

\begin{rem}\label{rem:extreme_rays_dim1}
Since $\mu_1$ and $\mu_2$ are probabilities on the real line with the same mean, they are in convex order
if and only if $\EE^{1}[(S_1-K)_+] \le \EE^{2}[(S_2-K)_+]$ for all $K\geq0$, i.e., calls with maturity $T_1$ are cheaper than calls with maturity $T_2$. See, e.g., \cite[Theorem~2.58]{FollmerSchied.11} for a proof.
\end{rem}

%%%%%%%%%%%%%%%%%%%%%%%%%%%%%%%%%%%%%%%%%
\section{Duality}\label{sec:duality}

In this section, we introduce the dual problems to sub/superreplicating the VIX and prove the absence of a duality gap as well as the existence of an extremal model. We let 
\begin{equation}
D_\mathrm{super}\equiv\sup_{\mu\in\mathcal{M}_{V}(\mu_{1},\mu_{2})}\mathbb{E}^{\mu}[V], \qquad D_\mathrm{sub}\equiv\inf_{\mu\in\mathcal{M}_{V}(\mu_{1},\mu_{2})}\mathbb{E}^{\mu}[V]. \label{eq:dual}
\end{equation}
These dual problems are of a non-standard type. Indeed, while maximizing (or minimizing) over $\mathcal{M}(\mu_{1},\mu_{2})$ is the ``martingale optimal transport'' problem (see, e.g, \cite{beiglbock,nutz,galichon}), the optimization over $\mathcal{M}_{V}(\mu_{1},\mu_{2})$ is quite different since one marginal (the law of the third component $V$) is not prescribed. The marginal is merely required to satisfy~(\ref{eq:def_MV}) and in particular, \eqref{eq:dual} is not a constrained Monge--Kantorovich transport problem.

Formally, the dual problem~\eqref{eq:dual} arises by permuting the inf and sup operators (written here for the superreplication problem) as shown below, where $\cM_+$ denotes the set of nonnegative measures on $(\RR_+^*)^2 \times \RR_+$ and acts as the set of Lagrange multipliers associated with the superreplication constraint:
\beaa
P_\mathrm{super} & \equiv & \inf_{(u_1,u_2,\Delta^S,\Delta^L)\in\cU_\mathrm{super}}\left\{\mathbb{E}^{1}[u_{1}(S_{1})]+\mathbb{E}^{2}[u_{2}(S_{2})]\right\} \\
& = & \inf_{(u_1,u_2,\Delta^S,\Delta^L) \;\mathrm{unconstrained}} \; \sup_{\mu\in \cM_+}\Big\{\mathbb{E}^{1}[u_{1}(S_{1})]+\mathbb{E}^{2}[u_{2}(S_{2})] \\
& & - \EE^\mu\left[ u_{1}(S_{1})+u_{2}(S_{2})+\Delta^{S}(S_{1},V)(S_{2}-S_{1})+\Delta^{L}(S_{1},V)\left(L\left(\frac{S_{2}}{S_{1}}\right)-V^{2}\right) - V\right]\Big\} \\
& \overset{(?)}{=} & \sup_{\mu\in \cM_+} \; \inf_{(u_1,u_2,\Delta^S,\Delta^L) \;\mathrm{unconstrained}} \Big\{\mathbb{E}^{1}[u_{1}(S_{1})]+\mathbb{E}^{2}[u_{2}(S_{2})] \\
& & - \EE^\mu\left[ u_{1}(S_{1})+u_{2}(S_{2})+\Delta^{S}(S_{1},V)(S_{2}-S_{1})+\Delta^{L}(S_{1},V)\left(L\left(\frac{S_{2}}{S_{1}}\right)-V^{2}\right) - V\right]\Big\} \\
& = & \sup_{\mu\in\mathcal{M}_{V}(\mu_{1},\mu_{2})}\mathbb{E}^{\mu}[V] \;\; \equiv \;\; D_\mathrm{super}.
\eeaa
The second equality stems from the fact that (a) if $(u_1,u_2,\Delta^S,\Delta^L)\in\cU_\mathrm{super}$ then the $\EE^\mu$ term is nonnegative for all $\mu$, hence the $\sup_{\mu\in \cM_+}$ term is equal to $\mathbb{E}^{1}[u_{1}(S_{1})]+\mathbb{E}^{2}[u_{2}(S_{2})]$, and (b) otherwise the $\EE^\mu$ term can be made as large and negative as desired by picking $\mu=\kappa \delta_{(s_1,s_2,v)}$ and letting $\kappa$ tend to $+\infty$, where $(s_1,s_2,v)$ are such that constraint (\ref{eq:constraint}) is violated, so the $\sup_{\mu\in \cM_+}$ term is equal to $+\infty$. The fourth equality uses that (a) if $\mu\in\mathcal{M}_{V}(\mu_{1},\mu_{2})$ then for all unconstrained $(u_1,u_2,\Delta^S,\Delta^L)$ the curly bracket term is equal to $\mathbb{E}^{\mu}[V]$, and (b) otherwise the curly bracket term can be made as large and negative as desired by picking a convenient $u_1$ or $u_2$ or $\Delta^S$ or $\Delta^L$ depending on which of the four conditions (\ref{eq:def_MV}) defining $\mathcal{M}_{V}(\mu_{1},\mu_{2})$ is not satisfied, so the infimum term is equal to $-\infty$.

It is well known that such a formal duality may fail for infinite dimensional linear programming problems. Next, we shall establish rigorously the absence of a duality gap for general options $f(s_{1},s_{2},v)$; this does not cause additional work compared to the VIX future.
We extend the definition of the dual problem to
\begin{equation*}
  D_\mathrm{super}\equiv\sup_{\mu\in\mathcal{M}_{V}(\mu_{1},\mu_{2})}\mathbb{E}^{\mu}[f] \label{eq:dualGeneralf}
\end{equation*}
and similarly extend $P_\mathrm{super}$ and $\cU_\mathrm{super}$ by writing $f(s_{1},s_{2},v)$ instead of $v$ on the right-hand side of~\eqref{eq:constraint}.

\begin{thm}\label{th:strongDuality}
 Let $f:\RR_+^*\times\RR_+^*\times\RR_+\to\RR$ be upper semicontinuous and satisfy
  \begin{equation}\label{eq:growthOfF}
   |f(s_{1},s_{2},v)|\leq C\big(1+s_{1}+ s_{2}+|L(s_{1})|+|L(s_{2})|+v^{2}\big)
  \end{equation}
  for some constant $C>0$. Then 
 $$
  D_\mathrm{super} \equiv\sup_{\mu\in\mathcal{M}_{V}(\mu_{1},\mu_{2})}\mathbb{E}^{\mu}[f] = \inf_{\cU_\mathrm{super}}\left\{\mathbb{E}^{1}[u_{1}(S_{1})]+\mathbb{E}^{2}[u_{2}(S_{2})]\right\} \equiv P_\mathrm{super}.
 $$
 Moreover, $D_\mathrm{super}\neq-\infty$ if and only if $\mathcal{M}_{V}(\mu_{1},\mu_{2})\neq\emptyset$, and in that case the supremum is attained.
\end{thm}

Choosing $f(s_{1},s_{2},v)= v$, this shows in particular that there is no duality gap for the superreplication of VIX futures, and that a worst-case model exists. The analogue for the subreplication follows by considering $f(s_{1},s_{2},v)= -v$.
In view of Theorem \ref{thm:arbitrage-free}, we also obtain yet another characterization for the absence of arbitrage.

\begin{corollary}\label{cor:arbitrage-free2}
The following assertions are equivalent:
\begin{enumerate}
\item[(i)] The market is arbitrage-free,
\item[(ii)] $D_\mathrm{super}\neq -\infty$,
\item[(iii)] $D_\mathrm{sub}\neq +\infty$.
\end{enumerate}
%When the market has an arbitrage, $D_\mathrm{super}=P_\mathrm{super}=P_\mathrm{super}^c=-\infty$ and $D_\mathrm{sub}=P_\mathrm{sub}=P_\mathrm{sub}^c=+\infty$.
\end{corollary}

Several duality results related to ours where previously obtained in \cite{beiglbock, BeiglbockJuillet.12, nutz, phl-demarco}, among others.  In our setting, the fact the one marginal of the measures in $\mathcal{M}_{V}(\mu_{1},\mu_{2})$ is not prescribed, combined with the non-compactness of the state space, necessitates a novel technique of proof.

%%%%%%%%%%%%%%%%%%%%%
\subsection{Proof of the duality theorem}

In the rest of this section, we report the proof of Theorem~\ref{th:strongDuality} in several steps; the strategy is to reduce our duality to a tailor-made auxiliary duality via the Minimax Theorem. Apart from having no duality gap, the auxiliary duality needs to satisfy two requirements: its constraints should be less restrictive than the ones of the original problem, and they need to be strong enough to imply the continuity and compactness properties needed for the application of the Minimax Theorem.

The growth condition \eqref{eq:growthOfF} immediately implies that $D_\mathrm{super}=-\infty$ if and only if $\mathcal{M}_{V}(\mu_{1},\mu_{2})$ is empty. Thus, we may focus on the case $\mathcal{M}_{V}(\mu_{1},\mu_{2})\neq \emptyset$, and then \eqref{eq:growthOfF} implies that $D_\mathrm{super}$ is finite.

%%%%%%%%%%%%%%%%%%%%%%%%%%%%%%%%%%%%%%%%%%%%%%%%%%%%%%%%%%%%%%%%%%%%
\subsubsection{Superhedge for a superlinearly growing function of $v^{2}$}

The main aim of this step is to find a superlinearly growing function of~$V^{2}$ which can be superhedged at a finite price. This will be crucial to prepare the ground for the Minimax Theorem later on.

\begin{lem}\label{le:delaValleeAppl}
  There exists a function $\xi:\RR_+\to\RR_+$ of superlinear growth such that $\EE^{i}[\xi(S_{i})]<\infty$ and $\EE^{i}[\xi(|L(S_{i})|)]<\infty$. Moreover, $\xi$ can be chosen convex, strictly increasing, and to satisfy $\xi(0)=\xi'(0)=0$.
\end{lem}

\begin{proof}
  The existence of $\xi$ follows from (the proof of) the de la Vall\'ee--Poussin theorem \cite[Theorem~II.22]{DellacherieMeyer.78} and Assumption~\ref{as:integrability}.
\end{proof}

\begin{lem}\label{le:hedgeForVixOption}
  There exists a continuous, increasing function $\phi:\RR_+\to\RR_+$ with $\phi(0)=0$ and $\phi(\infty)=\infty$ such that setting $\Delta^{L}(v)\equiv -1-\phi(v^{2})$, we have
	\begin{equation}\label{eq:hedgeForVixOption}
	  \phi(v^{2})v^{2}\leq \Delta^{L}(v)[L(s_{2}/s_{1})-v^{2}] + L(s_{2})-L(s_{1}) + \xi(|L(s_{1})|) + \xi(|L(s_{2})|)
	\end{equation}  
	for all $(s_{1},s_{2},v)\in(\RR_+^*)^{2}\times\RR_+$.
\end{lem}

\begin{proof}
  Let $\Delta^{L}(v)\equiv -1-\phi(v^{2})$; then the desired inequality~\eqref{eq:hedgeForVixOption} is equivalent to
	$$
	  \phi(v^{2}) L(s_{2}/s_{1}) \leq v^{2} + \xi(|L(s_{1})|) + \xi(|L(s_{2})|)
	$$
	for all $(s_{1},s_{2},v)\in(\RR_+^*)^{2}\times\RR_+$. Let $\zeta(x)\equiv \xi(x/2)$ for $x\geq0$ and extend this function to $\RR$ by setting $\zeta(x)\equiv \zeta(|x|)$ for $x<0$. Then,
  \beaa
    \zeta(L(s_{2}/s_{1})) = \zeta(L(s_{2}) - L(s_{1})) \leq \zeta(2|L(s_{2})|) + \zeta(2|L(s_{1})|) = \xi(|L(s_{2})|) + \xi(|L(s_{1})|)
  \eeaa
	by the convexity and symmetry of $\zeta$. Thus, it is sufficient to find $\phi$ such that
	\begin{equation}\label{eq:proofHedgeForVixOption}
	  \phi(v^{2}) L(s_{2}/s_{1}) \leq v^{2} + \zeta(L(s_{2}/s_{1})).
  \end{equation}
	
	Indeed, define 
	\begin{equation}\label{eq:proofHedgeForVixOption2}
	  \phi(y)\equiv \inf_{b>0} \frac{y+\zeta(b)}{b},\quad\quad y\geq0.
  \end{equation}
	Then $\phi$ is nonnegative and concave, and $\phi(y)b \leq y+\zeta(b)$ for all $b>0$. This inequality is trivial for $b\leq0$, so that
	$$
	  \phi(y)b \leq y+\zeta(b),\quad\quad b\in\RR.
	$$
	In particular, choosing $y=v^{2}$ and $b=L(s_{2}/s_{1})$, we see that~\eqref{eq:proofHedgeForVixOption} is satisfied. 
	It follows from~\eqref{eq:proofHedgeForVixOption2} that $\phi$ is increasing, $\zeta(0)=\zeta'(0)=0$ implies that $\phi(0)=0$, and the superlinear growth of $\zeta$ yields that $\phi(\infty)=\infty$.
\end{proof}

\subsubsection{An auxiliary duality}

Let $(\Omega,\cF)$ be a measurable space and let $\fP(\Omega)$ be the set of all probability measures on $(\Omega,\cF)$. %We assume that singletons are measurable, i.e., $\{\omega\}\in\cF$ for all $\omega\in\Omega$.
It is well known that infinite-dimensional linear programming duality fails in general, and often topological conditions are used to obtain a positive result. The following lemma holds in the space of measurable functions without any topology; instead, it is based on two features: finite-dimensional Lagrange multipliers and the no-arbitrage type condition~\eqref{eq:NAcondition}. Thus, it is in the spirit of the robust duality results provided in~\cite{BouchardNutz.13} for equality constraints and in~\cite{BayraktarZhou.14} for inequality constraints. We use the notation $\mu(f)\equiv \EE^{\mu}[f]$.

\begin{lem}\label{le:umbrella}
  Let $w=(w^{1},\dots,w^{n}): \Omega\to\R^{n}$ be measurable and such that for all $\alpha=(\alpha^{1},\dots,\alpha^{n})\in\R^{n}_{+}$, 
  \begin{equation}\label{eq:NAcondition}
    \alpha \cdot w\geq0\quad\mbox{implies}\quad \alpha \cdot w\equiv0.
  \end{equation}
  Moreover, let 
  \begin{equation}\label{eq:defPi}
    \Pi\equiv \{\pi\in\fP(\Omega) \,|\; \pi(|w^{i}|)<\infty, \; \pi(w^{i})\leq0,\;i=1,\dots,n\}.
\end{equation}
  For any measurable function $f: \Omega\to\R$, we have
  $$
    \sup_{\pi\in\Pi} \pi(f) = \inf \{x\in\R \,|\; x+\alpha \cdot w \geq f \mbox{ for some } \alpha\in\R^{n}_{+}\}
  $$
  and the infimum is attained if it is finite.
\end{lem}

%NOTE: The expectation is defined as $-\infty$ if positive and negative part of the integrand are have infinite integral.

\begin{proof}
  This is a special case of the one-step duality result in~\cite[Theorem~4.3]{BayraktarZhou.14}. Indeed, let $\cP\equiv\fP(\Omega)$ be the set of all probability measures on $(\Omega,\cF)$; then~\eqref{eq:NAcondition} is the robust no-arbitrage condition $\mathrm{NA}(\cP)$ of~\cite{BayraktarZhou.14,BouchardNutz.13}. Moreover, if we see $w$ as the increment of a stock price vector $S$ over a single period, then $\Pi$ is precisely the set of supermartingale measures for $S$. The lemma then follows from~\cite[Theorem~4.3]{BayraktarZhou.14} once we note that a set which is $P$-null for all $P\in\cP$ is necessarily empty since $\cP$ includes the Dirac measures at all points of~$\Omega$.
\end{proof}

\begin{rem}\label{rk:rk:necessityOfNA}
  Condition~\eqref{eq:NAcondition} is necessary for the validity of Lemma~\ref{le:umbrella}. Indeed, let $\Omega=\R_{+}^{*}$ and $n=1$. We consider the butterfly and put payoffs
  $$
    w(x)=(1-2|x-1/2|)^{+},\qquad f(x)=(1-x)^{+}
  $$  
  which are bounded continuous functions on $\Omega$. Then, $\Pi=\{\pi\in\fP(\Omega)\,|\, \pi((0,1))=0\}$ and hence $\sup_{\pi\in\Pi} \pi(f)=0$. However, the infimum equals one because $\lim_{x\to0+} \alpha w(x)=0$ for any $\alpha\geq0$ whereas $\lim_{x\to0+} f(x)=1$. In particular, there is a duality gap.
\end{rem}  

We shall apply the lemma as follows. Let $\Omega\equiv(\RR_+^*)^{2}\times\RR_+$ and let $\cf$ be its Borel $\sigma$-field. Moreover, let $\xi$ be the function introduced in Lemma~\ref{le:delaValleeAppl} and let
$$
 w^{i}= \xi(S_{i}) - \EE^{i}[\xi(S_{i})] - \xi(1), \quad  i\in\{1,2\},
$$
$$
 w^{i+2}= \xi(|L(S_{i})|) - \EE^{i}[\xi(|L(S_{i})|)]-1, \quad  i\in\{1,2\},
$$
$$
  w^{5}= \Phi(V) - m-1,
$$
where $\Phi(v)\equiv \phi(v^2)v^{2}$ and $m\geq0$ is the price of the right-hand side in~\eqref{eq:hedgeForVixOption} as implied by $\mu_{1}$ and $\mu_{2}$,
$$
  m\equiv \EE^2[L(S_{2})] - \EE^1[L(S_{1})] +\EE^1[\xi(|L(S_{1})|)] + \EE^2[\xi(|L(S_{2})|)].
$$
Let $w=(w^{1},\dots,w^{5})$. Since the functions $w^{i}$ all have strictly negative values at the point $(s_1,s_2,v)=(1,1,0)$, the no-arbitrage condition~\eqref{eq:NAcondition} holds.
%The functions $w^{i}$ all have strictly positive and negative values, and using the explicit form of the $w^{i}$ we see immediately that one cannot form a nonnegative linear combination $\alpha\cdot w$ with $\alpha\neq0$. 
We define $\Pi$ as in~\eqref{eq:defPi}. Moreover, let $G$ be the cone consisting of all functions of the form $g=\alpha\cdot w$ for some $\alpha\in\R^{n}_{+}$. Lemma~\ref{le:umbrella} yields the following.
	
\begin{corollary}\label{co:umbrella}
  Let $\tilde f: \Omega\to\RR$ be measurable. Then
	\begin{align*}
    \inf \{x\in\RR \,|\, x+g\geq \tilde f \mbox{ for some }g\in G\}
	  & = \sup_{\pi\in\Pi} \pi(\tilde f).
	\end{align*}
\end{corollary}

%%%%%%%%%%%%%%%%%%%%%%%%%%%%%%%%%%%%%%%%%%%%%%%%%%%%%
\subsubsection{Proof of duality for VIX options} 

We shall apply Corollary~\ref{co:umbrella} in combination with the Minimax Theorem. For that, we need the following facts.

\begin{lem}\label{le:compactnessOfPi}
  The set $\Pi$ is weakly compact. Moreover, if $f$ is as in Theorem~\ref{th:strongDuality}, then $\pi\mapsto \pi(f)$ is weakly upper semicontinuous on $\Pi$.
\end{lem}

\begin{proof} 
  To see that $\Pi$ is closed, let $\pi_n\in \Pi$ converge weakly to $\pi$. Let $i\in\{1,\ldots,5\}$ and note that $w^i$ is continuous and bounded from below. Hence, $w^i\wedge N$ is bounded and $\pi(w^i\wedge N) = \lim_{n\rightarrow \infty} \pi_n(w^i\wedge N)$  for $N\in\NN$. In view of $\pi_n(w^i\wedge N)\le \pi_n(w^i)\le 0$, it follows that $\pi(w^i\wedge N)\le 0$ and then monotone convergence yields $\pi(w^i)\le 0$, i.e., $\pi\in\Pi$.
  The tightness of $\Pi$ follows from 
  $$
    \lim_{s\to \infty} \xi(s)=\infty, \quad \lim_{s\to 0,\infty} \xi(|L(s)|)=\infty,\quad \lim_{v\to\infty} \Phi(v)=\infty
  $$
  and
  \begin{align}\label{eq:UIbound}
    \sup_{\pi\in\Pi} \EE^{\pi}\Bigg[\sum_{i\in\{1,2\}} \{\xi(S_{i})+\xi(|L(S_{i})|)\}+\Phi(V)\Bigg]<\infty;
  \end{align}
  indeed, the above is bounded by $\sum_{i\in\{1,2\}} \{\EE^{i}[\xi(S_{i})] + \EE^{i}[\xi(|L(S_{i})|)]\} + m  +2\xi(1)+3$ for all $\pi\in\Pi$. Thus, compactness holds by Prokhorov's theorem.  
Given $N\in\R$, the semicontinuity of $f$ implies the semicontinuity of $\pi\mapsto \pi(f\wedge N)$, by the Portmanteau theorem. To see that $\pi\mapsto \pi(f)$ is semicontinuous as well, we pass to the limit $N\to\infty$ using the growth condition on $f$, the superlinear and superquadratic growth of $\xi$ and $\Phi$, respectively, and~\eqref{eq:UIbound}.
\end{proof}

We have now prepared all the tools for the proof of the main result.

\begin{proof}[Proof of Theorem~\ref{th:strongDuality}]
  We first show that the supremum is attained. Indeed, note that $\mathcal{M}_{V}(\mu_1,\mu_2)\subset \Pi$ is a closed subset because the defining conditions~\eqref{eq:def_MV} are determined by bounded continuous test functions. Thus, it is weakly compact by Lemma~\ref{le:compactnessOfPi} which also shows that $\pi\mapsto \pi(f)$ is weakly upper semicontinuous and thus attains its supremum on $\mathcal{M}_{V}(\mu_1,\mu_2)$.
  
  Next, we turn to the duality. Let us write $H$ for the set of all functions $h$ of the form
  $$
	  h (s_{1},s_{2},v) = u_{1}(s_{1})-\EE^{1}[u_{1}(S_1)] + u_{2}(s_{2}) - \EE^{2}[u_{2}(S_2)] + \Delta^{S}(s_{1},v)(s_{2}-s_{1}) + \Delta^{L}(s_{1},v)(L(s_{2}/s_{1})-v^{2})
 $$
	where $u_{i}\in L^{1}(\mu_{i})$ and $\Delta^{S},\Delta^{L}$ are measurable. Note that $H$ is a linear space. Moreover, we let $H^{c}$ (resp.~$H^{cb}$) be the subspace consisting of those functions $h$ whose coefficients  $u_{i},\Delta^{S},\Delta^{L}$ are continuous (resp.\ continuous and bounded). By the definition of the primal problem, we then have
	$$
	  P_\mathrm{super} = \inf \{x\in\RR \,|\, x+h\geq f \mbox{ for some }h\in H\} \leq   \inf \{x\in\RR \,|\, x+h\geq f \mbox{ for some }h\in H^c\} \equiv P_\mathrm{super}^c .
	$$
	 Using Lemma~\ref{le:hedgeForVixOption} and in particular the definition of $w^{5}$, we see that for every $g\in G$ there exists $h\in H^c$ such that $g\leq h$. Together with the fact that $H^c$ is a linear space, that yields the first equality in
	\begin{align*}
	  P_\mathrm{super}^c
	  &=\inf \{x\in\RR \,|\, x+g+h\geq f \mbox{ for some }g\in G,\, h\in H^c\}\\
	  &\leq \inf \{x\in\RR \,|\, x+g+h\geq f \mbox{ for some }g\in G,\, h\in H^{cb}\}\\
	  & = \inf_{h\in H^{cb}} \inf \{x\in\RR \,|\, x+g\geq f-h \mbox{ for some }g\in G\}\\
	  & = \inf_{h\in H^{cb}} \sup_{\pi\in\Pi} \pi(f-h),
	\end{align*}
	where the last step used Corollary~\ref{co:umbrella} applied to the function $\tilde f\equiv f-h$ with a fixed $h\in H^{cb}$; note that $\tilde f$ is upper semicontinuous and satisfies~\eqref{eq:growthOfF}. Hence, by Lemma~\ref{le:compactnessOfPi}, the linear mapping $\pi\mapsto \pi(f-h)$ is weakly upper semicontinuous on the weakly compact set $\Pi$. On the other hand, $h\mapsto \pi(f-h)$ is linear on $H^{cb}$ for fixed $\pi\in\Pi$. In particular, $(\pi,h)\mapsto \pi(f-h)$ is ``concave-convexlike'' in the sense of~\cite{Sion.58} and it follows via the Minimax Theorem in the form of \cite[Theorem~4.2]{Sion.58} that
	$$
	  P_\mathrm{super}^c  \leq  \inf_{h\in H^{cb}} \sup_{\pi\in\Pi} \pi(f-h) = \sup_{\pi\in\Pi} \inf_{h\in H^{cb}} \pi(f-h).
	$$
	By linearity of $H^{cb}$, we have $\inf_{h\in H^{cb}} \pi(f-h) =-\infty$ if $\pi\in\Pi\setminus\cm_{V}(\mu_{1},\mu_{2})$, and since $\pi(h)=0$ for $\pi\in\cm_{V}(\mu_{1},\mu_{2})$, we conclude that
	$$
	  P_\mathrm{super}^c
	  \leq \sup_{\pi\in\cm_{V}(\mu_{1},\mu_{2})} \inf_{h\in H^{cb}} \pi(f-h)
	  \; = \; \sup_{\pi\in\cm_{V}(\mu_{1},\mu_{2})} \pi(f) \; = \; D_\mathrm{super}.
	$$
	As a consequence, $P_\mathrm{super}\leq P_\mathrm{super}^{c} \leq D_\mathrm{super}$.
	
	It remains to prove the converse inequality. Let $\mu\in\mathcal{M}_{V}(\mu_{1},\mu_{2})$ and $(u_1, u_2, \Delta^S, \Delta^L)\in \cU_\mathrm{super}$; thus
\begin{equation*}
 u_{1}(s_{1})+u_{2}(s_{2})+\Delta^{S}(s_{1},v)(s_{2}-s_{1})+\Delta^{L}(s_{1},v)\left(L\left(\frac{s_{2}}{s_{1}}\right)-v^{2}\right)\ge f(s_{1},s_{2},v).
\end{equation*}
Taking expectations on both sides yields
$
\mathbb{E}^{1}[u_{1}(S_{1})]+\mathbb{E}^{2}[u_{2}(S_{2})] \ge \mu(f)
$ 
and as $\mu\in\mathcal{M}_{V}(\mu_{1},\mu_{2})$ and $(u_1, u_2, \Delta^S, \Delta^L)\in \cU_\mathrm{super}$ were arbitrary, it follows that $P_\mathrm{super} \ge D_\mathrm{super}$. This completes the proof that $P_\mathrm{super}=P_\mathrm{super}^c=D_\mathrm{super}$.
\end{proof}

\begin{rem}\label{rk:contStrategies}
  One consequence of the above proof is that $P_\mathrm{super}=P_\mathrm{super}^c$, i.e., the value of the primal problem does not change if we require the functions $u_{i}, \Delta^{S}, \Delta^{L}$ to be continuous. %In addition, the delta $\Delta^{S}$ could be required to be bounded.
\end{rem}

%%%%%%%%%%%%%%%%%%%%%%%%%%
\subsection{Local volatility property of the superreplication price}\label{sec:localVolSuperrep}

The following result shows that the superreplication bound for the VIX future can be computed on the dual side by merely maximizing over models of ``local volatility'' type, i.e., models $\hat\mu\in \mathcal{M}_{V}(\mu_{1},\mu_{2})$ of the form $\hat\mu=\mu_{\Lambda}$ for some $\mu\in\mathcal{M}(\mu_{1},\mu_{2})$, with the notation introduced above Lemma~\ref{lem:projection_MV_M}.
This property greatly simplifies the computation of $D_\mathrm{super}$; it is particular to superreplication of the VIX because it is based on the concavity of the square root.

\begin{prop}
\label{prop:optimality_of_local_vol}We have
\begin{equation*}
D_\mathrm{super}=\sup_{\mu\in\mathcal{M}(\mu_{1},\mu_{2})}\EE^{\mu_{\Lambda}}\left[V\right] =\sup_{\mu\in\mathcal{M}(\mu_{1},\mu_{2})}\EE^1\left[\sqrt{\Lambda_\mu(S_1)}\right].
\end{equation*}
\end{prop}

\begin{proof}
We have $D_\mathrm{super}\geq \sup_{\mu\in\mathcal{M}(\mu_{1},\mu_{2})}\EE^{\mu_{\Lambda}}\left[V\right]$ by Lemma~\ref{lem:projection_MV_M} and the second equality holds by the definition of $\mu_{\Lambda}$. To see the converse inequality, let $\mu\in\mathcal{M}_{V}(\mu_{1},\mu_{2})$ and consider $\hat\mu\equiv (\mu_{(1,2)})_{\Lambda}$; cf.\ (\ref{eq:mu12}) for the notation. Then, using (\ref{eq:Expectation_V2}) and the concavity of the square root,
\begin{equation*}
\mathbb{E}^{\mu}[V] = \mathbb{E}^{\mu}\left[\mathbb{E}^{\mu}\left[\sqrt{V^2}\middle|S_1\right]\right] \le \mathbb{E}^{\mu}\left[\sqrt{\mathbb{E}^{\mu}[V^2|S_1]}\right] = \mathbb{E}^{\mu}\left[\sqrt{\Lambda_{\mu}(S_1)}\right]= \mathbb{E}^{\mu}\left[\sqrt{\Lambda_{\mu_{(1,2)}}(S_1)}\right] =\mathbb{E}^{\hat\mu}\left[V\right]. 
\end{equation*}
As $\mu\in\mathcal{M}_{V}(\mu_{1},\mu_{2})$ was arbitrary, the result follows.
\end{proof}

%%%%%%%%%%%%%%%%%%%%%%%%%%%%%%%%%%%%%%%%%%%%%%%%%%%%%%%%
\section{Characterization of market smiles for which the classical bounds are optimal}\label{sec:characterization_classical_is_optimal}

In this section, we characterize the market smiles $\mu_{1},\mu_{2}$ for which the classical upper bound $P_\mathrm{super}\leq\sigma_{12}$ and the classical lower bound $P_\mathrm{sub}\geq0$ from Section~\ref{sec:classicalBounds} are already optimal.  
Let us first consider the subreplication problem; here our result shows that the classical bound is never sharp in practice.

\begin{thm}\label{thm:D=0}
  We have $P_\mathrm{sub}=0$ if and only if $\mu_1=\mu_2$.
\end{thm}

\begin{proof}
Assume that $P_\mathrm{sub}=0$, thus $D_\mathrm{sub}=0$ and the infimum defining $D_\mathrm{sub}$ is attained (Theorem \ref{th:strongDuality}). Hence,  there exists $\mu\in\mathcal{M}_{V}(\mu_{1},\mu_{2})$ such that $\mathbb{E}^{\mu}[V]=0$. Since $V\ge 0$ $\mu$-a.s., this means that $V=0$ $\mu$-a.s.\ and thus $\sigma_{12}^2=0$ by (\ref{eq:Unconditional_expectation_V2}), whence $\mu_1=\mu_2$. Conversely, if $\mu_1=\mu_2$, let $\mu$ be the law of $(S_{1},S_{1},0)$ under $\mu_{1}$. Then $\mu\in\mathcal{M}_{V}(\mu_{1},\mu_{2})$ and $\mathbb{E}^{\mu}[V]=0$, so that $P_\mathrm{sub}=D_\mathrm{sub}=0$.
\end{proof}

Next, we consider the superreplication problem where the characterization turns out to be nondegenerate. We recall that two distributions $\nu_{1},\nu_{2}$ on $\mathbb{R}_{+}^*\times\mathbb{R}$ are said to be in convex order if 
$
\mathbb{E}^{\nu_{1}}[g] \le \mathbb{E}^{\nu_{2}}[g] 
$
for any convex function $g\,:\,\mathbb{R}_{+}^*\times\mathbb{R}\rightarrow\mathbb{R}$.

\pagebreak[2]

\begin{thm}
\label{thm:D=Dbar}The following assertions are equivalent:
\begin{enumerate}
\item[(i)] $P_\mathrm{super}=\sigma_{12}$,
\item[(ii)] there exists $\mu\in\mathcal{M}_{V}(\mu_{1},\mu_{2})$ such that $V=\sigma_{12}$ $\mu$-a.s.,
\item[(iii)] $\bar{\mathcal{M}}(\mu_{1},\mu_{2})\neq\emptyset$,
\item[(iv)] $\Law_{\mu_{1}}(S_{1}, L(S_{1})-\ell_1)$ and $\Law_{\mu_{2}}(S_{2}, L(S_{2}) - \ell_2)$ are in convex order,
\item[(iv')] $\Law_{\mu_{1}}(S_{1}, L(S_{1})+\sigma_{12}^2)$ and $\Law_{\mu_{2}}(S_{2}, L(S_{2}))$ are in convex order.
\end{enumerate}
\end{thm}

\begin{proof}
(i) $\Rightarrow$ (ii)
By Theorem \ref{th:strongDuality}, the supremum defining $D_\mathrm{super}=P_\mathrm{super}$ is attained, so there exists $\mu\in\mathcal{M}_{V}(\mu_{1},\mu_{2})$
such that $\mathbb{E}^{\mu}[V]=\sigma_{12}$.
Thus, by (\ref{eq:Unconditional_expectation_V2}), $\mathbb{E}^{\mu}\left[V^2\right] =\sigma_{12}^2=\mathbb{E}^{\mu}\left[V\right]^{2}$ which by the strict convexity of the square implies $V=\sigma_{12}$ $\mu$-a.s.

(ii) $\Rightarrow$ (iii) Let $\mu\in\mathcal{M}_{V}(\mu_{1},\mu_{2})$ be
such that $V=\sigma_{12}$ $\mu$-a.s. By Lemma \ref{lem:projection_MV_M}(i), the projection $\mu_{(1,2)}$ of $\mu$ onto the first two coordinates belongs to $\mathcal{M}(\mu_{1},\mu_{2})$, and
\[
\mathbb{E}^{\mu_{(1,2)}}\left[L\left(\frac{S_{2}}{S_{1}}\right)\middle|S_{1}\right]=\mathbb{E}^{\mu}\left[L\left(\frac{S_{2}}{S_{1}}\right)\middle|S_{1}\right]=\mathbb{E}^{\mu}\left[\mathbb{E}^{\mu}\left[L\left(\frac{S_{2}}{S_{1}}\right)\middle|S_{1},V\right]\middle|S_{1}\right]=\mathbb{E}^{\mu}\left[V^{2}\middle|S_{1}\right]=\sigma_{12}^{2}.
\]
As a consequence, $\mu_{(1,2)}\in\bar{\mathcal{M}}(\mu_{1},\mu_{2})$ and in particular
$\bar{\mathcal{M}}(\mu_{1},\mu_{2})\neq\emptyset.$

(iii) $\Leftrightarrow$ (iv) $\Leftrightarrow$ (iv') Let us define $M_{1}=(S_{1},L(S_{1})-\ell_1)$ and $M_{2}=(S_{2},L(S_{2}) - \ell_2)$ as well as 
\[
\mu_{M_{1}}(dx,dy)=\mu_{1}(dx)\delta_{L(x)-\ell_1}(dy),\qquad\mu_{M_{2}}(dx,dy)=\mu_{2}(dx)\delta_{L(x) - \ell_2}(dy).
\]
Then, $\bar{\mathcal{M}}(\mu_{1},\mu_{2})$ is precisely the set of probability measures $\nu$
on $(\mathbb{R}_{+}^*)^{2}$ such that
\begin{equation*}
M_{1} \sim \mu_{M_{1}},\qquad M_{2}\sim\mu_{M_{2}},\qquad\mathbb{E}^{\nu}\left[M_{2}|M_{1}\right]=M_{1}.
\end{equation*}
By Strassen's theorem \cite{strassen}, this set is nonempty if and only if $\mu_{M_{1}}$ and $\mu_{M_{2}}$ are in convex order, which yields the equivalence of~(iii) and~(iv), and then also of~(iv') due to $\sigma_{12}^2 = \ell_2 -\ell_1$.

(iii) $\Rightarrow$ (i) 
Let $\mu\in\bar{\mathcal{M}}(\mu_{1},\mu_{2})$ and recall the definition of $\mu_\Lambda$ above Lemma \ref{lem:projection_MV_M}. Since $\mu_\Lambda\in\mathcal{M}_{V}(\mu_{1},\mu_{2})$ and $V=\sqrt{\Lambda_\mu(S_1)}=\sigma_{12}$
$\mu_\Lambda$-a.s., we have $P_\mathrm{super}=D_\mathrm{super}\ge\sigma_{12}$. Conversely, $P_\mathrm{super}\leq\sigma_{12}$ by~\eqref{eq:P_le_sigma12}.
\end{proof}

%In Sections \ref{sec:binomial}, \ref{sec:mu2_compact_support}, and \ref{sec:three-point} we will examine particular examples and general situations where $D_\mathrm{super}<\sigma_{12}$. Section \ref{sec:examples_where_classical_is_optimal} gives examples where $D_\mathrm{super}=\sigma_{12}$.

The necessary and sufficient conditions in Theorem~\ref{thm:D=Dbar} are not straightforward to check given the marginals. While the convex ordering of two measures on $\R$ can be verified by computing the one-parameter family of call option prices, cf.\ Remark~\ref{rem:extreme_rays_dim1}, no simple family of test functions exists in two or more dimensions. See also Johansen \cite{johansen72,johansen74} and Scarsini \cite{scarsini} for more precise (negative) results.
Thus, we are interested in simpler criteria, at the expense of not being sharp. The following is a condition that involves only call and put prices, and we shall give more conditions in the context of the examples in Sections~\ref{sec:mu2_compact_support} and~\ref{sec:examples_where_classical_is_optimal}.

\begin{prop}\label{prop:general_sufficient_condition_D<sigma12}
Denote by $C_i(K)$ and $P_i(K)$ the prices at time 0 of the call and put options with maturity~$T_i$ and strike~$K$. Let 
\begin{equation*}
\Psi_1(K) \equiv \Phi_1\left(K + \frac{1}{2}\sigma_{12}^{2}\tau\right), \quad \Psi_2(K) \equiv \Phi_2(K),
\label{eq:E11}
\end{equation*}
where
\begin{equation*}
\Phi_i(K) = \begin{cases}
\frac{C_i(e^{K})}{S_{0}}-\int_{e^{K}}^{\infty}\frac{C_i(k)}{k^{2}}dk & \mbox{if}\,\,K>\ln S_{0}\\
\ln S_{0}-K+\frac{P_i(e^{K})}{S_{0}}-\int_{e^{K}}^{S_{0}}\frac{P_i(k)}{k^{2}}dk-\int_{S_{0}}^{\infty}\frac{C_i(k)}{k^{2}}dk & \mbox{otherwise.}
\end{cases}
\end{equation*}
If there exists $K\in\mathbb{R}$ such that $\Psi_1(K)>\Psi_2(K)$,
then $P_\mathrm{super}<\sigma_{12}$.
\end{prop}

\begin{proof}
By Theorem \ref{thm:D=Dbar}, a necessary condition for $P_\mathrm{super}=\sigma_{12}$
is that $\Law_{\mu_{1}}(\ln S_{1}-\frac{1}{2}\sigma_{12}^{2}\tau)$ and $\Law_{\mu_{1}}(\ln S_{2})$
are in convex order. Since $\mathbb{E}^{1}[\ln S_{1}-\frac{1}{2}\sigma_{12}^{2}\tau]=\mathbb{E}^{2}[\ln S_{2}]$,
this is equivalent to
\[
\forall K\in\mathbb{R},\quad\mathbb{E}^{1}[(\ln S_{1}-\frac{1}{2}\sigma_{12}^{2}\tau-K)_{+}]\le\mathbb{E}^{2}[(\ln S_{2}-K)_{+}].
\]
Now, from the Carr--Madan formula \cite{carr-madan}, 
\begin{multline*}
\mathbb{E}^{i}[(\ln S_{i}-K)_{+}]=(\ln S_{0}-K)_{+}+\frac{\mathbb{E}^{i}[S_{i}]-S_{0}}{S_{0}}\mathbf{1}_{\ln S_{0} \ge K}
-\int_{0}^{S_{0}}\mathbf{1}_{k\ge e^{K}}\frac{P_i(k)}{k^{2}}dk \\ -\int_{S_{0}}^{\infty}\mathbf{1}_{k\ge e^{K}}\frac{C_i(k)}{k^{2}}dk +\frac{O_i(e^{K})}{S_{0}}
\end{multline*}
where $O_i(e^{K})$ denotes the
price of the out-the-money option, i.e., $O_i(e^{K})=P_{i}(e^{K})$ if $e^{K}<S_{0}$ 
and $O_i(e^{K})=C_{i}(e^{K})$ otherwise. Since $\mathbb{E}^{i}[S_i] = S_{0}$, we have $\mathbb{E}^{i}[(\ln S_{i}-K)_{+}]=\Phi_i(K)$ and
\begin{equation*}
\Psi_1(K) = \mathbb{E}^{1}[(\ln S_{1}-\frac{1}{2}\sigma_{12}^{2}\tau-K)_{+}], \quad \Psi_2(K) = \mathbb{E}^{2}[(\ln S_{2}-K)_{+}].
\end{equation*}
Therefore, the necessary condition for $P_\mathrm{super} = \sigma_{12}$ can be stated as $\Psi_{1}(K)\le \Psi_{2}(K)$ for all $K\in\R$. 
If this condition is not met, then $P_\mathrm{super}<\sigma_{12}$ since $P_\mathrm{super} \le \sigma_{12}$ by~\eqref{eq:P_le_sigma12}.
\end{proof}
%\begin{rem}
%When $K$ tends to $+\infty,$ both $\Psi_{1}(K)$ and $\Psi_{2}(K)$ tend
%to 0. When $K$ tends to $-\infty$, denoting by $\mathrm{VS}_i$ the price of the variance swap with maturity $T_i$,
%%\begin{eqnarray*}
%%\Psi_{1}(K) -(\ln S_{0}-K) & \underset{K \to -\infty}{\longrightarrow} & -\frac{1}{2}\sigma_{12}^{2}\tau-\int_{0}^{S_{0}}\frac{P(T_{1},k)}{k^{2}}dk-\int_{S_{0}}^{\infty}\frac{C(T_{1},k)}{k^{2}}dk\\
%% & = & \ln S_{0}-\frac{1}{2}\sigma_{12}^{2}\tau-K-\frac{T_{1}}{2}\mathrm{VS}(T_{1})
%%\end{eqnarray*}
%\begin{eqnarray*}
%\Psi_{1}(K) -(\ln S_{0}-K) & \underset{K \to -\infty}{\longrightarrow} & -\frac{1}{2}\sigma_{12}^{2}\tau-\int_{0}^{S_{0}}\frac{P_1(k)}{k^{2}}dk-\int_{S_{0}}^{\infty}\frac{C_1(k)}{k^{2}}dk = -\frac{1}{2}\sigma_{12}^{2}\tau-\frac{T_{1}}{2}\mathrm{VS}_1, \\ 
%\Psi_{2}(K) -(\ln S_{0}-K) & \underset{K \to -\infty}{\longrightarrow} & -\int_{0}^{S_{0}}\frac{P_2(k)}{k^{2}}dk-\int_{S_{0}}^{\infty}\frac{C_2(k)}{k^{2}}dk=-\frac{T_{2}}{2}\mathrm{VS}_2,
%\end{eqnarray*}
%so that the difference tends to zero,
%\begin{eqnarray*}
%\Psi_{2}(K)-\Psi_{1}(K) & \underset{K \to -\infty}{\longrightarrow} & \frac{1}{2}\sigma_{12}^{2}\tau-\frac{1}{2}\left(T_{2}\mathrm{VS}_2-T_{1}\mathrm{VS}_1\right)=0.
%\end{eqnarray*}
%As a consequence, log-strikes $K$ such that $\Psi_1(K)>\Psi_2(K)$ cannot be found in the limit $K\rightarrow\pm\infty$.
%\end{rem}

%%%%%%%%%%%%%%%%%%%%%%%%%%%%%%%%%%%%%%%%%%%%%%%%%%%%%%%%%%%%%%%%%%%%%%%%%%%%%%%%%%%%%%

\section{A family of functionally generated portfolios}\label{sec:analytical_portfolios}

In this section, we introduce a new family of portfolios that sub/superreplicate the VIX. Their main merits are their simple functional form and that their sub/superreplication property is guaranteed by construction for all values of the underlying---in contrast to numerical solutions of the linear programming problems.
While these portfolios are not optimal in general, i.e., their prices do not attain $P_\mathrm{sub}$ and $P_\mathrm{super}$, they often improve the classical bounds from Section~\ref{sec:classicalBounds}, in particular the lower bound. In specific examples, they even turn out to be optimal, as we shall see in the subsequent sections.

Our portfolios are based on concave/convex payoffs of both the S\&P~500 and its logarithm. Let us start with superreplicating portfolios. For a convex function $\varphi:\mathbb{R}_{+}^{*}\times\mathbb{R}\rightarrow\mathbb{R}$
and $s_{1}>0$, we denote by
%\[
%U_{1+}^{\varphi}(s_{1})\equiv\sup_{v\ge0}\left\{ v-\varphi(s_{1},L(s_{1})+v^{2})\right\} 
%\]
\[
\varphi^*_\mathrm{super}(s_{1})\equiv\sup_{v\ge0}\left\{ v-\varphi(s_{1},L(s_{1})+v^{2})\right\} 
\]
the smallest function $u_{1}:\mathbb{R}_{+}^{*}\rightarrow\mathbb{R}\cup\{+\infty\}$
such that $u_{1}(s_{1})+\varphi(s_{1},L(s_{1})+v^{2})\ge v$ for all $s_{1}>0$ and $v\ge0$.
Moreover, we denote by $\partial_{i,r}\varphi$ the right 
derivative of $\varphi$ with respect to its $i$-th argument.

\begin{prop}
\label{prop:convex_portfolio} Let $\varphi:\mathbb{R}_{+}^{*}\times\mathbb{R}\rightarrow\mathbb{R}$
be convex and define
\begin{alignat}{2}
u_{1}(s_{1}) & =  \varphi^*_\mathrm{super}(s_{1}),      &\Delta^{S}(s_{1},v) & =  -\partial_{1,r}\varphi(s_{1},L(s_{1})+v^{2}),\label{eq:portfolio_convex}\\
u_{2}(s_{2})&=\varphi(s_{2},L(s_{2})),    \qquad &\Delta^{L}(s_{1},v)&=-\partial_{2,r}\varphi(s_{1},L(s_{1})+v^{2}).\nonumber 
\end{alignat}
Then the superreplication constraint (\ref{eq:constraint}) holds.
%, i.e., $(u_{1},u_{2},\Delta^{S},\Delta^{L})\in\cU_\mathrm{super}$.
\end{prop}

\begin{proof}
Since $\varphi$ is convex, 
\[
\varphi(s_{2},L(s_{2}))-\varphi(s_{1},L(s_{1})+v^{2})\ge\partial_{1,r}\varphi(s_{1},L(s_{1})+v^{2})(s_{2}-s_{1})+\partial_{2,r}\varphi(s_{1},L(s_{1})+v^{2})(L(s_{2})-L(s_{1})-v^{2})
\]
and as a consequence,
\begin{align*}
u_{1}&(s_{1})+u_{2}(s_{2})+\Delta^{S}(s_{1},v)(s_{2}-s_{1})+\Delta^{L}(s_{1},v)\left(L\left(\frac{s_{2}}{s_{1}}\right)-v^{2}\right)\\
&=u_{1}(s_{1})+\varphi(s_{2},L(s_{2}))-\partial_{1,r}\varphi(s_{1},L(s_{1})+v^{2})(s_{2}-s_{1})-\partial_{2,r}\varphi(s_{1},L(s_{1})+v^{2})(L(s_{2})-L(s_{1})-v^{2})\\
&\ge u_{1}(s_{1})+\varphi(s_{1},L(s_{1})+v^{2})\ge v. \qedhere
\end{align*}
\end{proof}

%
%\begin{rem}
%When $\varphi$ is differentiable, then $(\Delta^{S}(s_{1},v),\Delta^{L}(s_{1},v))$
%is (the negative of) the gradient of $\varphi$ at the point $(s_{1},L(s_{1})+v^{2})$.
%Otherwise, picking the right partial derivatives is just one particular
%admissible choice of $(\Delta^{S}(s_{1},v),\Delta^{L}(s_{1},v))$. Any lower-bounding
%hyperplane, ``tangent'' to $\varphi$ at point $(s_{1},L(s_{1})+v^{2})$,
%defines an acceptable $(\Delta^{S}(s_{1},v),\Delta^{L}(s_{1},v))$.
%\end{rem}

We remark that the use of the right derivative is not crucial; any other tangent will do as well.

A result similar to Proposition \ref{prop:convex_portfolio} holds for the subreplication of VIX futures. Given a concave function $\varphi:\mathbb{R}_{+}^{*}\times\mathbb{R}\rightarrow\mathbb{R}$, we set
\[
\varphi^*_\mathrm{sub}(s_{1})\equiv\inf_{v\ge0}\left\{ v-\varphi(s_{1},L(s_{1})+v^{2})\right\}.
\]

\begin{prop}
	\label{prop:concave_portfolio}
	Let $\varphi:\mathbb{R}_{+}^{*}\times\mathbb{R}\rightarrow\mathbb{R}$
	be concave and define
	\begin{alignat}{2}
	u_{1}(s_{1}) & =  \varphi^*_\mathrm{sub}(s_{1}),        &\Delta^{S}(s_{1},v) & =  -\partial_{1,r}\varphi(s_{1},L(s_{1})+v^{2}), \nonumber\\%\label{eq:portfolio_concave}\\
	u_{2}(s_{2})&=\varphi(s_{2},L(s_{2})),     \qquad&\Delta^{L}(s_{1},v)&=-\partial_{2,r}\varphi(s_{1},L(s_{1})+v^{2}).\nonumber 
	\end{alignat}
	Then the subreplication constraint (\ref{eq:constraint_sub}) holds.
\end{prop}

As mentioned in Section~\ref{sec:characterization_classical_is_optimal}, the space of convex functions $\varphi$ in two dimensions is still intractable. Thus, we specialize further to the form
\begin{equation}
\varphi(x,y)=\psi(ax+y)\label{eq:convex_function_1d}
\end{equation}
where $\psi:\R\to\R$ is a (one-dimensional) convex function and $a\in\R$. (Adding a constant in front of~$y$ does not increase the generality.) 

\begin{definition}\label{de:functionallyGen}
  Let $\psi:\R\to\R$ be convex (concave) and $a\in\R$. The portfolio defined by Proposition~\ref{prop:convex_portfolio} (Proposition~\ref{prop:concave_portfolio}) based on $\varphi(x,y)=\psi(ax+y)$ is called the \emph{superreplicating (subreplicating) portfolio generated by $\psi$ and $a$}.
\end{definition}

We call these portfolios \emph{functionally generated} because they are determined by a single real function and a constant. As convex functions on $\R$ are well approximated by linear combinations of call payoffs, is is easy to search numerically over a representative subset of this class. Of course, any functionally generated portfolio is superreplicating as a special case of Proposition~\ref{prop:convex_portfolio}, and the analogue holds for subreplication using concave functions.

The classical superreplication portfolio~(\ref{eq:portfolio_superrep_classical}) corresponds to the particular
case where $\varphi(x,y)=by$; that is, $\psi(z)=bz$ is a linear function and $a=0$, so that $\varphi$ does not depend on the first variable $x$. Indeed, in this case, we have
$u_{2}(s_{2})=bL(s_{2})$, $\Delta^{S}(s_{1},v)=0$, and $\Delta^{L}(s_{1},v)=-b$.
For $\varphi^*_\mathrm{super}(s_{1})$ to be finite, one must choose $b>0$,
in which case $\varphi^*_\mathrm{super}(s_{1})=\frac{1}{4b}-bL(s_{1})\equiv u_{1}(s_{1})$.
Minimizing 
$$
  \mathbb{E}^{1}\left[u_{1}(S_{1})\right]+\mathbb{E}^{2}[u_{2}(S_{2})]=\frac{1}{4b}-b\mathbb{E}^{1}\left[L(S_{1})\right]+b\mathbb{E}^{2}[L(S_{2})]=\frac{1}{4b}+b\sigma_{12}^{2}
$$
over the parameter $b$ yields $b=\frac{1}{2\sigma_{12}}$ and we recover
(\ref{eq:portfolio_superrep_classical}).

Our portfolios are more general in that we consider a convex function of $(s_{2},L(s_{2}))$ rather than a linear function of $L(s_{2})$. We remark that it is meaningless to consider functions $\varphi$ of the first variable alone:
for $\varphi^*_\mathrm{super}(s_{1})=\sup_{v\ge0} \{ v-\varphi(s_{1},L(s_{1})+v^{2})\}$ to be finite, $\varphi$ must depend
on the second variable and in fact be unbounded.

Let $\mathcal{F}_{\mathrm{cvx}}(\mathbb{R}_{+}^{*}\times\mathbb{R})$ and $\mathcal{F}_{\mathrm{cvx}}(\mathbb{R})$ be the sets of all convex functions on $\mathbb{R}_{+}^{*}\times\mathbb{R}$ and $\R$, respectively. 
The two families of superreplicating portfolios considered above correspond to the price bounds
	\begin{align}
	P_\mathrm{super}^{\mathrm{cvx}} & \equiv  \inf_{\varphi\in\mathcal{F}_{\mathrm{cvx}}(\mathbb{R}_{+}^{*}\times\mathbb{R})}\left\{ \mathbb{E}^{1}\left[\sup_{v\ge0}\left\{ v-\varphi(S_{1},L(S_{1})+v^{2})\right\}\right]+\mathbb{E}^{2}[\varphi(S_{2},L(S_{2}))]\right\}, \nonumber \\ 
	P_\mathrm{super}^{\mathrm{cvx},1}&\equiv\inf_{\psi\in\mathcal{\mathcal{F}_{\mathrm{cvx}}}(\mathbb{R}),a\in\mathbb{R}}\left\{ \mathbb{E}^{1}\left[\sup_{v\ge0}\left\{ v-\psi(aS_{1}+L(S_{1})+v^{2})\right\} \right]+\mathbb{E}^{2}[\psi(aS_{2}+L(S_{2}))]\right\} \label{eq:cvx1def}
\end{align}
which satisfy 
\[
	P_\mathrm{super}^{\mathrm{cvx},1}\ge P_\mathrm{super}^{\mathrm{cvx}}\ge P_\mathrm{super};
\]
the expectation of a non-integrable function is read as $+\infty$ in the above formulas.
For the analogous definitions in the subreplication problem, exchanging convex/concave as well as inf/sup, we have $P_\mathrm{sub}^{\mathrm{ccv},1}\le P_\mathrm{sub}^{\mathrm{ccv}}\le P_\mathrm{sub}$.

\medskip

The following result shows that functionally generated portfolios improve the classical subreplication bound $P_\mathrm{sub}\geq0$ in all relevant cases. While we already know from the abstract result in Theorem~\ref{thm:D=0} that $P_\mathrm{sub}>0$ when $\mu_1 \neq \mu_2$, we now construct  an explicit, functionally generated subreplicating portfolio that has strictly positive price.

\begin{prop}
\label{prop:P>0} Let $\mu_1 \neq \mu_2$ be in convex order. 
Then there exists a functionally generated subreplicating portfolio $(u_{1},u_{2},\Delta^{S},\Delta^{L})\in\cU_\mathrm{sub}$ with strictly positive price.

More precisely, it is generated by the concave function $\psi(z)\equiv -\gamma (z+b)_{-}$ and a constant $a>0$, where $\gamma>0$ and $b\in\R$. The values of the constants depend on $\mu_{1},\mu_{2}$ and can be found explicitly as indicated in the proof.
\end{prop}

\begin{proof}
We consider the concave function $\psi:\mathbb{R}\rightarrow\mathbb{R}$ defined by
\[
\psi(z)\equiv \gamma \min(z+b,0)
\]
where $\gamma>0$ and $b\in\RR$. Moreover, let $a>0$. Then, denoting
\begin{equation}\label{eq:LambdaAB}
\Lambda_{a,b}(s) \equiv L(s) + as + b,
\end{equation}
we have
\[
u_{2}(s_{2})=\varphi(s_{2},L(s_{2}))=\psi(as_2+L(s_2))=\gamma\min(\Lambda_{a,b}(s_{2}),0) = -\gamma\Lambda_{a,b}(s_{2})_{-} \;.
\]
Since
$
\varphi(s_{1},L(s_{1})+v^{2})=\gamma\min(v^{2}+\Lambda_{a, b}(s_{1}),0),
$
the infimum of $v-\varphi(s_{1},L(s_{1})+v^{2})$ over $v\ge 0$ can only be attained for $v=0$ or $v^{2}=-\Lambda_{a, b}(s_{1})$, whence
%\begin{eqnarray*}
%\varphi^*_\mathrm{sub}(s_{1}) = \begin{cases}
%\min\left(-\gamma\Lambda_{a,b}(s_{1}),\sqrt{-\Lambda_{a,b}(s_{1})}\right) & \mathrm{if\,\,}\Lambda_{a,b}(s_{1})\le0\\
%0 & \mathrm{if\,\,}\Lambda_{a,b}(s_{1})\ge0 
%\end{cases}
%\qquad= \min\left(\gamma\Lambda_{a,b}(s_{1})_{-},\sqrt{\Lambda_{a,b}(s_{1})_{-}}\right).
%\end{eqnarray*}
$$
\varphi^*_\mathrm{sub}(s_{1}) = 
\min\left(-\gamma\Lambda_{a,b}(s_{1}),\sqrt{-\Lambda_{a,b}(s_{1})}\right)\mathbf{1}_{\Lambda_{a,b}(s_{1})\le0} = \min\left(\gamma\Lambda_{a,b}(s_{1})_{-},\sqrt{\Lambda_{a,b}(s_{1})_{-}}\right).
$$

\begin{figure}
\begin{center}
\includegraphics[trim={3cm 3cm 3cm 3cm},clip,width=11cm,height=6cm]{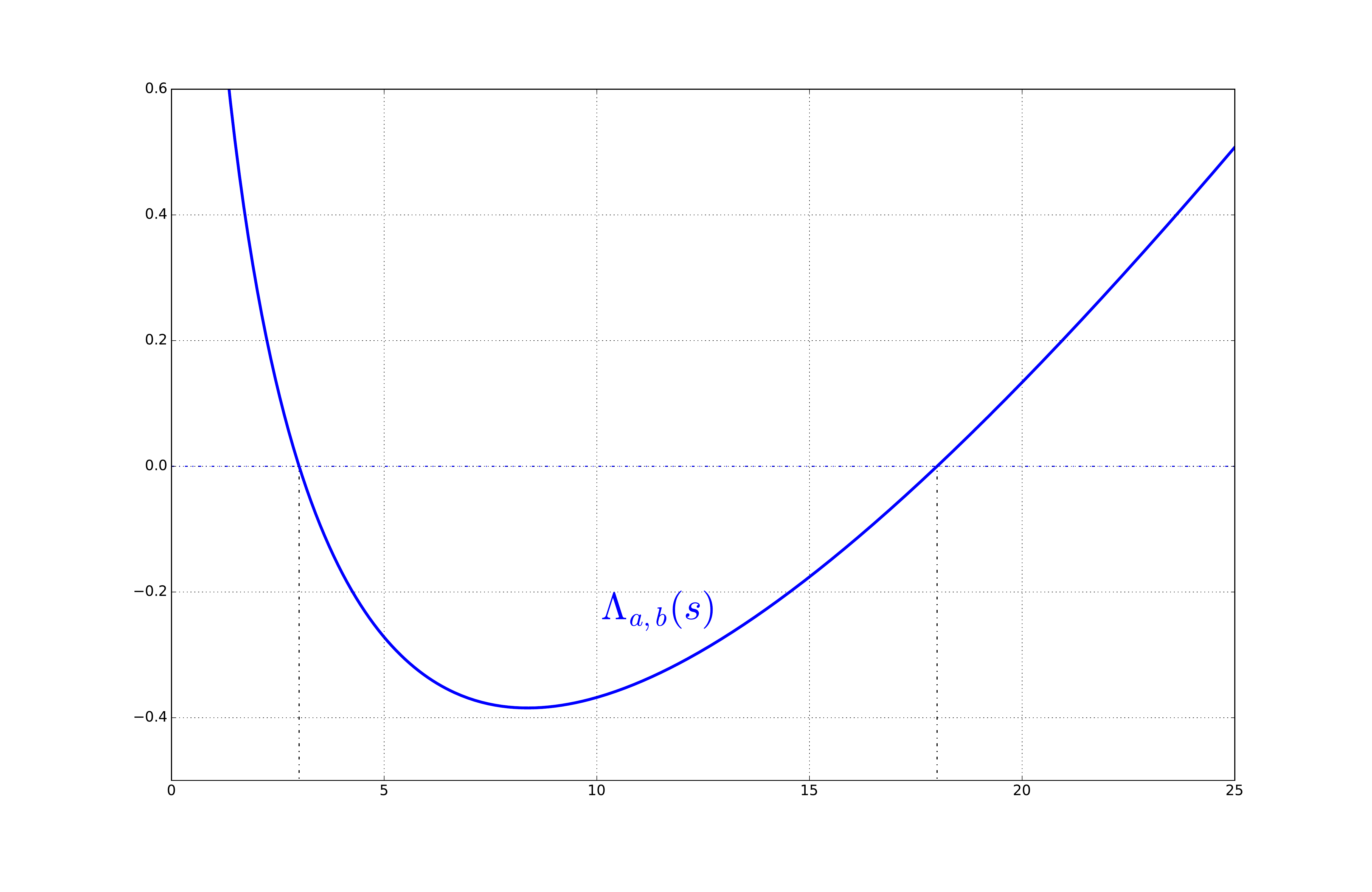}
\end{center}
\caption{Graph of $\Lambda_{a,b}$ for $a>0$}
\label{fig:lambda_ab}
\end{figure}

As $a>0$, the function $\Lambda_{a,b}$ is bounded from below on $\mathbb{R}_{+}^{*}$ (see Figure \ref{fig:lambda_ab}). For $b<-\frac{2}{\tau}-L\left(\frac{2}{a\tau}\right)$, the minimum of $\Lambda_{a,b}$ is strictly negative,
$$
  M_{a,b}\equiv \max (\Lambda_{a,b})_{-} >0.
$$
The function $\Lambda_{a,b}$ then has two distinct zeros $n_{1}<n_{2}$ and we observe that $n_{1}\to0$ as $b\to -\infty$ and $n_{2}\to\infty$ as $a\downarrow 0$. 
Let $0<\gamma\leq\frac{1}{\sqrt{M_{a,b}}}$. Then $\gamma\Lambda_{a,b}(s_{1})_{-}\le\sqrt{\Lambda_{a,b}(s_{1})_{-}}$ and hence
$$
  u_{1}(s_{1})\equiv \varphi^*_\mathrm{sub}(s_{1}) =\gamma\Lambda_{a,b}(s_{1})_{-}\;.
$$

Next, we show that there exist parameters $a,b$ such that 
\begin{equation}\label{eq:posDiffClaim}
 \int_{n_{1}}^{n_{2}} \Lambda_{a,b}(s)\, (\mu_{2} -\mu_{1})(ds)>0.
\end{equation}
Since $\mu_{1}\neq\mu_{2}$ are in convex order and $L$ is strictly convex,
$$
  \delta \equiv \int \Lambda_{a,b}(s) \, (\mu_{2} -\mu_{1})(ds) = \int L(s) \, (\mu_{2} -\mu_{1})(ds) >0
$$
and this value is independent of $a,b$. Let $\eps\equiv\delta/3$. In view of Assumption~\ref{as:integrability}, choosing $b$ small enough guarantees that $\int \mathbf{1}_{(0,n_{1}]}\Lambda_{a,b} \, d\mu_{i} \leq \eps$ and then choosing also $a>0$ small enough yields  $\int \mathbf{1}_{[n_{2},\infty)}\Lambda_{a,b} \, d\mu_{i} \leq \eps$, for $i\in\{1,2\}$. As a result, we have 
$$
 \int_{n_{1}}^{n_{2}} \Lambda_{a,b}(s)\, (\mu_{2} -\mu_{1})(ds)\ge\delta/3
$$
which proves our claim~\eqref{eq:posDiffClaim}. 
With this choice of $a>0$ and $b<-\frac{2}{\tau}-L\left(\frac{2}{a\tau}\right)$, we have 
\[
\mathbb{E}^{1}\left[u_{1}(S_{1})\right]+\mathbb{E}^{2}\left[u_{2}(S_{2})\right]
 =  \gamma\left(\mathbb{E}^{1}\left[\Lambda_{a,b}(S_{1})_{-}\right]-\mathbb{E}^{2}\left[\Lambda_{a,b}(S_{2})_{-}\right]\right)\\
=\gamma \int_{n_{1}}^{n_{2}} \Lambda_{a,b}(s)_{-}(\mu_{1} -\mu_{2})(ds) > 0.
\]
Optimizing over the parameter $\gamma$ leads us to the choice $\gamma=\frac{1}{\sqrt{M_{a,b}}}$. Summarizing, the portfolio
$$
  u_{1}(s_{1})=\frac{1}{\sqrt{M_{a,b}}} \Lambda_{a,b}(s_{1})_{-},\quad u_{2}(s_{2})=-\frac{1}{\sqrt{M_{a,b}}} \Lambda_{a,b}(s_{2})_{-}
$$
with deltas as in Proposition~\ref{prop:concave_portfolio}, is subreplicating at price
$
   \frac{1}{\sqrt{M_{a,b}}} \left(\mathbb{E}^{2}\left[\Lambda_{a,b}(S_{2})\right]-\mathbb{E}^{1}\left[\Lambda_{a,b}(S_{1})\right]\right) >0.
$
\end{proof}

%NOTE: actually, fixing any $a>0$ and then choosing $b$ small enough will work as well.

\begin{rem}\label{rk:compactSupport} In many important cases it is straightforward to find $a,b$ satisfying~\eqref{eq:posDiffClaim}. Indeed, suppose that $\mu_{1},\mu_{2}$ have continuous densities $f_{1},f_{2}$. Then, is suffices to choose $a,b$ such that $[n_{1},n_{2}]\subset \{f_{1}>f_{2}\}$. Or, if~$\mu_2$ (and hence~$\mu_{1}$) is concentrated on a compact interval $I\subset \R_{+}^{*}$, then we can choose $a,b$ such that $I\subset [n_{1},n_{2}]$.
%By the strict convexity of $\Lambda_{a,b}$, this simple choice leads to a positive price
%\[
%\mathbb{E}^{1}\left[u_{1}(S_{1})\right]+\mathbb{E}^{2}\left[u_{2}(S_{2})\right]
%= \gamma\left(\mathbb{E}^{1}\left[\Lambda_{a,b}(S_{1})_{-}\right]-\mathbb{E}^{2}\left[\Lambda_{a,b}(S_{2})_{-}\right]\right)
%= \gamma\left(\mathbb{E}^{2}\left[\Lambda_{a,b}(S_{2})\right]-\mathbb{E}^{1}\left[\Lambda_{a,b}(S_{1})\right]\right) >0
%\]
%which is again maximized by $\gamma=\frac{1}{\sqrt{M_{a,b}}}$.
\end{rem}

%%%%%%%%%%%%%%%%%%%%%%%%%%%%%%%%%%%%%%%%%%%%%%%%%%%%%%%%%%%%%%%%%%%%%%%%%%%%%%%%%
\section{The case where $\mu_2$ is a Bernoulli distribution}\label{sec:binomial}

In this section, we study in detail an example where the
classical upper bound $\sigma_{12}$ is typically not optimal, i.e.,
$P_\mathrm{super}<\sigma_{12}$. We will explicitly compute the
optimal bound $P_\mathrm{super}$ and derive ($\eps$-)optimal superreplicating portfolios within the functionally generated class. In all of this section, $\mu_{2}$ is a Bernoulli distribution,
\begin{equation}
\mu_{2}=p\delta_{s_{2}^{u}}+(1-p)\delta_{s_{2}^{d}},\qquad 0< s_{2}^{d}<S_0<s_{2}^{u},\qquad p = \frac{S_0-s_{2}^{d}}{s_{2}^{u}-s_{2}^{d}}\in(0,1).\label{eq:mu2_bernoulli}
\end{equation}
Thus, $S_{2}$
can only take the two values $s_{2}^{d}<s_{2}^{u}$ and these are the only free parameters---the fact that $\mu_2$ has mean $S_0$ determines the value of $p$. 

In our first result, we show that in the absence of arbitrage, the sets $\mathcal{M}(\mu_{1},\mu_{2})$ and $\mathcal{M}_V(\mu_{1},\mu_{2})$ both have a unique element. As a consequence, $D_\mathrm{super}=D_\mathrm{sub}$ and this number can be computed explicitly as the expectation under the unique risk-neutral measure. Some notation is needed to state this result. We define the functions 
\begin{equation*}
\pi_{u}(s_{1})\equiv\frac{s_{1}-s_{2}^{d}}{s_{2}^{u}-s_{2}^{d}},\qquad\pi_{d}(s_{1})\equiv\frac{s_{2}^{u}-s_{1}}{s_{2}^{u}-s_{2}^{d}}=1-\pi_{u}(s_{1})\label{eq:piu_pid}
\end{equation*}
which will represent the unique martingale transition probabilities. 
Moreover, an important role will be played by the function
\begin{equation}
\Lambda_b(s_{1})\equiv \Lambda_{s_2^d,s_2^u}(s_{1})\equiv \pi_{u}(s_{1})L\left(\frac{s_{2}^{u}}{s_{1}}\right)+\pi_{d}(s_{1})L\left(\frac{s_{2}^{d}}{s_{1}}\right),\qquad s_1>0. \label{eq:def_lambda_b}
\end{equation}
%Note that, from the definition (\ref{eq:def_Lambda}), $\Lambda_b(S_1)=\Lambda_\beta(S_1)$ $\mu_1$-a.s.
Note that the function $\Lambda_b$ depends only on the two values $s_{2}^{d}$, $s_{2}^{u}$. When $s_{1}\in[s_{2}^{d},s_{2}^{u}]$, $\Lambda_b(s_{1})$ is the risk-neutral price of the FSLC in the one-step binomial model (hence the subscript $b$), given the price $s_1$ of the S\&P~500 index at time $T_1$. In particular, in this model, one can replicate the log-contract payoff at $T_2$ by holding $\Lambda_b(s_{1})$ units of cash at $T_1$ and $\Delta_b$ units of the index over $[T_1,T_2]$, where $\Delta_b$ is the delta of the log-contract in the one-step binomial model:
\begin{eqnarray}
\forall s_1 > 0,\;\;\forall s_2 \in \{s_{2}^{d},s_{2}^{u}\},\qquad L\left(\frac{s_2}{s_1}\right) & = & \Lambda_b(s_1) + \Delta_b (s_2 - s_1) ,\label{eq:replication_binomial} \\
\Delta_b & \equiv & \frac{L(s_2^u)-L(s_2^d)}{s_2^u -s_2^d} \; < \; 0.  \label{eq:def_Delta_b}
\end{eqnarray}
We observe that $\Lambda_b(s_{2}^{d})=\Lambda_b(s_{2}^{u})=0$. Moreover, from (\ref{eq:replication_binomial}), the function $\Lambda_b$ is strictly concave. As a consequence, $\Lambda_b(s_{1})>0$ for all
$s_{1}\in(s_{2}^{d},s_{2}^{u})$, and $\Lambda_b(s_{1})<0$ for all
$s_{1}\notin[s_{2}^{d},s_{2}^{u}]$. We denote by $\lambda_b$ the square root of $\Lambda_b$ on the interval $[s_{2}^{d},s_{2}^{u}]$,
\begin{equation*}\label{eq:lambdaDef}
\lambda_b(s_1) = \sqrt{\Lambda_b(s_1)},\qquad s_{1}\in[s_{2}^{d},s_{2}^{u}], \qquad\mbox{and}\qquad \bar{\lambda}_b \equiv \max_{[s_2^d,s_2^u]}\lambda_b.
\end{equation*}
For any $\sigma\in[0,\bar{\lambda}_b)$, the equation $\lambda_b(s_1) = \sigma$ is equivalent to the equation $\Lambda_b(s_1) = \sigma^2$ which has exactly two solutions $s_1$ that we denote by
%\begin{equation}\label{eq:lambdaDefSol}
%  \mbox{$\lambda_{b,-}^{-1}(\sigma) < \lambda_{b,+}^{-1}(\sigma)$ denote the two solutions of the equation} \quad \lambda_b(s_1) = \sigma,\quad s_1\in [s_2^d,s_2^u]
%\end{equation}
\begin{equation}\label{eq:lambdaDefSol}
  \lambda_{b,-}^{-1}(\sigma) < \lambda_{b,+}^{-1}(\sigma).
\end{equation}
For convenience, we write $\lambda_{b,-}^{-1}(\bar{\lambda}_b)=\lambda_{b,+}^{-1}(\bar{\lambda}_b)$ for the unique solution of $\lambda_b(s_1) = \bar{\lambda}_b$. Figure~\ref{fig:lambda_bin} shows the graphs of $\Lambda_b$ and $\lambda_b$.

\begin{figure}
%\title{$\lambda_b$}
\begin{center}
\includegraphics[trim={3cm 3cm 3cm 3cm},clip,width=11cm,height=6cm]{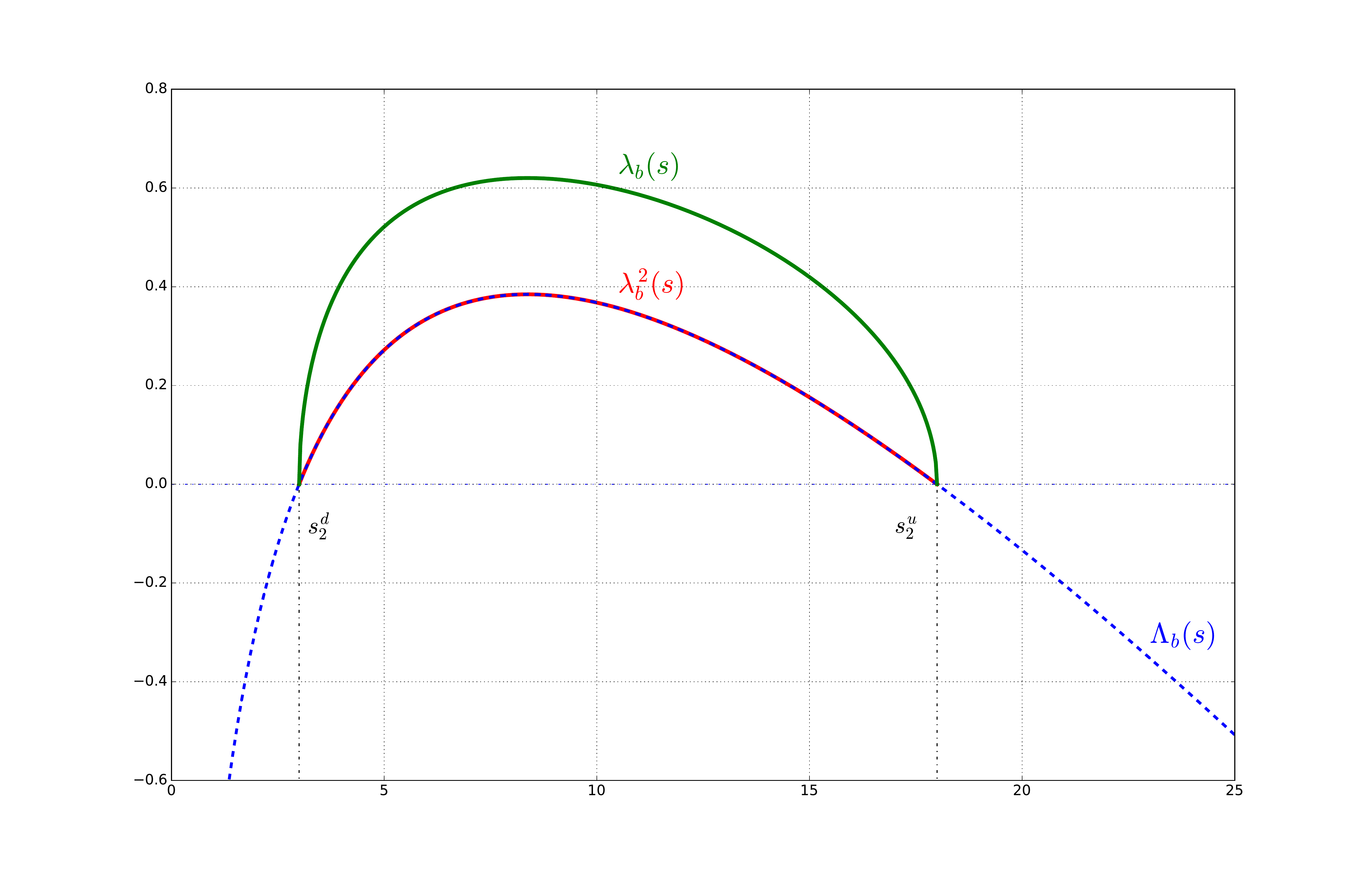}
\end{center}
\caption{$\Lambda_b$ and its square-root, $\lambda_b$}
\label{fig:lambda_bin}
\end{figure}

\begin{thm}\label{thm:3cases}
	Let $\mu_{2}$ be the Bernoulli distribution (\ref{eq:mu2_bernoulli}). Then, there is no arbitrage, or equivalently $\mathcal{M}(\mu_{1},\mu_{2})\neq\emptyset$, if and only if $\mathrm{{supp}(\mu_{1})}\subset[s_{2}^{d},s_{2}^{u}]$. In this case, $\mathcal{M}(\mu_{1},\mu_{2})$ has a unique element
	$\beta$, given by
	\begin{eqnarray}
	\beta(ds_{1},ds_{2}) & = & \mu_{1}(ds_{1})\left(\pi_{u}(s_{1})\delta_{s_{2}^{u}}(ds_{2})+\pi_{d}(s_{1})\delta_{s_{2}^{d}}(ds_{2})\right)\label{eq:beta}
	\end{eqnarray}
	and $\mathcal{M}_{V}(\mu_{1},\mu_{2})$ has a unique element
	$\beta_{\Lambda}$, given by
	\[
	\beta_{\Lambda}(ds_{1},ds_{2},dv)=\beta(ds_{1},ds_{2})\delta_{\lambda_b(s_{1})}(dv).
	\]
	In particular, $V=\lambda_b(S_{1})$ $\beta_{\Lambda}$-a.s. Moreover, 
	\begin{enumerate}
	\item[(i)] if $\mu_1$ is the Bernoulli distribution that takes values in $\{\lambda_{b,-}^{-1}(\sigma),\lambda_{b,+}^{-1}(\sigma)\}$ for some $\sigma\in[0,\lambda_b(S_0)]$, then $\beta\in\bar{\mathcal{M}}(\mu_{1},\mu_{2})$ and $P_\mathrm{super}=\sigma_{12}$,
	\item[(ii)] if $\mu_{1}$ is a different distribution, then $\bar{\mathcal{M}}(\mu_{1},\mu_{2})=\emptyset$ and $P_\mathrm{super}=\mathbb{E}^{1}[\lambda_b(S_{1})]<\sigma_{12}$.
	\end{enumerate}
\end{thm}

\begin{proof}
We first characterize the absence of arbitrage. Let $\beta\in\mathcal{M}(\mu_{1},\mu_{2})$. Then, $\mu_{1}$ and $\mu_{2}$ are in convex order and in particular $\mathrm{{supp}(\mu_{1})}\subset[s_{2}^{d},s_{2}^{u}]$. 
Since $\mathbb{E}^{\beta}\left[S_{2}|S_{1}\right]=S_{1}$, we have
\[
\beta(S_{2}=s_{2}^{u}|S_{1})s_{2}^{u}+(1-\beta(S_{2}=s_{2}^{u}|S_{1}))s_{2}^{d}=S_{1}.
\]
That is, the transition probabilities satisfy $\beta(S_{2}=s_{2}^{u}|S_{1}) = \pi_{u}(S_{1})$ and $\beta(S_{2}=s_{2}^{d}|S_{1})=\pi_{d}(S_{1})$, and as a consequence, $\beta$ is uniquely determined and given by~(\ref{eq:beta}). 
Conversely, assume that $\mathrm{{supp}(\mu_{1})}\subset[s_{2}^{d},s_{2}^{u}]$ and let $\beta$
be the probability measure on $(\mathbb{R}_{+}^*)^{2}$ defined by~(\ref{eq:beta}). We readily verify that $\beta\in\mathcal{M}(\mu_{1},\mu_{2})$ and in particular $\mathcal{M}(\mu_{1},\mu_{2})\neq\emptyset$. The latter is equivalent to the absence of arbitrage; cf.\ Theorem~\ref{thm:arbitrage-free}.

Next, suppose that $\mathcal{M}(\mu_{1},\mu_{2})\neq\emptyset$ and let $\beta$ be its element. We have $\beta_{\Lambda}\in\mathcal{M}_{V}(\mu_{1},\mu_{2})$; cf.\ Lemma \ref{lem:projection_MV_M}(ii).
Let $\mu\in\mathcal{M}_{V}(\mu_{1},\mu_{2})$; we prove that $\mu=\beta_{\Lambda}$.
The martingale condition $\mathbb{E}^{\mu}\left[S_{2}|S_{1},V\right]=S_{1}$
implies that
\[
\mu(S_{2}=s_{2}^{u}|S_{1},V)s_{2}^{u}+(1-\mu(S_{2}=s_{2}^{u}|S_{1},V))s_{2}^{d}=S_{1},
\]
hence the transition probabilities
\[
\mu(S_{2}=s_{2}^{u}|S_{1},V)=\pi_{u}(S_{1})=\beta(S_{2}=s_{2}^{u}|S_{1}),\qquad\mu(S_{2}=s_{2}^{d}|S_{1},V)=\pi_{d}(S_{1})=\beta(S_{2}=s_{2}^{d}|S_{1})
\]
do not depend on $V$, and since the first marginal of both $\mu$
and $\beta$ is $\mu_{1}$, the projection of $\mu$ onto the first two
coordinates is equal to $\beta$. Moreover, $\mu$-a.s.,
\begin{multline*}
V^{2} = \mathbb{E}^{\mu}\left[L\left(\frac{S_{2}}{S_{1}}\right)\middle|S_{1},V\right] = \mu(S_{2}=s_{2}^{u}|S_{1},V)L\left(\frac{s_{2}^{u}}{S_{1}}\right)+\mu(S_{2}=s_{2}^{d}|S_{1},V)L\left(\frac{s_{2}^{d}}{S_{1}}\right) \\
 = \beta(S_{2}=s_{2}^{u}|S_{1})L\left(\frac{s_{2}^{u}}{S_{1}}\right)+\beta(S_{2}=s_{2}^{d}|S_{1})L\left(\frac{s_{2}^{d}}{S_{1}}\right)
 = \mathbb{E}^{\beta}\left[L\left(\frac{S_{2}}{S_{1}}\right)\middle|S_{1}\right]
 =  \lambda_b(S_{1})^{2}.
\end{multline*}
Therefore, $V=\lambda_b(S_{1})$ $\mu$-a.s., showing that $\mu=\beta_{\Lambda}$ is the unique element of $\mathcal{M}_{V}(\mu_{1},\mu_{2})$.

(i) By the definition of $\bar{\mathcal{M}}(\mu_{1},\mu_{2})$, we have $\beta\in\bar{\mathcal{M}}(\mu_{1},\mu_{2})$ if and only if $\Lambda_\beta(S_1)$ is $\mu_1$-a.s.\ constant. Due to the strict concavity of $\Lambda_\beta = \Lambda_b$, this happens only if $\mu_1$ is atomic with at most two atoms. If $\mu_{1}$ is a Dirac mass, since it has mean $S_0$, it can only be $\delta_{S_0}$ (which we consider a special case of the Bernoulli distribution). If $\mu_{1}$ has two atoms at $s_1^d < s_1^u$ then we must have $\Lambda_b(s_1^d) = \Lambda_b(s_1^u)$, i.e., $\lambda_b(s_1^d) = \lambda_b(s_1^u)$, and $s_1^d < S_0 < s_1^u$. As a consequence, there exists $\sigma\in[0,\lambda_b(S_0))$ such that $s_1^d=\lambda_{b,-}^{-1}(\sigma)$ and $s_1^u=\lambda_{b,+}^{-1}(\sigma)$. We recall from Theorem \ref{thm:D=Dbar} that $\bar{\mathcal{M}}(\mu_{1},\mu_{2})\neq\emptyset$ implies $P_\mathrm{super}=\sigma_{12}$.
  
% OLD PROOF, CORRECTED AS PER THE REFEREE'S RECOMMENDATION  
%  (i) By the definition of $\bar{\mathcal{M}}(\mu_{1},\mu_{2})$, we have $\beta\in\bar{\mathcal{M}}(\mu_{1},\mu_{2})$ if and only if $\Lambda_\beta(S_1)$ is $\mu_1$-a.s.\ constant, which is equivalent to saying that $\lambda_b(S_1)$ is $\mu_1$-a.s.\ constant. Due to the strict concavity of $\lambda_b$, this happens only if $\mu_1$ is atomic with at most two atoms. If $\mu_{1}$ is a Dirac mass, since it has mean $S_0$, it can only be $\delta_{S_0}$ (which we consider a special case of the Bernoulli distribution). If $\mu_{1}$ has two atoms at $s_1^d < s_1^u$ then we must have $\lambda_b(s_1^d) = \lambda_b(s_1^u)$ and $s_1^d < S_0 < s_1^u$. As a consequence, there exists $\sigma\in[0,\lambda_b(S_0))$ such that $s_1^d=\lambda_{b,-}^{-1}(\sigma)$ and $s_1^u=\lambda_{b,+}^{-1}(\sigma)$. We recall from Theorem \ref{thm:D=Dbar} that $\bar{\mathcal{M}}(\mu_{1},\mu_{2})\neq\emptyset$ implies $P_\mathrm{super}=\sigma_{12}$.
 
 (ii) If we are not in the case (i), then $\bar{\mathcal{M}}(\mu_{1},\mu_{2})=\emptyset$, so Theorem~\ref{thm:D=Dbar} implies that $P_\mathrm{super}<\sigma_{12}$. Moreover,
\[
D_\mathrm{super}\equiv\sup_{\mu\in\mathcal{M}_{V}(\mu_{1},\mu_{2})}\mathbb{E}^{\mu}[V]=\mathbb{E}^{\beta_{\Lambda}}[V]=\mathbb{E}^{\beta_{\Lambda}}[\lambda_b(S_{1})]=\mathbb{E}^{1}[\lambda_b(S_{1})]
\]
and then the claim follows as $P_\mathrm{super}=D_\mathrm{super}$ by Theorem \ref{th:strongDuality}.
\end{proof}
 
\medskip

% ------------------------------------------------------------------------------------------------ %
Next, we derive an explicit superreplicating portfolio with $\eps$-optimal price. It turns out that such a portfolio can be chosen of the functionally generated form~(\ref{eq:portfolio_convex}), (\ref{eq:convex_function_1d}). To that end, we first observe that $\Lambda_b(s_1)$ is of the form $-\Lambda_{a,b}$ as defined in~\eqref{eq:LambdaAB},
\begin{equation}
\Lambda_b(s_1) = -L(s_1)-as_1-b = -\Lambda_{a,b}(s_1) \label{eq:form_Lambda_b}
\end{equation}
with $a=-\Delta_b>0$ and $b=-L(s_{2}^{d}) + \Delta_b s_{2}^{d} = -L(s_{2}^{u}) + \Delta_b s_{2}^{u} $. Indeed, from (\ref{eq:replication_binomial}),
\begin{equation*}
\forall s_1 > 0, \qquad \Lambda_b(s_{1})= L(s_{2}^{d})-L(s_{1}) - \Delta_b (s_{2}^{d}-s_1) = L(s_{2}^{u})-L(s_{1}) - \Delta_b (s_{2}^{u} - s_1). \label{eq:Phi_delta_log}
\end{equation*}
As a consequence, we have
\begin{equation}
\forall s_1,s_2>0,\qquad \Lambda_b(s_{1})-\Lambda_b(s_{2})= L\left(\frac{s_2}{s_{1}}\right) + \Delta_b (s_{1}-s_2). \label{eq:Phi_s1_minus_Phi_s2}
\end{equation}

Let $0<\varepsilon< \bar{\lambda}_b$, $s_{1,\varepsilon}^d=\lambda_{b,-}^{-1}(\varepsilon)$, and $s_{1,\varepsilon}^u=\lambda_{b,+}^{-1}(\varepsilon)$, i.e., $s_{1,\varepsilon}^d < s_{1,\varepsilon}^u$ are the two values such that $\lambda_b(s_{1,\varepsilon}^d) = \lambda_b(s_{1,\varepsilon}^u) = \varepsilon$. Note that $s_2^d < s_{1,\varepsilon}^d < s_{1,\varepsilon}^u < s_2^u$ and define
\begin{equation*}
\psi(z)\equiv\frac{1}{2\varepsilon}(z+b)_{+}. \label{eq:psi_superrep_binomiale}
\end{equation*}
Using this function as a generator as in Proposition~\ref{prop:convex_portfolio}, we have
\begin{eqnarray*}
u_{2}(s_{2}) = \varphi(s_{2},L(s_{2}))=\psi(as_{2}+L(s_{2}))=\frac{1}{2\varepsilon}(as_{2}+L(s_{2})+b)_{+}=\frac{1}{2\varepsilon}(\Lambda_{a,b}(s_2))_{+}=\left(-\frac{1}{2\varepsilon}\Lambda_{b}(s_{2})\right)_{+}.
\end{eqnarray*}
Notice that $u_2 = 0$ on the support of $\mu_2$, so that the payoff $u_2(S_2)$ is \textit{free} at time 0. Moreover, one can obtain the wealth $\varphi(S_{1},L(S_{1})+V^{2})$ at $T_1$ for zero initial cost; cf.\ the proof of Proposition~\ref{prop:convex_portfolio}.
We now complete the portfolio using the procedure detailed in Proposition \ref{prop:convex_portfolio}: as
\[
\varphi(s_{1},L(s_{1})+v^{2})=\frac{1}{2\varepsilon}(as_{1}+L(s_{1})+b+v^{2})_{+}=\frac{1}{2\varepsilon}\left(v^{2}-\Lambda_{b}(s_{1})\right)_{+},
\]
we have
\[
\varphi^*_\mathrm{super}(s_{1})=\sup_{v\ge0}\left\{ v-\varphi(s_{1},L(s_{1})+v^{2})\right\} =\sup_{v\ge0}\left\{ v-\frac{1}{2\varepsilon}\left(v^{2}-\lambda_{b}^{2}(s_{1})\right)_{+}\right\}.
\]
%\[
%v-\varphi(s_{1},L(s_{1})+v^{2})=v-\frac{1}{2\varepsilon}\left(v^{2}-\lambda_{b}^{2}(s_{1})\right)_{+}
%\]
If $s_{1}\in[s_{1,\varepsilon}^{d},s_{1,\varepsilon}^{u}]$, then $\lambda_{b}(s_{1})\ge\varepsilon$ and  the above supremum is attained for $v=\lambda_b (s_1)$ and $\varphi^*_\mathrm{super}(s_{1}) =\lambda_{b}(s_{1})$.
If $s_{1}\notin[s_{1,\varepsilon}^{d},s_{1,\varepsilon}^{u}]$, then $\lambda_{b}(s_{1})<\varepsilon$,
therefore the supremum is attained for $v=\varepsilon$ and $\varphi^*_\mathrm{super}(s_{1})\le\varepsilon$. As a consequence, the portfolio defined by
\begin{alignat}{2}
u_1(s_1) &= \begin{cases}
    \lambda_b(s_1) &\quad \text{if } s_1 \in [s_{1,\varepsilon}^d, s_{1,\varepsilon}^u]\\
    \varepsilon &\quad \text{if } s_1 \notin [s_{1,\varepsilon}^d, s_{1,\varepsilon}^u],\\
    \end{cases}
    &\qquad\quad
    \Delta^{S}(s_{1},v) &= -\partial_{1,r}\varphi(s_{1},L(s_{1})+v^{2})=\frac{\Delta_{b}}{2\varepsilon}\mathbf{1}_{v\ge\lambda_{b}(s_{1})},  \label{eq:Xi_epsilon} \\
    u_2(s_2) &= \begin{cases}
    0 &\quad \text{if } s_2 \in [s_2^d, s_2^u]\\
    -\frac{1}{2\varepsilon}\Lambda_b(s_2)&\quad \text{if } s_2 \notin [s_2^d, s_2^u],\\
    \end{cases}
    &\qquad\quad
    \Delta^{L}(s_{1},v) &= -\partial_{2,r}\varphi(s_{1},L(s_{1})+v^{2})=-\frac{1}{2\varepsilon}\mathbf{1}_{v\ge\lambda_{b}(s_{1})} \nonumber
\end{alignat}
%
%\begin{equation}\label{eq:Xi_epsilon}
%\begin{aligned}
%&u_1(s_1) = \begin{cases}
%    \lambda_b(s_1) &\quad \text{if } s_1 \in [s_{1,\varepsilon}^d, s_{1,\varepsilon}^u]\\
%    \varepsilon &\quad \text{if } s_1 \notin [s_{1,\varepsilon}^d, s_{1,\varepsilon}^u]\\
%    \end{cases} 
%    &u_2(s_2) = \begin{cases}
%    0 &\quad \text{if } s_2 \in [s_2^d, s_2^u]\\
%    -\frac{1}{2\varepsilon}\Lambda_b(s_2)&\quad \text{if } s_2 \notin [s_2^d, s_2^u]\\
%    \end{cases} \\
%    &\Delta^{S}(s_{1},v) = -\partial_{1,r}\varphi(s_{1},L(s_{1})+v^{2})=\frac{\Delta_{b}}{2\varepsilon}\mathbf{1}_{v\ge\lambda_{b}(s_{1})}
%    & \Delta^{L}(s_{1},v) = -\partial_{2,r}\varphi(s_{1},L(s_{1})+v^{2})=-\frac{1}{2\varepsilon}\mathbf{1}_{v\ge\lambda_{b}(s_{1})}
%\end{aligned}
%\end{equation}
is superreplicating by Proposition~\ref{prop:convex_portfolio}, since $u_{1}\ge \varphi^*_\mathrm{super}$. 
The payoffs $u_1$ and $u_2$ are plotted in Figure~\ref{fig:u1_bin}. Both $u_1$ and $u_2$ are continuous functions, with $u_1\in L^1(\mu_1)$ and $u_2\in L^1(\mu_2)$.
%
%\begin{figure}
%\title{$u_1$}\\
%\begin{center}
%\includegraphics[trim={3cm 3cm 3cm 3cm},clip,width=11cm,height=6cm]{graphes/u1_binomiale.pdf}
%%\includegraphics[scale=0.35, trim={2.5cm 3.5cm 0 3cm},clip]{graphes/u1_binomiale.pdf}
%\end{center}
%\caption{$T_1$ profile of an $\eps$-optimal superreplicating portfolio when $\mu_2$ is Bernoulli}
%\label{fig:u1_bin}
%\end{figure}
%
%\begin{figure}
%\title{$u_2$}
%\begin{center}
%\includegraphics[trim={3cm 3cm 3cm 3cm},clip,width=11cm,height=6cm]{graphes/u2_binomiale.pdf}
%%\includegraphics[scale=0.35, trim={2.5cm 3.5cm 0 3cm},clip]{graphes/u2_binomiale.pdf}
%\end{center}
%\caption{$T_2$ profile of an $\eps$-optimal superreplicating portfolio when $\mu_2$ is Bernoulli}
%\label{fig:u2_bin}
%\end{figure}

\begin{figure}
\begin{center}
\includegraphics[trim={3cm 3cm 3cm 3cm},clip,width=8cm,height=5.5cm]{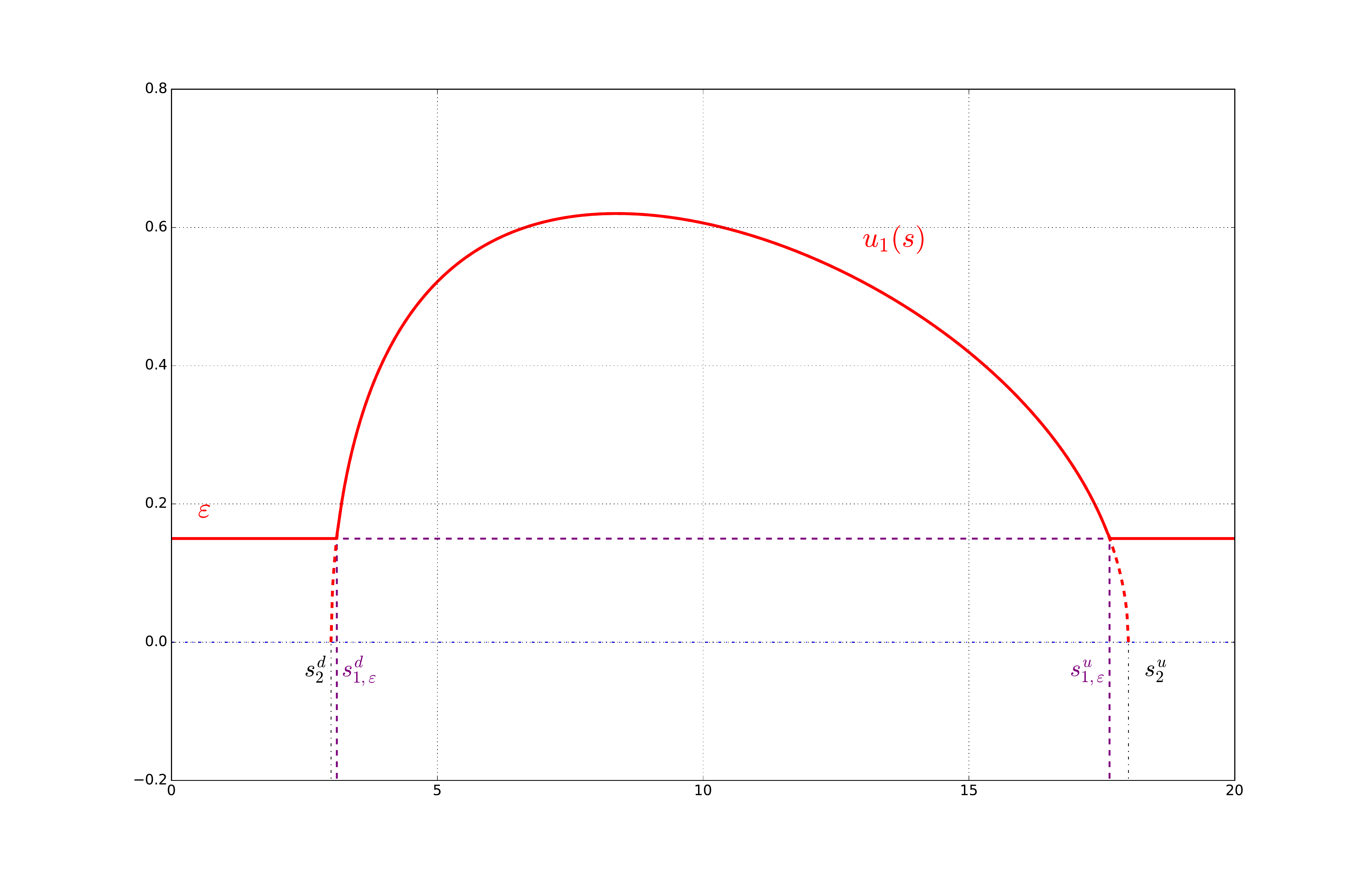}
\includegraphics[trim={3cm 3cm 3cm 3cm},clip,width=8cm,height=5.5cm]{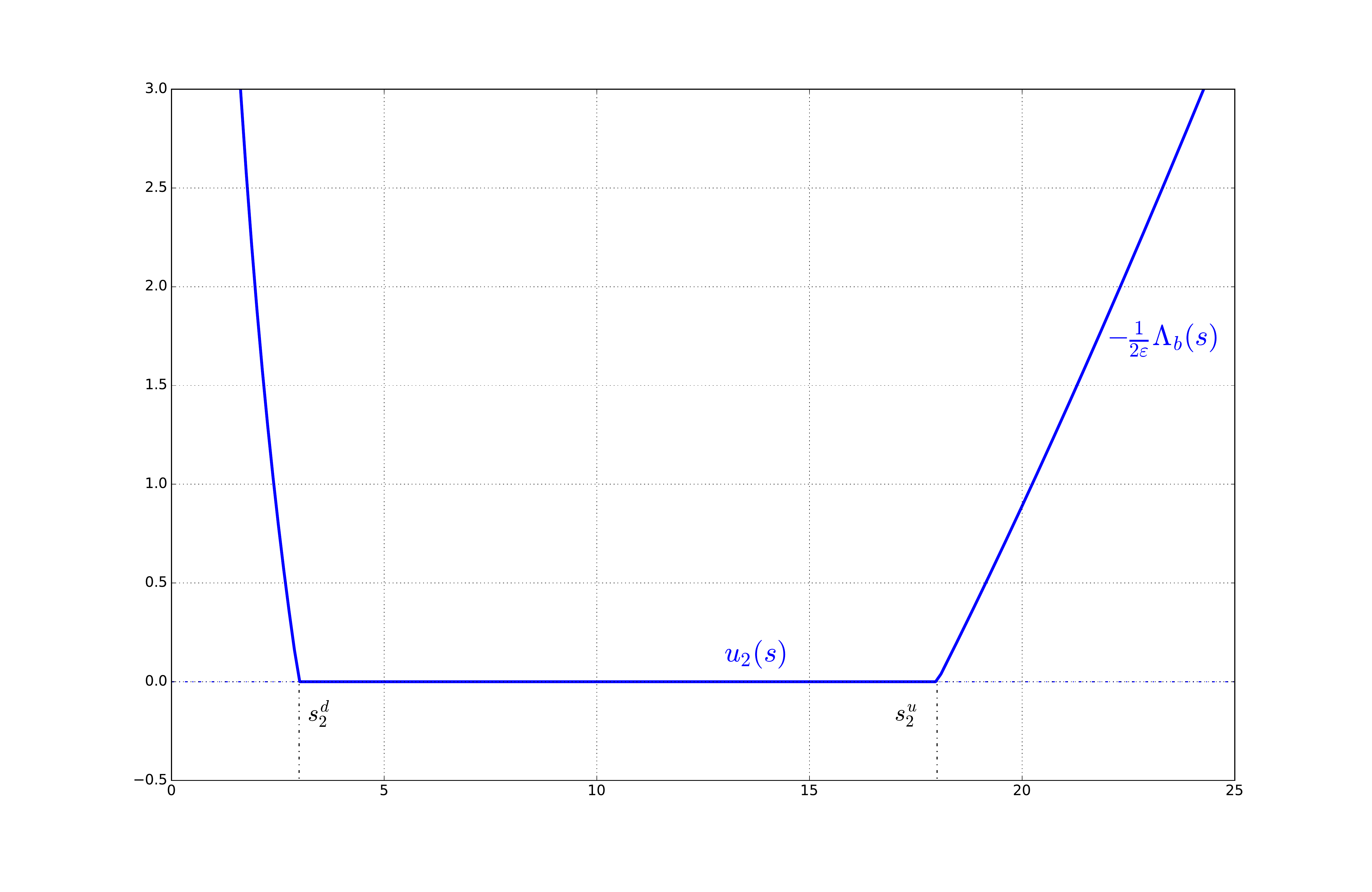}
\end{center}
\caption{Profiles of an $\eps$-optimal superreplicating portfolio when $\mu_2$ is Bernoulli}
\label{fig:u1_bin}
\end{figure}

\begin{prop}\label{prop:continuous_superrep_portfolio} 
Let $\mu_{2}$ be the Bernoulli distribution~(\ref{eq:mu2_bernoulli}) and $\mathrm{{supp}(\mu_{1})}\subset[s_{2}^{d},s_{2}^{u}]$. For $0<\varepsilon< \bar{\lambda}_b$, the portfolio defined in~(\ref{eq:Xi_epsilon}) is superreplicating and has price $\mathbb{E}^{1}\left[u_1(S_1)\right] + \mathbb{E}^{2}\left[u_2(S_2)\right]\leq P_\mathrm{super} + \varepsilon$. Moreover, if $\text{supp}(\mu_1) \subset (s_2^d, s_2^u)$, then $\mathbb{E}^{1}\left[u_1(S_1)\right] + \mathbb{E}^{2}\left[u_2(S_2)\right]=P_\mathrm{super}$ for all small enough $\eps>0$. 
\end{prop}

\begin{proof}
  The superreplication property has already been argued. By definition, $u_2(s_2) = 0$ $\mu_2$-almost everywhere and as a consequence,
\begin{multline*}
     \mathbb{E}^{1}\left[u_1(S_1)\right] + \mathbb{E}^{2}\left[u_2(S_2)\right]
    = \mathbb{E}^{1}\left[u_1(S_1)\right]
    = \mathbb{E}^{1}\left[\lambda_b(S_1)\mathbf{1}_{S_1 \in [s_{1,\varepsilon}^d, s_{1,\varepsilon}^u]}\right] + \mathbb{E}^{1}\left[\varepsilon\mathbf{1}_{S_1 \notin [s_{1,\varepsilon}^d, s_{1,\varepsilon}^u]}\right] \\
    \le \mathbb{E}^{1}\left[\lambda_b(S_1)\right] + \varepsilon = P_\mathrm{super} + \varepsilon
\end{multline*}
using the formula from Theorem~\ref{thm:3cases}. If $\text{supp}(\mu_1) \subset (s_2^d, s_2^u)$ and  $\varepsilon > 0$ is such that $\text{supp}(\mu_1) \subset [s_{1,\varepsilon}^d, s_{1,\varepsilon}^u]$, then we have $\mathbb{E}^{1}\left[\varepsilon\mathbf{1}_{S_1 \notin [s_{1,\varepsilon}^d, s_{1,\varepsilon}^u]}\right]=0$ in the above and thus $\mathbb{E}^{1}\left[u_1(S_1)\right] + \mathbb{E}^{2}\left[u_2(S_2)\right]=P_\mathrm{super}$.
\end{proof}

\section{The case where $\mu_2$ has compact support}\label{sec:mu2_compact_support}

In this section, we focus on the case where $\mu_2$ has compact support in $\RR_+^*$. We denote
\begin{equation}
s_2^d \equiv \min \, \mathrm{supp}(\mu_{2}) > 0, \qquad s_2^u \equiv \max \, \mathrm{supp}(\mu_{2}) < +\infty \label{eq:mu2_compact_support}.
\end{equation}
We assume throughout this section that the market is arbitrage-free, i.e., that $\mu_1$ and $\mu_2$ are in convex order. This implies that $\mathrm{{supp}(\mu_{1})}\subset[s_{2}^{d},s_{2}^{u}]$, in particular,
\begin{equation*}
s_1^d \equiv \min \, \mathrm{supp}(\mu_{1}), \qquad s_1^u \equiv \max \, \mathrm{supp}(\mu_{1})
\end{equation*}
satisfy $s_2^d\le s_1^d\le s_1^u \le s_2^u$.
Continuing the discussion from the preceding section, we seek sufficient conditions on the market smiles $\mu_1$ and $\mu_2$ under which $P_\mathrm{super} < \sigma_{12}$, and corresponding portfolios.
Following Theorem~\ref{thm:D=Dbar}, three strategies can be used to prove that $P_\mathrm{super} < \sigma_{12}$. We can (i) find a superreplication portfolio whose price is strictly smaller than $\sigma_{12}$, (ii) show that $\Law_{\mu_{1}}(S_{1}, L(S_{1})-\ell_1)$ and $\Law_{\mu_{2}}(S_{2}, L(S_{2}) - \ell_2)$ are not in convex order, or (iii) verify that $\bar\cm(\mu_1,\mu_2)=\emptyset$. 
The following result uses all three strategies. We recall from Section~\ref{sec:binomial} the function $\Lambda_b(s_{1})\equiv \Lambda_{s_2^d,s_2^u}(s_{1})$ defined in (\ref{eq:def_lambda_b}), its root $\lambda_b(s_{1})=\sqrt{\Lambda_b(s_{1})}$ for $s_1\in [s_2^d,s_2^u]$ as well as $\bar{\lambda}_b = \max_{[s_2^d,s_2^u]}\lambda_b>0$ and the solutions $\lambda_{b,\pm}^{-1}(\sigma)$ from~\eqref {eq:lambdaDefSol}.

\begin{prop}\label{prop:mu2_comp_supp}
	Assume (\ref{eq:mu2_compact_support}) and absence of arbitrage. Each of the following implies $P_\mathrm{super} < \sigma_{12}$:
	\begin{itemize}
	\item[(i)] $\mathbb{E}^{1}\left[\lambda_b(S_1)\right]<\sigma_{12}$,
	\item[(ii)] $L(s_1^d)-\ell_1 > L(s_2^d)-\ell_2$, which holds in particular if $s_1^d=s_2^d$ and $\mu_1\neq \mu_2$,
	\item[(iii)] $\mu_{1}(A)>0$ for $A \equiv (s_2^d,\lambda_{b,-}^{-1}(\sigma_{12})) \cup (\lambda_{b,+}^{-1}(\sigma_{12}),s_2^u)$.
	\end{itemize}
\end{prop}

\begin{proof}
	(i) Let $\varepsilon \in (0,\bar{\lambda}_b)$ be such that $\mathbb{E}^{1}\left[\lambda_b(S_1)\right]+\varepsilon<\sigma_{12}$. Exactly as in Section~\ref{sec:binomial}, the portfolio defined in~(\ref{eq:Xi_epsilon}) superreplicates the VIX, and its price is bounded from above by $\mathbb{E}^{1}\left[\lambda_b(S_1)\right] + \varepsilon<\sigma_{12}$. This proves that $P_\mathrm{super} < \sigma_{12}$.
	
	(ii) Let $L(s_1^d)-\ell_1 > L(s_2^d)-\ell_2$. As $L$ is decreasing, this means that the support of $\Law_{\mu_{1}}(L(S_{1})-\ell_1)$ is not included in the convex hull of the support of $\Law_{\mu_{2}}(L(S_{2}) - \ell_2)$, so these distributions are not in convex order, and then the same holds for $\Law_{\mu_{1}}(S_{1}, L(S_{1})-\ell_1)$ and $\Law_{\mu_{2}}(S_{2}, L(S_{2}) - \ell_2)$. By Theorem~\ref{thm:D=Dbar}, we conclude that $P_\mathrm{super}<\sigma_{12}$. If $\mu_1\neq \mu_2$, then $\ell_1 < \ell_2$ due to the strict convexity of $L$, which yields the second claim.
	
	(iii) Since $\mu_1$ and $\mu_2$ are in convex order, $\cm(\mu_1,\mu_2)\neq\emptyset$. For any $\mu\in\cm(\mu_1,\mu_2)$ and for $\mu_1$-almost all $s_1$,
	\begin{equation*}
	\Lambda_\mu(s_1) \le \sup_{\pi\in\Pi_{s_1}} \EE^{\pi}\left[ L\left(\frac{S_2}{s_1}\right)\right] = \Lambda_b(s_1), \label{eq:upper_bound_lambda_mu}
	\end{equation*}
	where $\Pi_{s_1}$ denotes the set of all probability measures $\pi$ such that $\mathrm{{supp}(\pi)}\subset[s_{2}^{d},s_{2}^{u}]$ and $\EE^{\pi}[S_2]=s_1$. Indeed, the inequality follows directly from the definition (\ref{eq:def_Lambda}) of $\Lambda_\mu$ and the fact that $\mathrm{{supp}(\mu_{2})}\subset[s_{2}^{d},s_{2}^{u}]$; moreover, since $L$ is convex, the supremum is attained when $\pi$ is the Bernoulli distribution that takes values in $\{s_{2}^{d},s_{2}^{u}\}$ and has mean $s_1$, whence the equality. As a consequence,
	\[
	\sigma_{12}^2 = \EE^1[\Lambda_\mu(S_1)]\le \EE^1[\Lambda_b(S_1)],
	\]
	hence $\sigma_{12} \le\bar{\lambda}_b$ and the set $A$ is well defined. Let $\mu\in\cm(\mu_1,\mu_2)$, then $\Lambda_\mu(s_1)\le\Lambda_b(s_1)<\sigma_{12}^2$ for $\mu_{1}$-a.e.\ $s_{1}\in A$. As a consequence, if $\mu_{1}(A)>0$, there exists no $\mu\in\cm(\mu_1,\mu_2)$ such that $\Lambda_\mu(S_1)=\sigma_{12}^2$ $\mu_1$-a.s., that is, $\bar\cm(\mu_1,\mu_2)=\emptyset$. By Theorem~\ref{thm:D=Dbar}, this implies $P_\mathrm{super} < \sigma_{12}$.
\end{proof}

\section{Smiles for which the classical upper bound is optimal}\label{sec:examples_where_classical_is_optimal} 

In this section, we show how to construct examples of arbitrage-free smiles $\mu_{1},\mu_{2}$ such that the classical upper bound is optimal, i.e., $P_\mathrm{super}=\sigma_{12}$. We recall from Theorem~\ref{thm:D=Dbar} that this is equivalent to $\bar{\cm}(\mu_1,\mu_2)$ of Definition~\ref{def:Mbar} being nonempty. Going backward, let $\mu\in\bar{\cm}(\mu_1,\mu_2)$, that is, $\mu\in\cm(\mu_1,\mu_2)$ and the price $\Lambda_\mu(S_1) \equiv \mathbb{E}^{\mu}\left[L\left(S_{2}/S_{1}\right)\middle|S_{1}\right]$ of the FSLC at $T_{1}$ is constant and equal to $\sigma_{12}^{2}$. Disintegrating $\mu=\mu_{1}\otimes T$ into its first marginal $\mu_{1}$ and a transition kernel $T(s,dx)=\mu(dx|S_{1}=s)$, these two conditions can be stated as
$$
  \int x \, T(s,dx) =s,\qquad\! \int L(x) \, T(s,dx) =L(s)+ \sigma_{12}^2 \qquad\mbox{for $\mu_{1}$-a.e.\ $s\in\R_{+}^{*}$.}
$$
 Conversely, if $T$ is a stochastic kernel with these two properties, then $\mu_{1}\otimes T\in \bar{\cm}(\mu_1,\mu_2)$. In brief, constructing an element of $\bar{\cm}(\mu_1,\mu_2)$ boils down to determining such a kernel.

One instance, similar to Example~4.14 in~\cite{phl-demarco}, is the following conditional Bernoulli model. Given $\sigma_{12}\geq0$ and measurable functions $\alpha_d, \alpha_u$ such that $0 < \alpha_d < 1 < \alpha_u$, let
\begin{equation}
\label{eq:cond_binomial}
  T(s,dx) \equiv p(s)\delta_{\alpha_{u}(s)s}(dx)  + (1-p(s)) \delta_{\alpha_{d}(s)s}(dx).
\end{equation}
Then, $T$ satisfies the two conditions if and only if, $\mu_{1}$-a.s.,
$p(s)\alpha_{u}(s)+(1-p(s))\alpha_{d}(s)=1$, i.e.,
\[
p(s)=\frac{1-\alpha_{d}(s)}{\alpha_{u}(s)-\alpha_{d}(s)},
\]
and  $p(s) L\left(\alpha_{u}(s)\right)+(1-p(s))L\left(\alpha_{d}(s)\right)=\sigma_{12}^{2}$,
i.e.,
\[
\frac{1-\alpha_{d}(s)}{\alpha_{u}(s)-\alpha_{d}(s)}\ln\alpha_{u}(s)+\frac{\alpha_{u}(s)-1}{\alpha_{u}(s)-\alpha_{d}(s)}\ln\alpha_{d}(s)=-\frac{1}{2}\sigma_{12}^{2}\tau.
\]
The latter can be rewritten as
\[
f_{\alpha_d(s)}(\alpha_u(s)) \equiv \frac{1-\alpha_d(s)}{\alpha_u(s)-\alpha_d(s)}\ln \frac{\alpha_u(s)}{\alpha_d(s)} + \ln \alpha_d(s) = -\frac{1}{2}\sigma_{12}^{2}\tau
\]
and for $0<a_{0}<1$, $a\mapsto f_{a_{0}}(a)$ is a decreasing continuous function on $[1, \infty)$ such that $f_{a_{0}}(1)=0$ and $f_{a_{0}}(\infty)=\ln a_{0}$. Therefore, if $\alpha_d$ is given and $\ln \alpha_d < -\frac{1}{2}\sigma_{12}^{2}\tau$, then there exists a unique $\alpha_u = f_{\alpha_d}^{-1}(-\frac{1}{2}\sigma_{12}^{2}\tau)$ such that the above conditions are satisfied. Now, \eqref{eq:cond_binomial} defines a kernel $T$. If $\mu_{1}$ is any distribution on $\R_{+}^{*}$ with mean $S_{0}$, we can set $\mu=\mu_{1}\otimes T$ and define $\mu_{2}$ to be the second marginal of $\mu$. Then, we have constructed smiles $\mu_{1},\mu_{2}$ together with an element  $\mu\in\bar{\cm}(\mu_1,\mu_2)$, thus $P_\mathrm{super}=\sigma_{12}$.
A similar reasoning can be applied to three-point or $n$-point models instead of Bernoulli.

%%%%%%%%%%%%%%%%%%%%%%%%%%%%%%%%%%%%%%%
\section{Numerical results}\label{sec:num}

\subsection{LP solver}\label{sec:LPsolver}

When we numerically solve the primal problems (\ref{eq:primal}) and (\ref{eq:primal_sub}), we discretize the payoffs $u_1$ and $u_2$ using a finite basis of out-the-money (OTM) calls and puts, cash $\alpha$, an initial delta $\Delta$, as well as the log-contract (so that the LP solver can exactly recover the classical superreplicating portfolio). Moreover, we decompose the deltas $\Delta^{S}$ and $\Delta^{L}$ over a polynomial basis, with respective orders $n_S$ and $n_L$:
%\[
%\Delta^{S}(s_{1},v)=\sum_{i,j=0}^{n_{S}}\Delta_{ij}^{S}\left(\frac{s_{1}}{S_{0}}\right)^{i}v^{j},\qquad\Delta^{L}(s_{1},v)=\sum_{i,j=0}^{n_{L}}\Delta_{ij}^{L}\left(\frac{s_{1}}{S_{0}}\right)^{i}v^{j}.
%\]
\[
\Delta^{S}(s_{1},v)=\sum_{0\le i+j\le n_S}\Delta_{ij}^{S}\left(\frac{s_{1}}{S_{0}}\right)^{i}v^{j},\qquad\Delta^{L}(s_{1},v)=\sum_{0\le i+j\le n_L}\Delta_{ij}^{L}\left(\frac{s_{1}}{S_{0}}\right)^{i}v^{j}.
\]
Numerically, we can only check the super/subreplication constraints (\ref{eq:constraint}) and (\ref{eq:constraint_sub}) on a large but finite grid $G$ of values of $(s_1,s_2,v)$. Therefore our numerical upper bound is
\bea
P_\mathrm{super}^\mathrm{num}\equiv\inf_{\theta\in \Theta}\left\{\alpha+\sum_{i=1}^{m_{1}}\omega_{i}^{1}O_1(K_{i}^{1})+\sum_{i=1}^{m_{2}}\omega_{i}^{2}O_2(,K_{i}^{2})+\beta_{1}\mathrm{VS}_1+\beta_{2}\mathrm{VS}_2\right\} \label{eq:primal_num}
\eea
where
\begin{eqnarray*}
\theta = (\alpha,\Delta,\omega_{1}^{1},\ldots,\omega_{m_{1}}^{1},\omega_{1}^{2},\ldots,\omega_{m_{2}}^{2},\beta_{1},\beta_{2},(\Delta_{ij}^{S},0\le i,j\le n_{S}),(\Delta_{ij}^{L},0\le i,j\le n_{L}))^{t}\in\mathbb{R}^{p}
\end{eqnarray*}
(with $p = 4+m_{1}+m_{2}+\frac{(n_{S}+1)(n_{S}+2)}{2}+\frac{(n_{L}+1)(n_{L}+2)}{2}$) and $\Theta$ is the set of variables $\theta$ such that
\begin{multline*}
\forall (s_{1},s_{2},v)\in G,\qquad \Pi_\mathrm{num}(s_1,s_2,v)\equiv\alpha+\Delta(s_{1}-S_{0})+\sum_{i=1}^{m_{1}}\omega_{i}^{1}g(s_{1},K_{i}^{1})+\sum_{i=1}^{m_{2}}\omega_{i}^{2}g(s_{2},K_{i}^{2}) \\
+\beta_{1}\left(-\frac{2}{T_{1}}\ln\frac{s_{1}}{S_{0}}\right)+\beta_{2}\left(-\frac{2}{T_{2}}\ln\frac{s_{2}}{S_{0}}\right)
+\Delta^{S}(s_{1},v)(s_{2}-s_{1})+\Delta^{L}(s_{1},v)\left(-\frac{2}{\tau}\ln\frac{s_{2}}{s_{1}}-v^{2}\right)\ge v.
\end{multline*}
Here, $g(s,K)$ denotes the OTM vanilla payoff, i.e.,
\[
g(s,K)=\begin{cases}
(K-s)_{+} & \mathrm{if}\,\, K\le S_{0}\quad\mathrm{(OTM\,\, put)}\\
(s-K)_{+} & \mathrm{if}\,\, K>S_{0}\quad\mathrm{(OTM\,\, call)},
\end{cases}
\]
$O_i(K)$ denotes the market price of the OTM vanilla payoff $g(S_{i},K)$,
and $\mathrm{VS}_i$ denotes the price of the variance swap of
maturity $T_{i}$, i.e., of the payoff $-\frac{2}{T_{i}}\ln\frac{S_{i}}{S_{0}}$.
Similarly, our numerical lower bound reads
\bea
P_\mathrm{sub}^\mathrm{num}\equiv\sup_{\theta\in \Theta}\left\{\alpha+\sum_{i=1}^{m_{1}}\omega_{i}^{1}O_1(K_{i}^{1})+\sum_{i=1}^{m_{2}}\omega_{i}^{2}O_2(K_{i}^{2})+\beta_{1}\mathrm{VS}_1+\beta_{2}\mathrm{VS}_2\right\}  \label{eq:primal_sub_num}
\eea
where $\Theta$ is the set of variables $\theta$ such that for all $(s_{1},s_{2},v)\in G$, $\Pi_\mathrm{num}(s_1,s_2,v)\le v$. Note that if we could check the constraints everywhere, and not only on the finite grid $G$, we would have $P_\mathrm{super} \le P_\mathrm{super}^\mathrm{num}$ and $P_\mathrm{sub}^\mathrm{num} \le P_\mathrm{sub}$, which means that $P_\mathrm{super}^\mathrm{num}$ and $P_\mathrm{sub}^\mathrm{num}$ would be acceptable upper and lower bounds. Thus it is important to use a large enough grid $G$.

\medskip

To solve the problems (\ref{eq:primal_num}) and (\ref{eq:primal_sub_num}), we have used the software package MOSEK, with $m_1=m_2=30$, $n_S=n_L=4$, and a grid of constraints $G$ made of 130 values of $s_1$, 130 values of $s_2$ (unevenly distributed from 0.01 to 5, and including the strikes of the OTM calls and puts) and 100 values of $v$ (unevenly distributed from $1.9\%$ to $273\%$). Let us first consider the case where the smiles $\mu_1$ and $\mu_2$ are those of a SABR model
%\bea
%dS_t &=& \sigma_t S_t^{-\beta} \, dW_t \nonumber \\
%d\sigma_t &=& \alpha \sigma_t \, dZ_t \label{eq:SABRmodel} \\
%d\langle W,Z\rangle_t &=& \rho \, dt \nonumber
%\eea
\bea
dS_t = \sigma_t S_t^{-\beta} \, dW_t, \qquad d\sigma_t = \alpha \sigma_t \, dZ_t, \qquad d\langle W,Z\rangle_t &=& \rho \, dt \label{eq:SABRmodel}
\eea
with $\sigma_0 = 20\%$, $\beta = -0.7$, $\alpha = 1$, $\rho = -50\%$, and $T_1=2$ months. The corresponding smiles and densities are reported in Figure \ref{fig:sabr_densities}. The implied volatility of the FSLC is $\sigma_{12} \approx 22.8\%$. The LP solver yields $P_\mathrm{super}^\mathrm{num} \approx 22.8\%$, together with the classical superreplicating portfolio (\ref{eq:portfolio_superrep_classical}), so the classical upper bound seems to be optimal. For the lower bound, we get $P_\mathrm{sub}^\mathrm{num} \approx 7.2\%$, which is much larger than the classical lower bound (zero), and the corresponding portfolio $(u_1,u_2,\Delta^S,\Delta^L)$ is reported in Figure \ref{fig:sabr_optimal_subreplicating_portfolio}. 

\begin{figure}[h]
\begin{center}
\includegraphics[width=8cm,height=6cm]{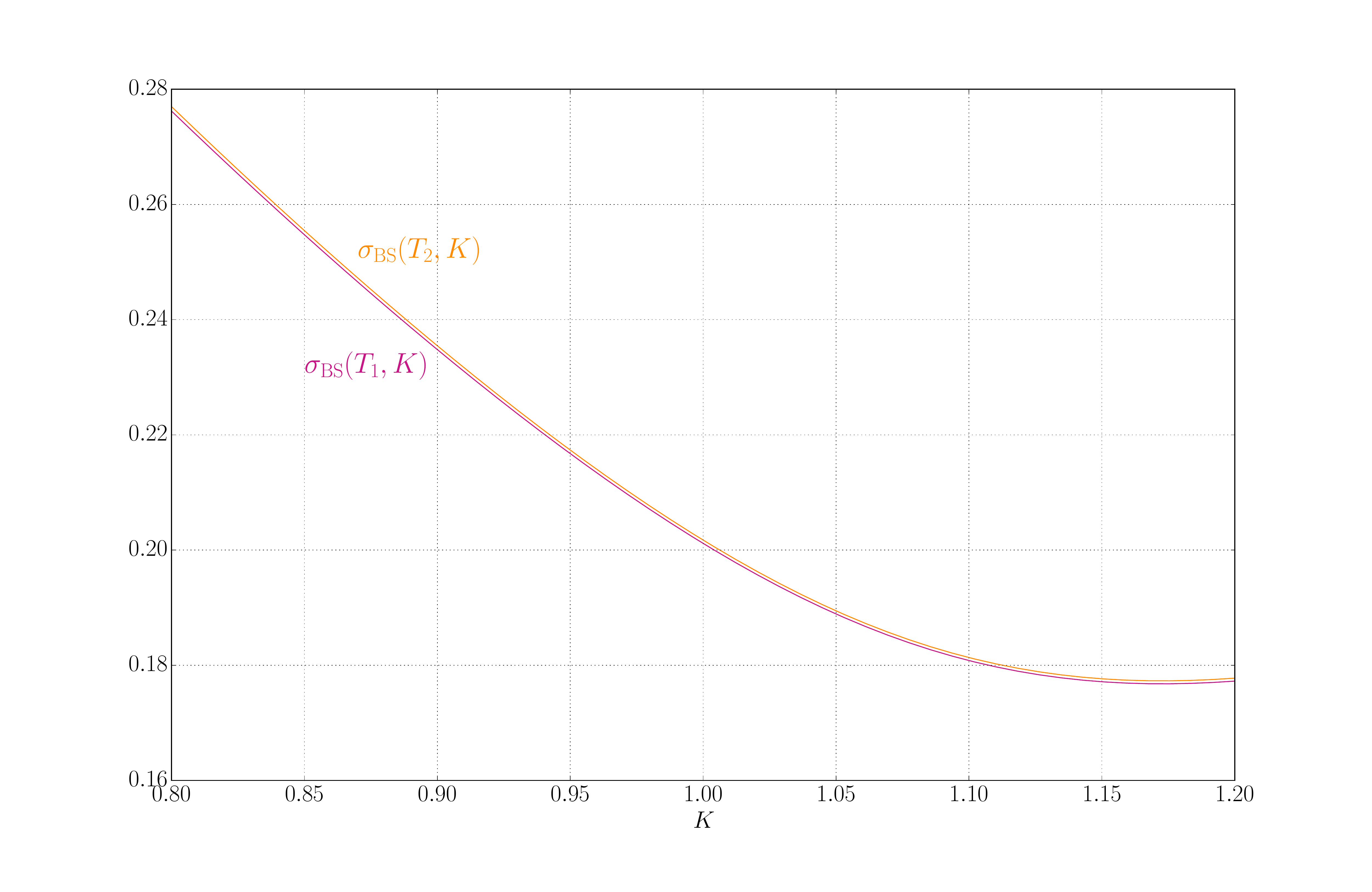}
\includegraphics[width=8cm,height=6cm]{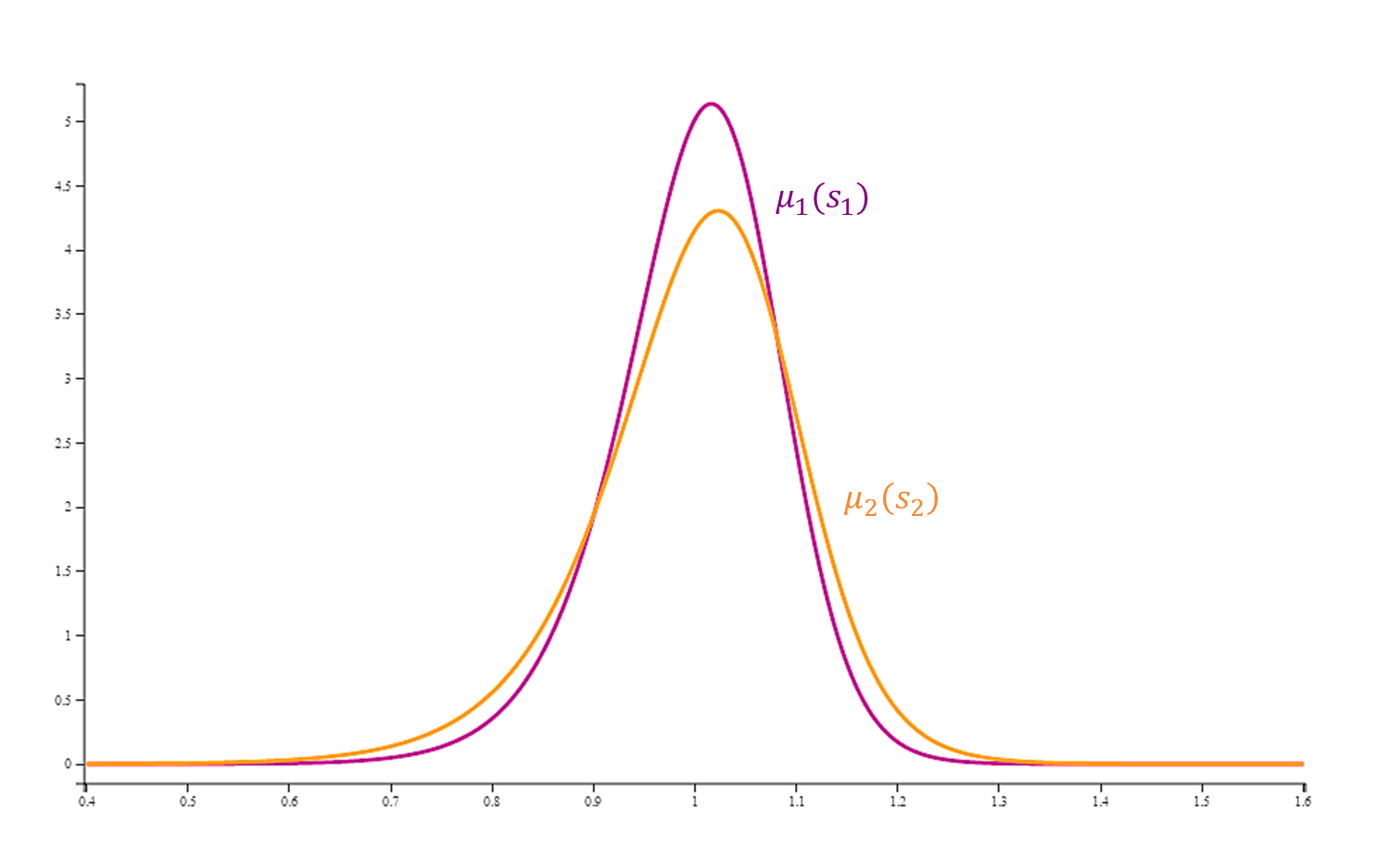}
\end{center}
\caption{Smiles (left) and densities (right) of the SABR model (\ref{eq:SABRmodel}) at maturities $T_1=2$ months and $T_2=T_1+30$ days; $\sigma_0 = 20\%$, $\beta = -0.7$, $\alpha = 1$, and $\rho = -50\%$}
\label{fig:sabr_densities}
\end{figure}

\begin{figure}
\begin{center}
\includegraphics[width=8cm,height=6cm]{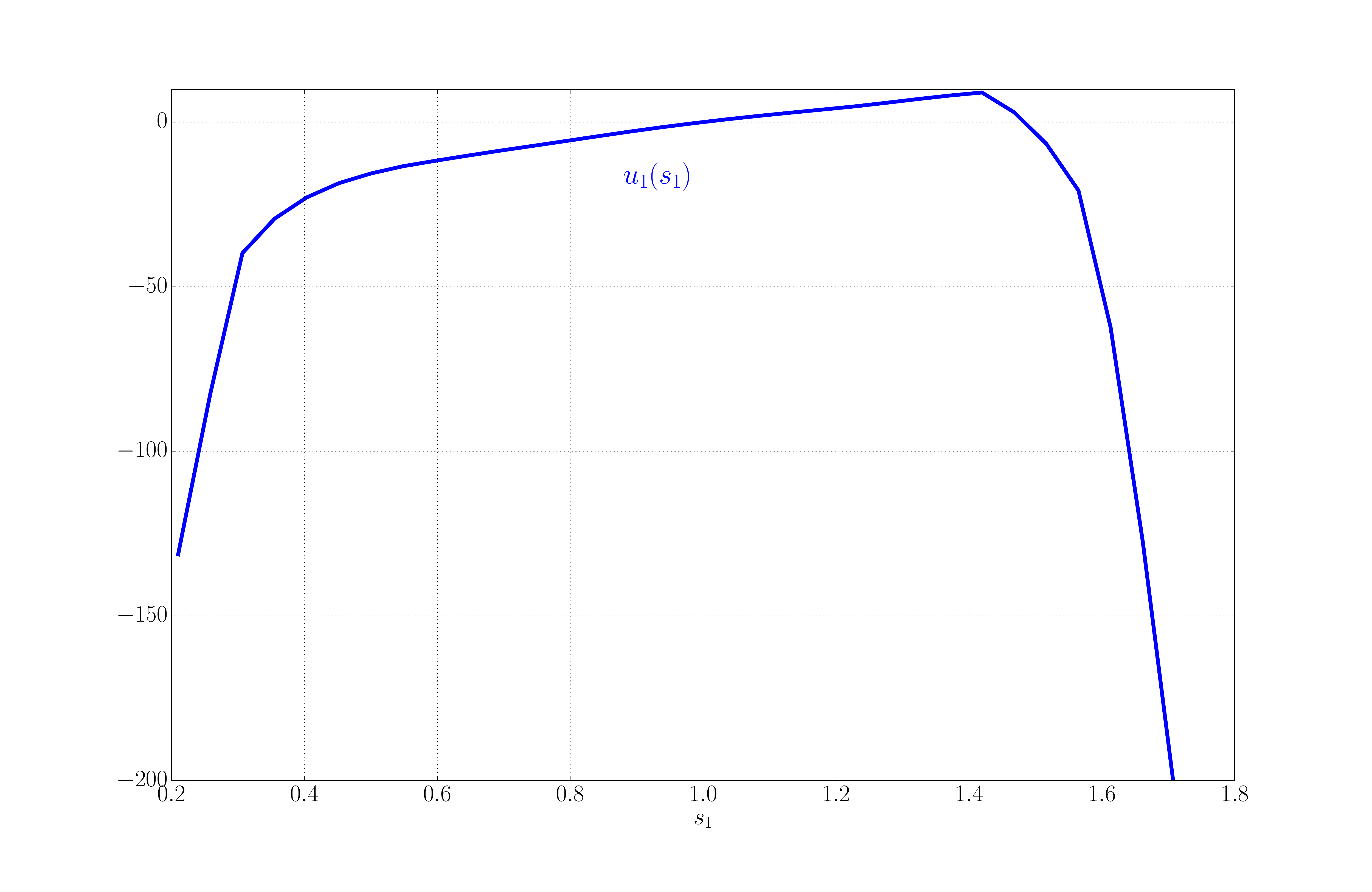}
\includegraphics[width=8cm,height=6cm]{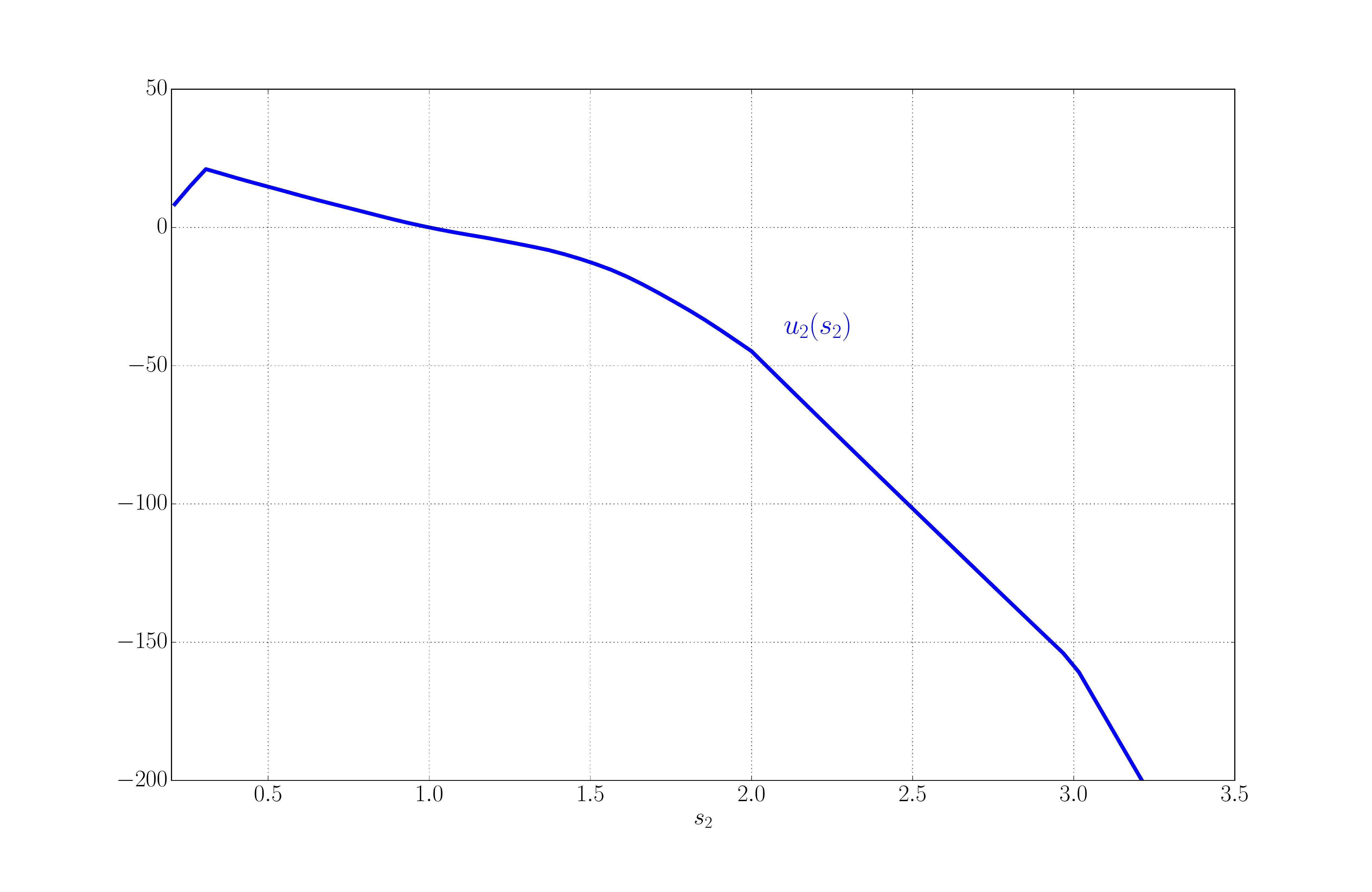}\\
\includegraphics[trim={4cm 4cm 4cm 4cm}, clip,width=8cm,height=5.5cm]{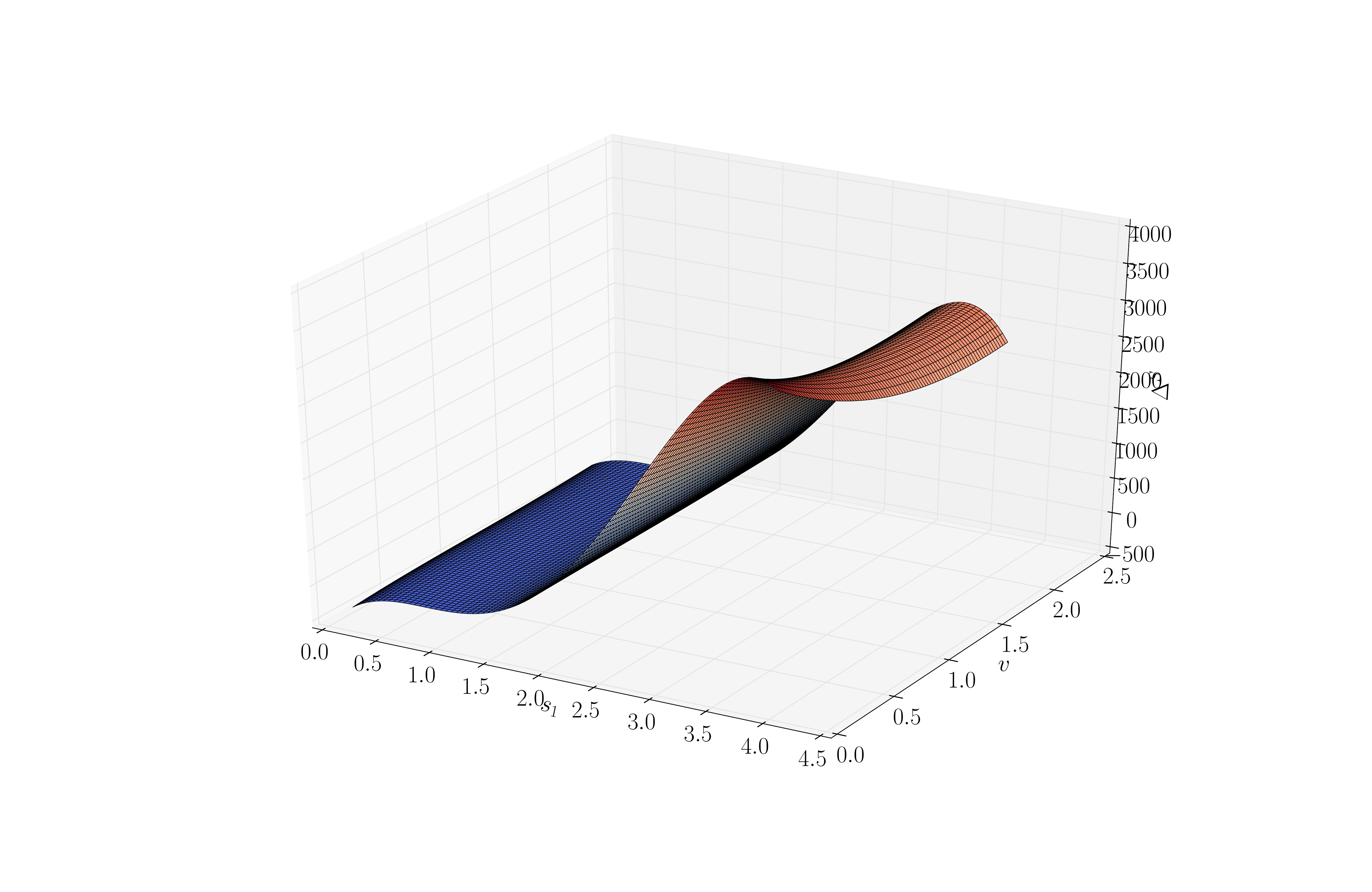}
\includegraphics[trim={4cm 4cm 4cm 4cm}, clip,width=8cm,height=5.5cm]{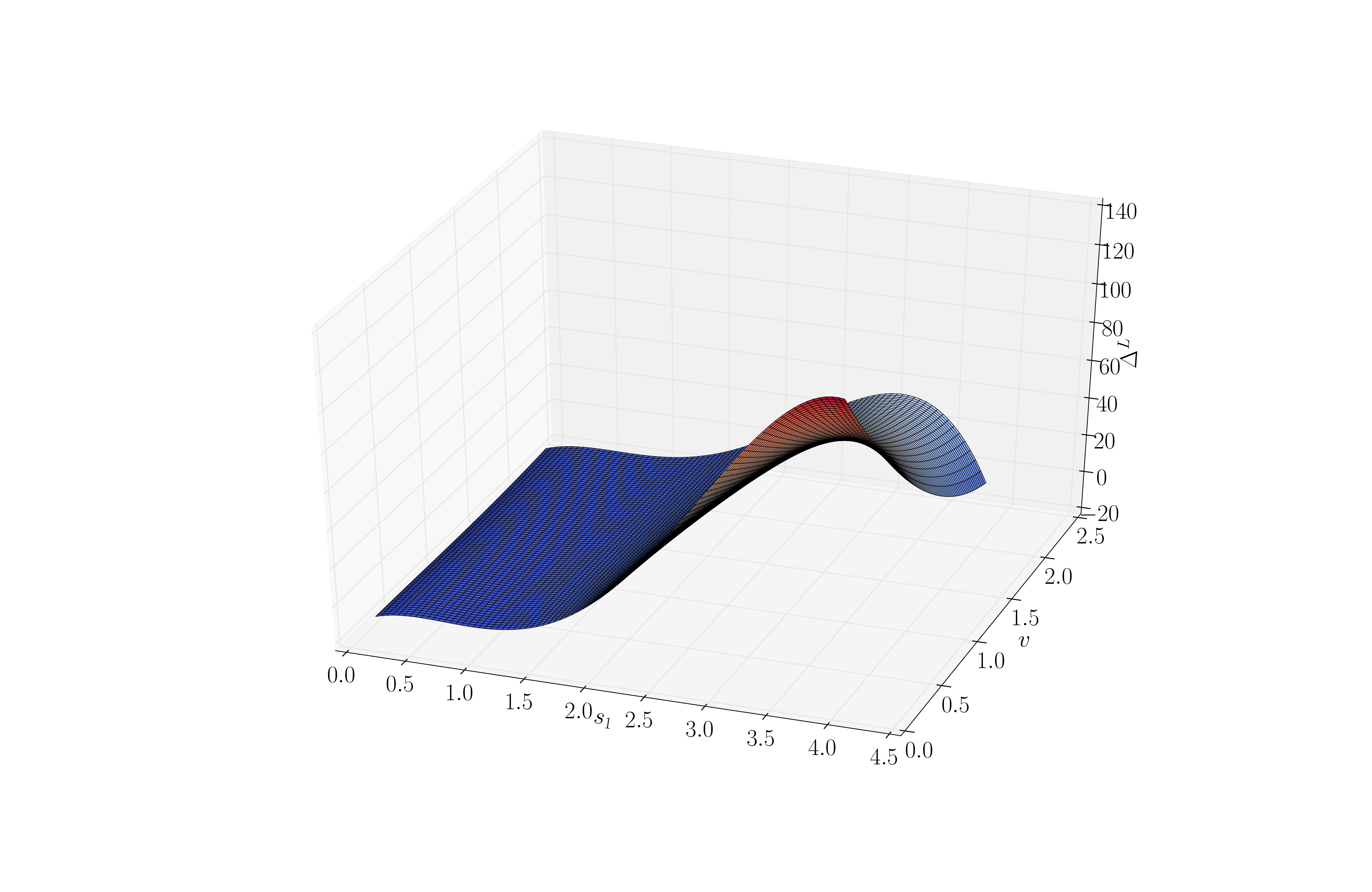}
\end{center}
\caption{Numerical optimal subreplicating portfolio $u_1$ (top left), $u_2$ (top right), $\Delta^S$ (bottom left), $\Delta^L$ (bottom right) when $\mu_1$ and $\mu_2$ are the SABR risk-neutral measures of Figure \ref{fig:sabr_densities}, using polynomials $\Delta^S$, $\Delta^L$ of degree 4}
\label{fig:sabr_optimal_subreplicating_portfolio}
\end{figure}

In Figure \ref{fig:sabr_subreplication_constraints} we check the subreplication constraints $\Pi_\mathrm{num}(s_1,s_2,v)\le v$. Note that, since the grid $G$ is finite, subreplication is not guaranteed everywhere. For instance, Figure \ref{fig:sabr_subreplication_constraints_v2p28} shows that at the very high value $v=228\%$, the subreplication constraint is not satisfied for some $(s_1,s_2)$. This is because $228\%$ is not a value of $v$ in the grid~$G$; its nearest neighbors in this grid are $v=182\%$ and $v=273\%$.

\begin{figure}
\begin{center}
\includegraphics[trim={2.5cm 2cm 2.5cm 2.5cm}, clip, width=8cm,height=5.5cm]{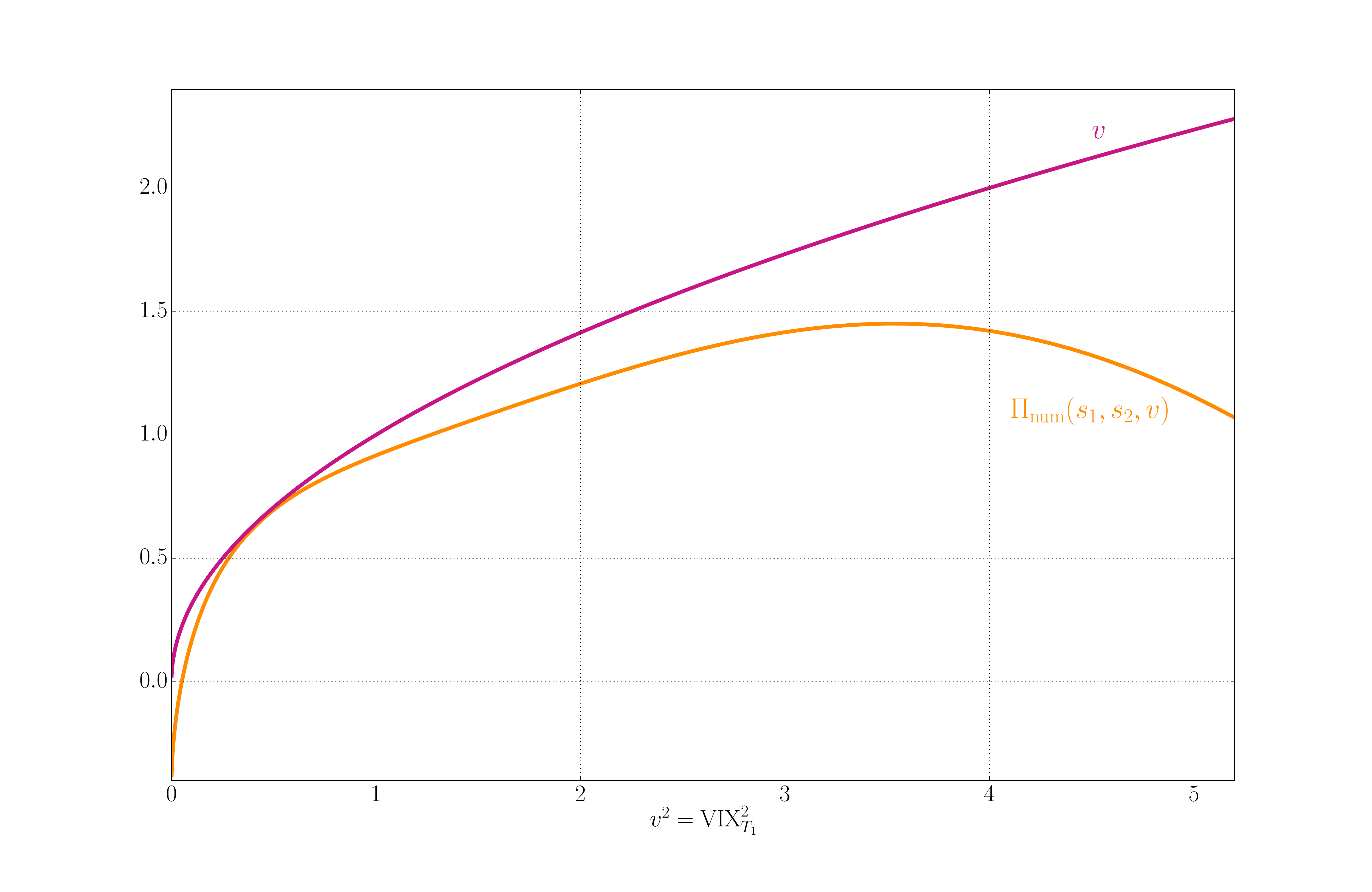}
\includegraphics[trim={4cm 4cm 4cm 4cm}, clip, width=8cm,height=5.5cm]{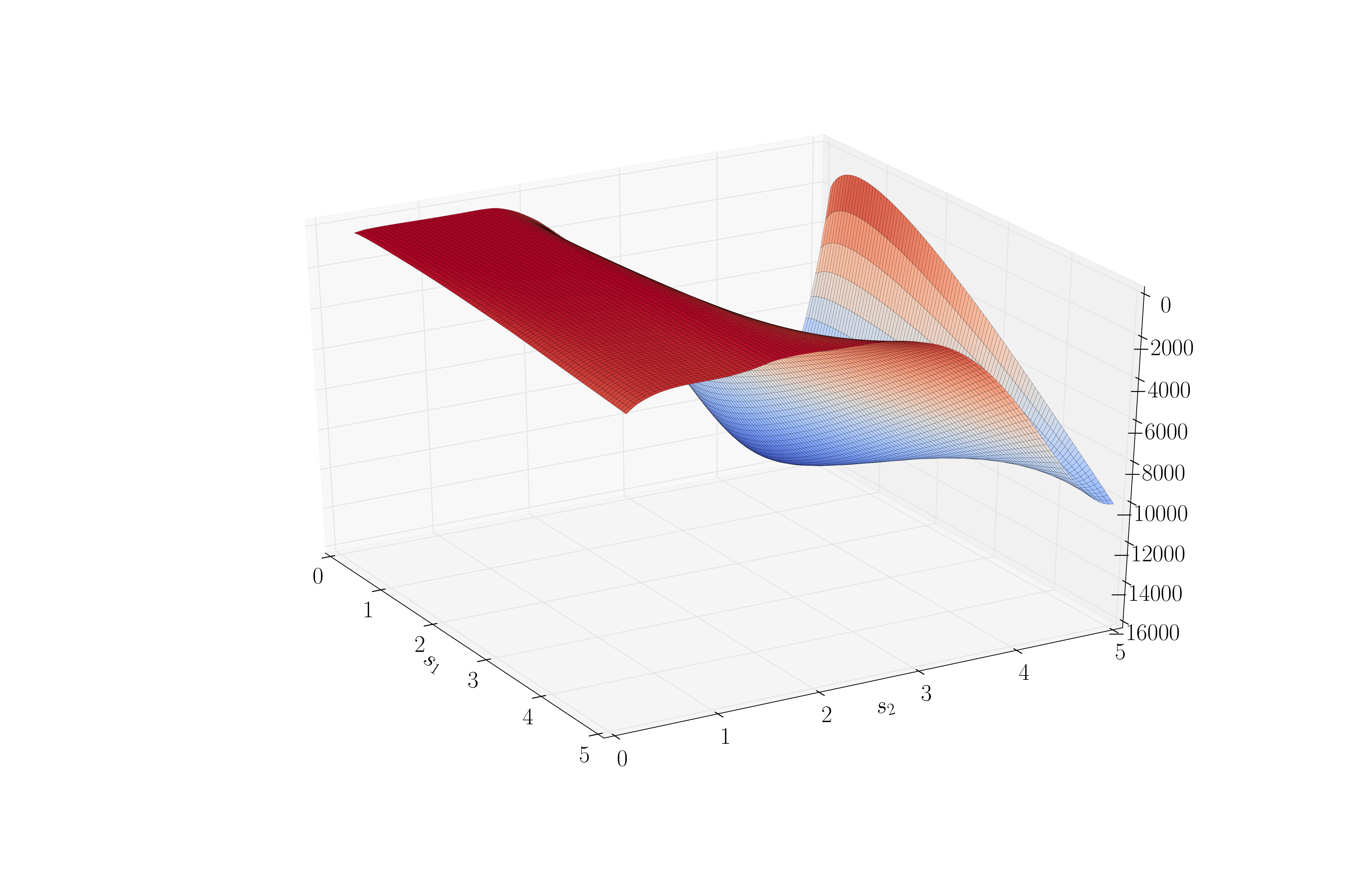}
\end{center}
\caption{Left: $\Pi_\mathrm{num}(s_1,s_2,v)$ and $v$ as a function of $v^2$, for $s_1=1.03$ and $s_2=1.18$. Right: $\Pi_\mathrm{num}(s_1,s_2,v)-v$ as a function of $(s_1,s_2)$ for $v=47\%$. This quantity is always negative. We used the SABR risk-neutral measures of Figure \ref{fig:sabr_densities}.}
\label{fig:sabr_subreplication_constraints}
\end{figure}

\begin{figure}
\begin{center}
\includegraphics[trim={3cm 3cm 3cm 3cm},width=8cm,height=6cm]{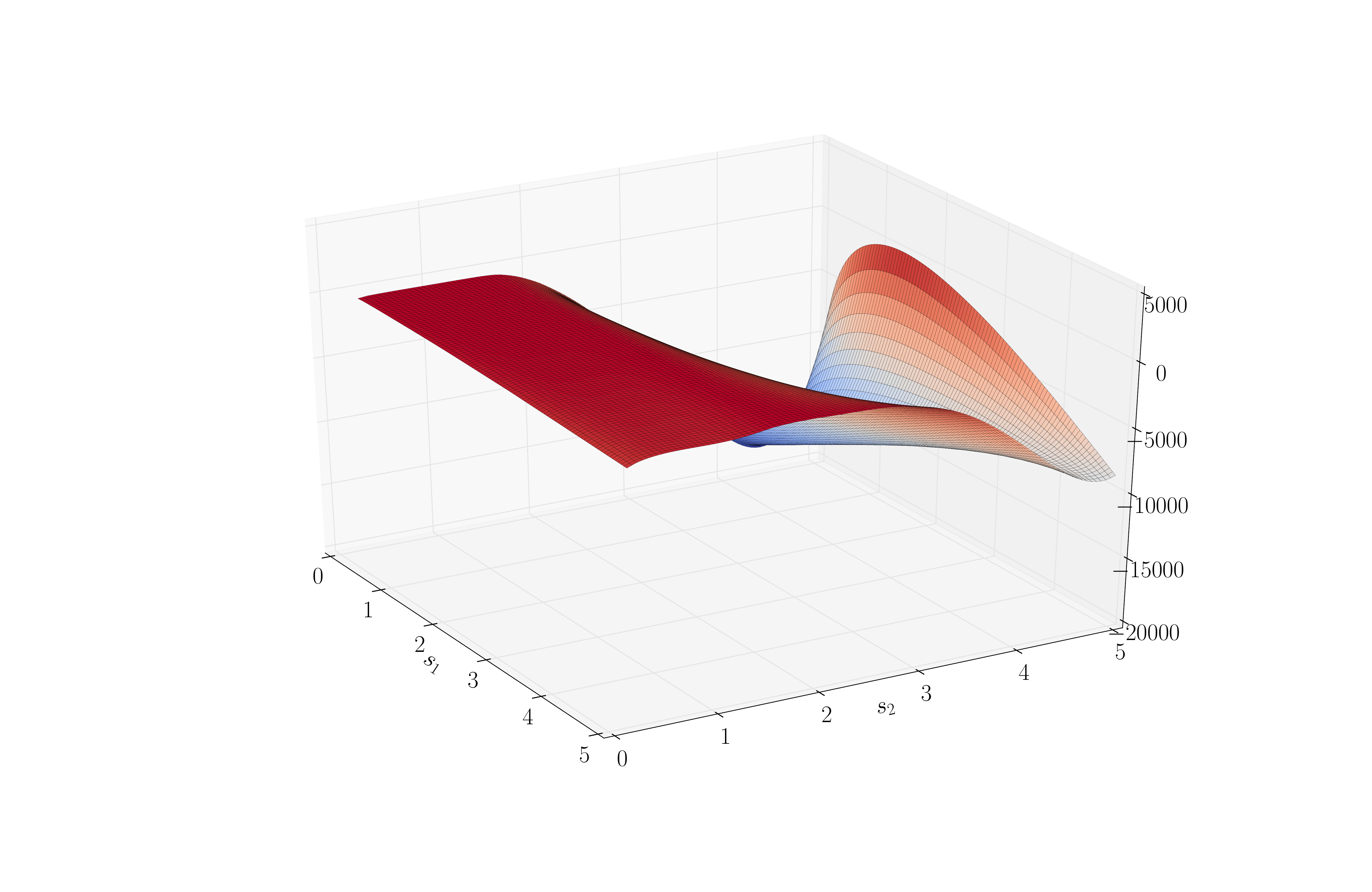}
\end{center}
\caption{$\Pi_\mathrm{num}(s_1,s_2,v)-v$ as a function of $(s_1,s_2)$ for $v=228\%$. This quantity is \emph{not} always negative (maximum value: 26). We used the SABR risk-neutral measures of Figure~\ref{fig:sabr_densities}.}
\label{fig:sabr_subreplication_constraints_v2p28}
\end{figure}

\subsection{Optimization over functionally generated subreplicating portfolios}

We recall from~\eqref{eq:cvx1def} that $P_\mathrm{sub} \ge P_\mathrm{sub}^{\mathrm{ccv},1}$,
where $P_\mathrm{sub}^{\mathrm{ccv},1}$ is the bound obtained from functionally generated portfolios,
\[
P_\mathrm{sub}^{\mathrm{ccv},1}=\sup_{\psi\in\mathcal{\mathcal{F}_{\mathrm{ccv}}}(\mathbb{R}),a\in\mathbb{R}}\left\{ \mathbb{E}^{1}\left[\inf_{v\ge0}\left\{ v-\psi(aS_{1}+L(S_{1})+v^{2})\right\} \right]+\mathbb{E}^{2}[\psi(aS_{2}+L(S_{2}))]\right\} .
\]

\subsubsection{Piecewise linear profiles}
Here we consider concave functions on $\RR$ that are piecewise linear. We start with a partition $-1=x_0<\cdots<x_N=1$ of $[-1, 1]$ and look at the concave, piecewise linear functions defined on $[-1, 1]$ with kinks at the points $x_i$ and a nonpositive right slope. These can be parametrized in the following way: 
\[
\psi_\omega(x) = \omega_{N-1}(x_{N-1}-x) - \displaystyle \sum_{i=1}^{N-1} \omega_{i-1} (x_i-x)_+ \, ,
\qquad \omega \in \mathbb{R}_+^{N}.
\]
We then extend the domain of definition of these functions to $\RR$ by linear extrapolation. Finally, we consider the homothetic transforms of the $\psi_\omega$, i.e., $\left\{s \mapsto \psi_\omega(\gamma s + b), \omega \in \mathbb{R}_+^{N}, \gamma>0, b \in \RR\right\}$; these form a subset of~$\mathcal{F}_{\mathrm{ccv}}(\RR)$.
The associated optimization problem $P_{\mathrm{kink}}^{N}$ satisfies $P_{\mathrm{kink}}^{N} \le P_\mathrm{sub}^{\mathrm{ccv},1}$, where
\[
P_{\mathrm{kink}}^{N} = \sup_{\omega \in \mathbb{R}_+^{N},\gamma>0,a,b\in\mathbb{R}} \left\{ \mathbb{E}^{1}\left[\inf_{v\ge0}\left\{ v-\psi_\omega(\gamma(v^{2}+\Lambda_{a,b}(S_{1})))\right\} \right]+\mathbb{E}^{2}[\psi_\omega(\gamma\Lambda_{a,b}(S_{2}))]\right\} .
\]
By concavity of the square root, the above infimum is reached either at the kinks of $\psi_\omega$ or for $v=0$:
\begin{equation*}
u_{1}(s_1) = \inf_{v\ge0}\left\{ v-\psi_\omega(\gamma(v^{2}+\Lambda_{a,b}(s_{1})))\right\} = 
\min \left\{-\psi_\omega(\gamma \Lambda_{a,b}(s_{1})), \min_{i, x_i \geq \gamma\Lambda_{a,b}(s_{1})}\left\{\sqrt{\frac{x_i}{\gamma}-\Lambda_{a,b}(s_{1})} - \psi_\omega(x_i)\right\}\right\}.
\end{equation*}
The results we obtained for market and SABR risk-neutral densities all have $v^2 \mapsto \psi_\omega(\gamma(v^{2}+\Lambda_{a,b}(s_{1})))$ fit the VIX very closely for one value of $s_1$ (the tip of $\Lambda_{a,b}$), then slip to the left, as shown in Figure \ref{fig:kink_subrep}. Lower bounds are reported in Table \ref{tab:numerical_results}.

\begin{figure}
    \centering
    \begin{minipage}[b]{0.33\textwidth}
        \centering
        \includegraphics[width=6cm, height=4cm]{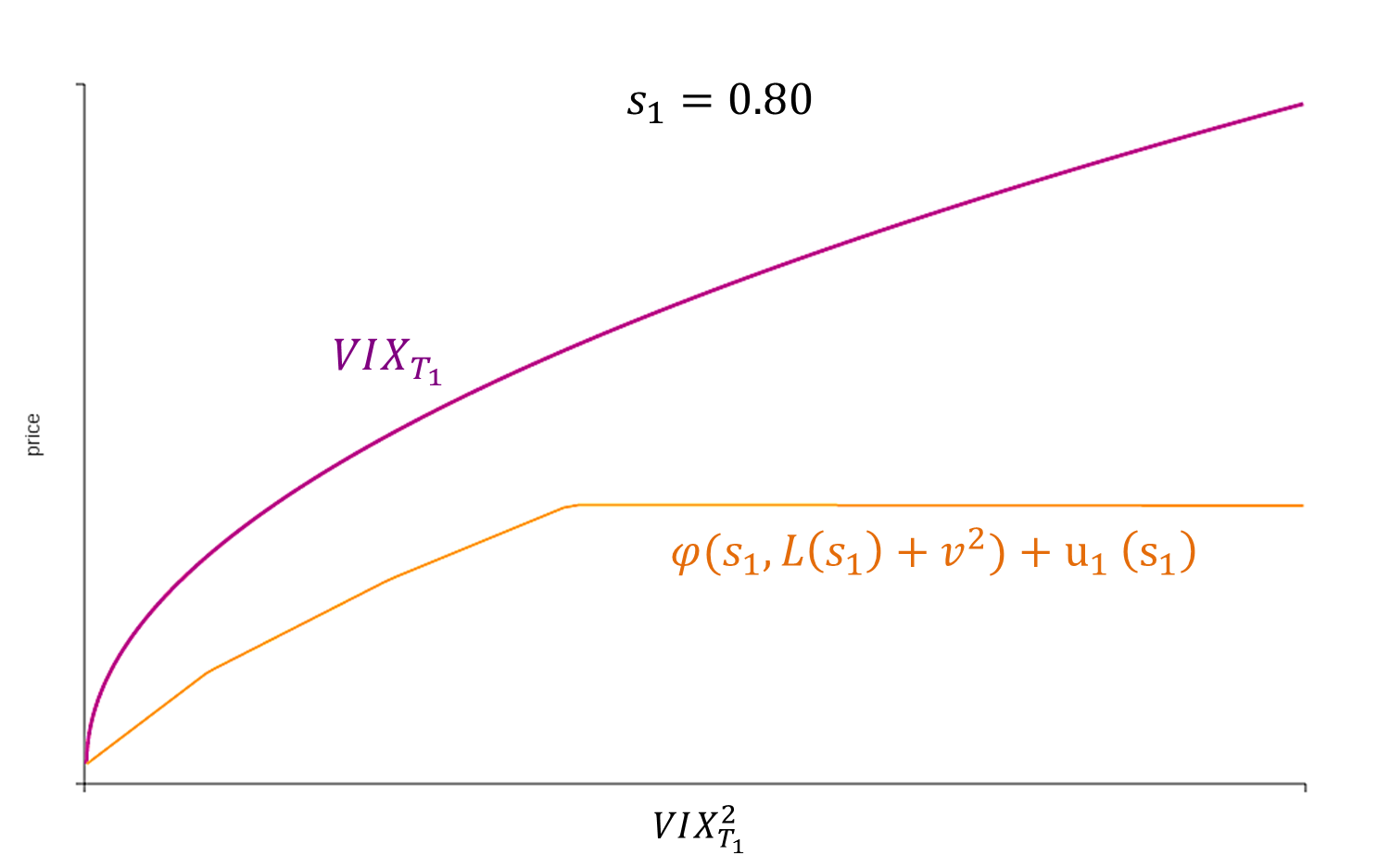}
    \end{minipage}%
    ~ 
    \begin{minipage}[b]{0.33\textwidth}
        \centering
        \includegraphics[width=6cm, height=4cm]{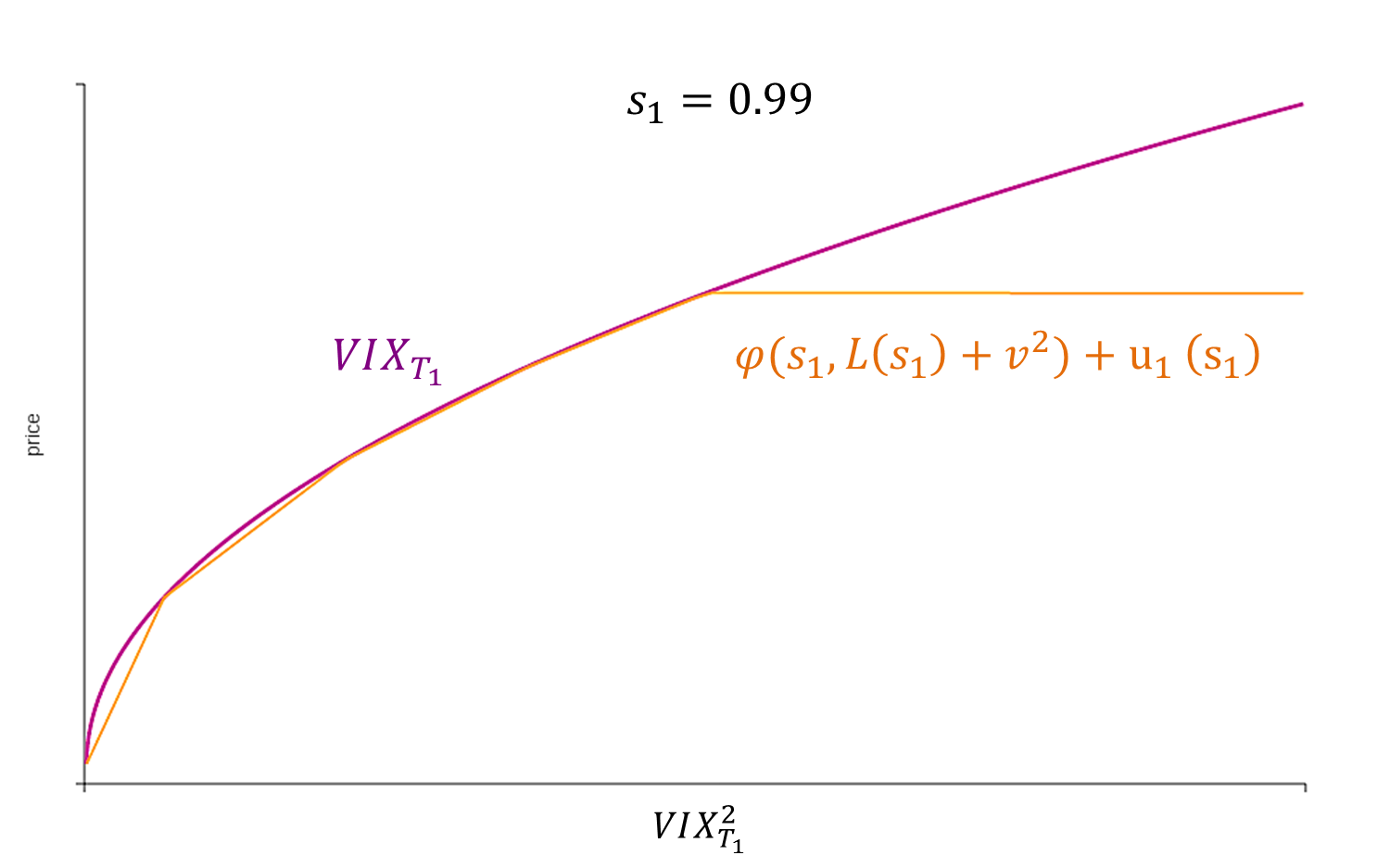}
    \end{minipage}%
    ~ 
    \begin{minipage}[b]{0.33\textwidth}
        \centering
        \includegraphics[width=6cm, height=4cm]{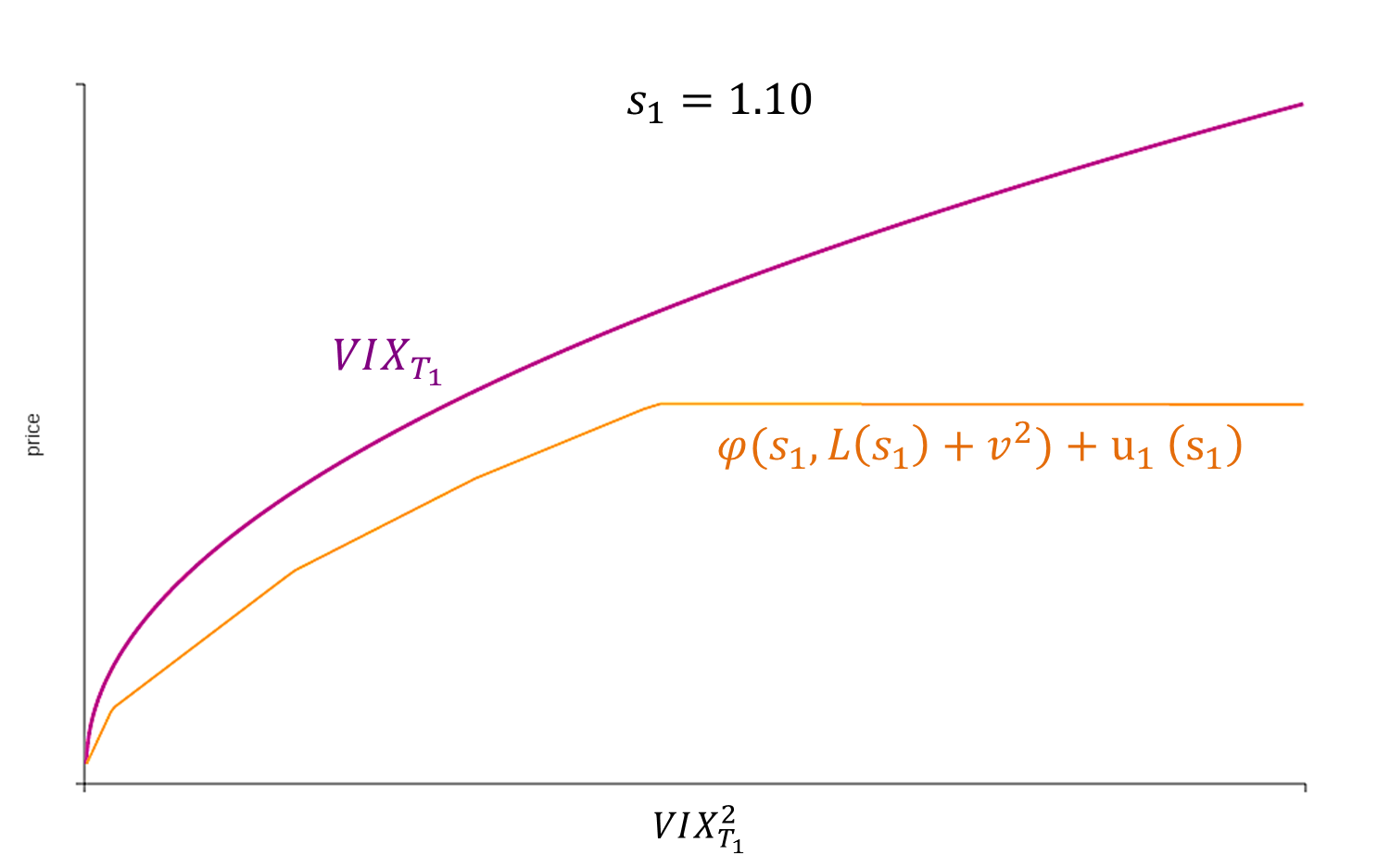}
    \end{minipage}
    \caption{Subreplication with a piecewise linear profile}
    \label{fig:kink_subrep}
\end{figure}

\subsubsection{Cut square root}

To obtain an even closer fit to the VIX, we consider the concave function
\[
\psi_\eps(x) = \begin{cases}
\frac{x}{\sqrt{\eps}} & \mathrm{if\,\,} x<\eps\\
\sqrt{x} & \mathrm{if\,\,} x \ge \eps.
\end{cases}
\]
for positive $\eps$. The corresponding $u_1$ is then
\begin{equation*}
u_{1}(s_1) = \inf_{v\ge0}\left\{ v-\psi_\eps(v^{2}+\Lambda_{a,b}(s_{1}))\right\} 
= -\psi_\eps(\Lambda_{a,b}(s_{1})_+) = 
\begin{cases}
-\psi_\eps(\Lambda_{a,b}(s_{1})) & \mathrm{if\,\,} \Lambda_{a,b}(s_{1}) \geq 0\\
0 & \mathrm{if\,\,} \Lambda_{a,b}(s_{1}) < 0
\end{cases}
\end{equation*}
which leads to the optimization problem
\[
P_{\mathrm{sqrt}} = \sup_{\eps>0, a, b \in \mathbb{R}} \left\{ -\mathbb{E}^{1}\left[\psi_\eps(\Lambda_{a,b}(S_{1})_+) \right]+\mathbb{E}^{2}[\psi_\eps(\Lambda_{a,b}(S_{2}))]\right\} 
\le P_\mathrm{sub}^{\mathrm{ccv},1}.
\]
Out of the concave functions we tested, this cut square root yielded the best results, that is, the highest lower bound for the price of the VIX future, as can be seen in Table \ref{tab:numerical_results}. The LP solver of Section~\ref{sec:LPsolver} gave a better lower bound for both the SABR smiles ($7.2\%$ versus $6.0\%$) and the market smiles as of May 5, 2016 (8.4\% versus 7.8\%), but the portfolio it yields is not guaranteed to subreplicate everywhere, as the subreplication constraint is only verified for a finite grid. By contrast, the functionally generated portfolios are subreplicating everywhere, by construction.  

\begin{figure}
    \centering
    \begin{minipage}[b]{0.33\textwidth}
        \centering
        \includegraphics[width=6cm, height=4cm]{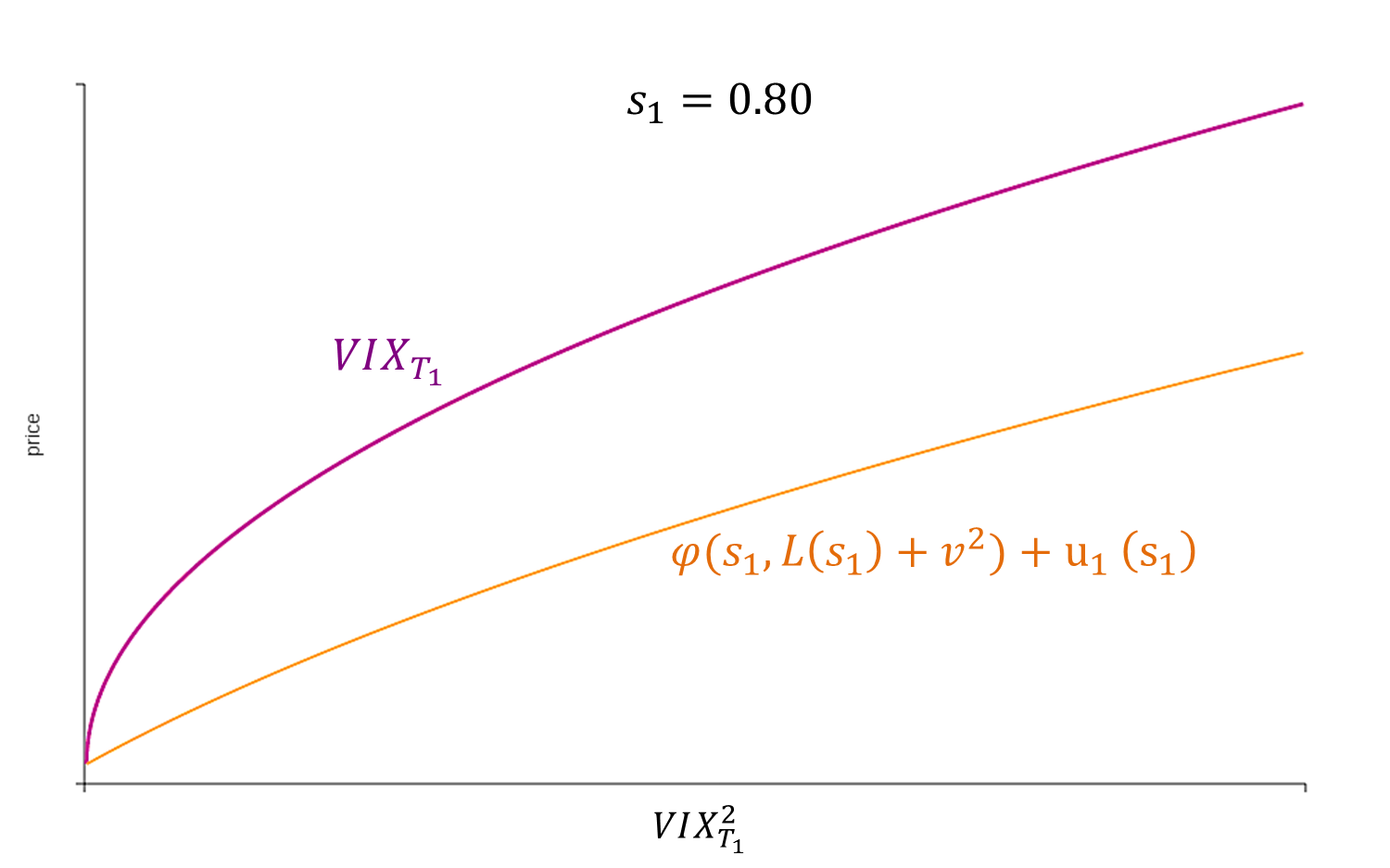}
     \end{minipage}%
    ~ 
    \begin{minipage}[b]{0.33\textwidth}
        \centering
        \includegraphics[width=6cm, height=4cm]{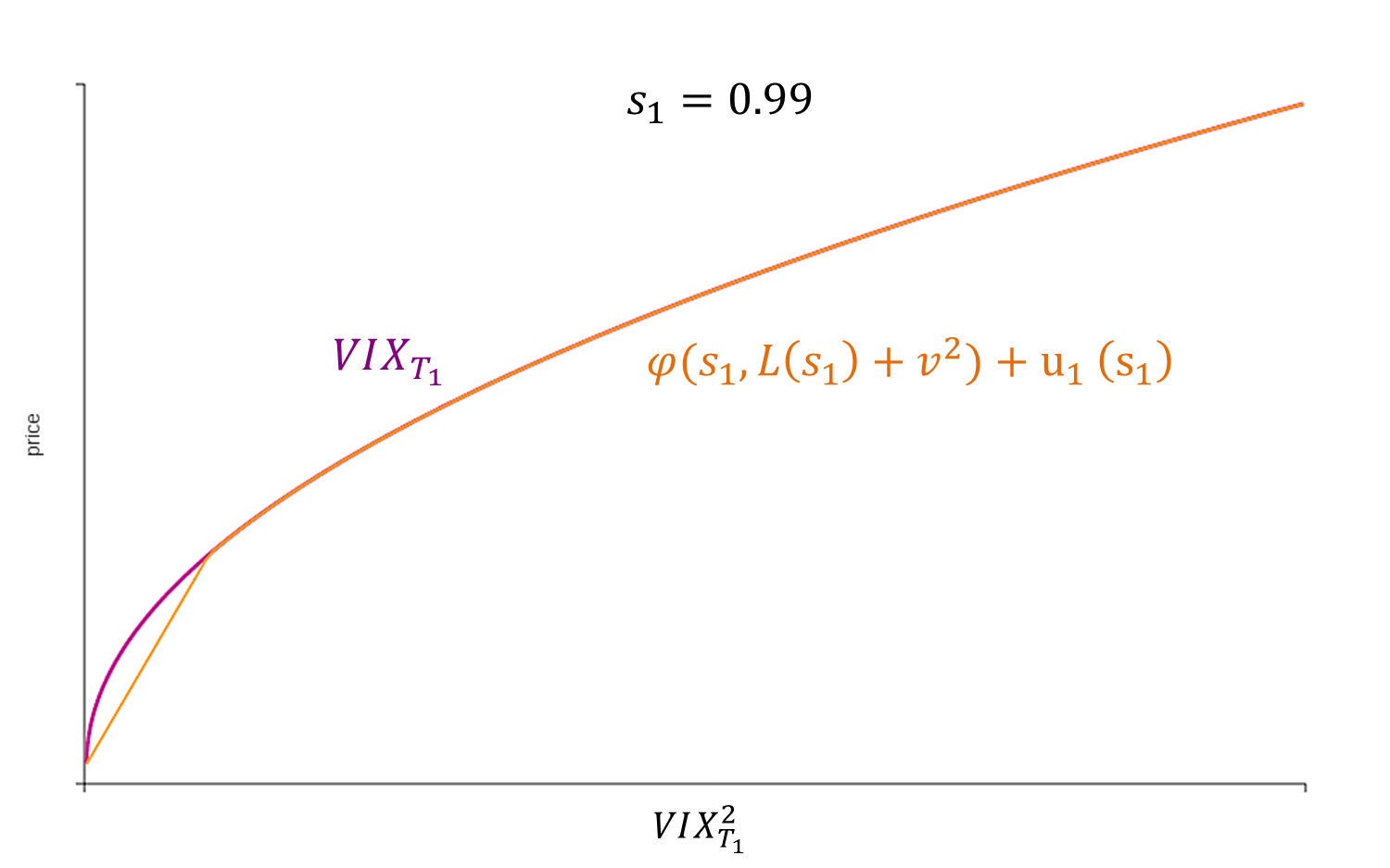}
     \end{minipage}%
    ~ 
    \begin{minipage}[b]{0.33\textwidth}
        \centering
        \includegraphics[width=6cm, height=4cm]{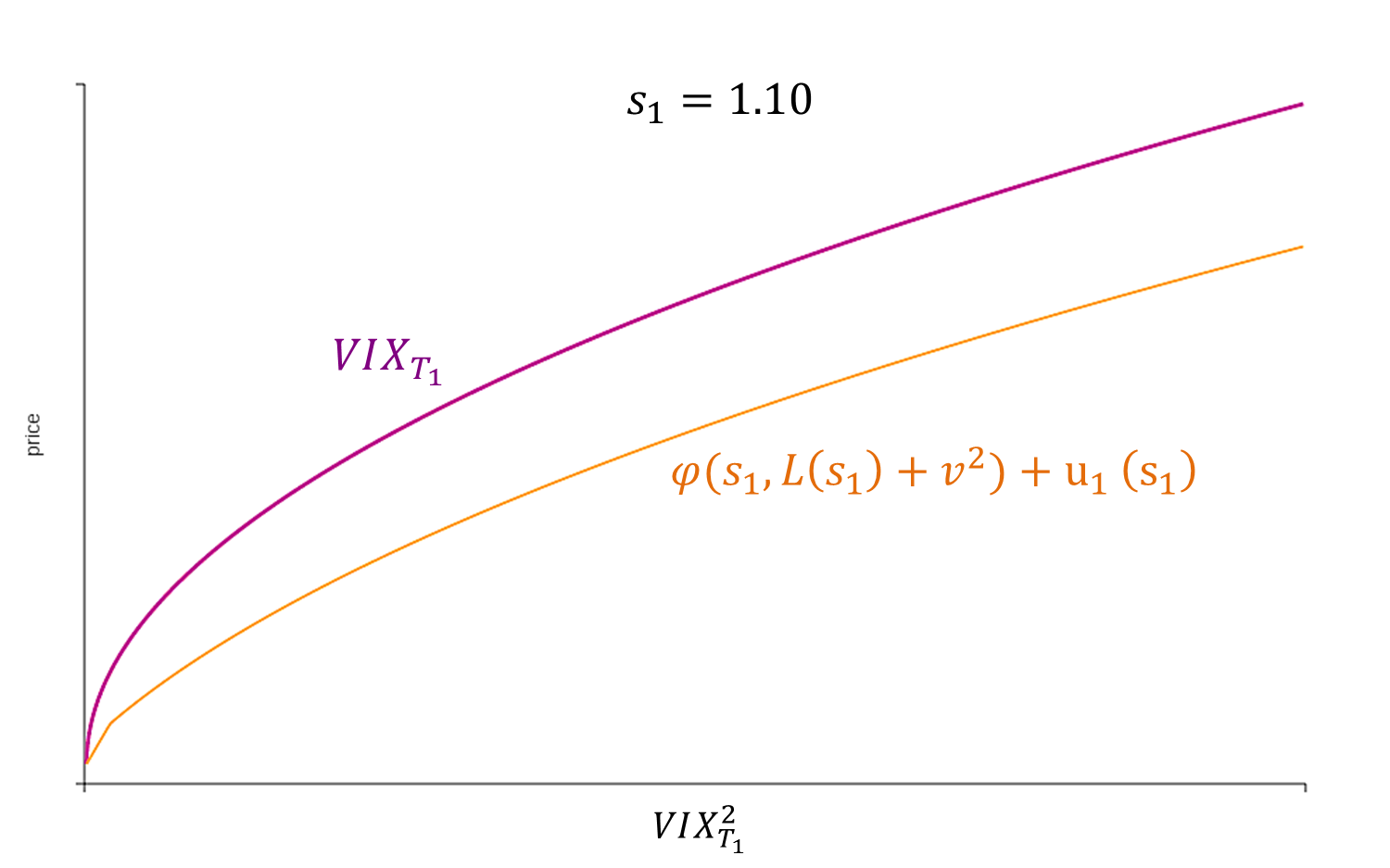}
    \end{minipage}
    \caption{Subreplication with a cut square root profile}
    \label{fig:sqrt_subrep}
\end{figure}

\begin{table}[]
\caption{Numerical results}
\centering
\begin{tabular}{|c|c|c|c|}
\hline
\multicolumn{2}{|c|}{} & \begin{tabular}[c]{@{}c@{}}SABR model (\ref{eq:SABRmodel}),\\ $T_1 = 2$ months\end{tabular} & \begin{tabular}[c]{@{}c@{}}Market smiles as of \\ May 5, 2016; $T_1 = 10$ days \end{tabular} \\ \hline
\multirow{ 5}{*}{Lower bound} & Classical lower bound & 0\% & 0\% \\
\cline{2-4}
& Piecewise linear profiles ($N=1$ kink) & 4.6\% & 4.4\% \\
\cline{2-4}
 & Piecewise linear profiles ($N=10$ kinks) & 5.2\% & 7.2\% \\
  \cline{2-4}
 & Cut square root & 6.0\% & 7.8\% \\
 \cline{2-4}
 & Lower bound from LP solver & 7.2\% & 8.4\% \\ \hline \hline
\multicolumn{2}{|c|}{Classical upper bound} & 22.8\% & 16.7\% \\ \hline \hline
\multicolumn{2}{|c|}{Upper bound from LP solver} & 22.8\% & 16.7\% \\ \hline
\end{tabular}
\label{tab:numerical_results}
\end{table}

\subsection*{Acknowledgments} We would like to thank Bruno Dupire for fruitful discussions, as well as Lorenzo Bergomi, Stefano De Marco, Pierre Henry-Labord\`ere, the Associate Editor, and two anonymous referees for their helpful comments on a preliminary version of this article. The research of M.\ Nutz is supported in part by an Alfred P.\ Sloan Fellowship and NSF Grant DMS-1512900.

%%%%%%%%%%%%%%%%%%%%%%%%%%%%%%%%%%%%%%%%%%%

\end{document}